\documentclass[11pt,letterpaper,onecolumn]{article}

\usepackage[T1]{fontenc}
\usepackage{lmodern}
\usepackage[letterpaper,margin=1in]{geometry}
\usepackage{microtype}
\usepackage{mathtools}
\usepackage{amssymb}
\usepackage{amsthm}
\usepackage{mathrsfs}
\usepackage{xcolor}
\definecolor{repositorylink}{HTML}{0366d6}
\usepackage{booktabs}
\usepackage{tabularx}
\usepackage{array}
\usepackage{enumitem}
\usepackage{flafter}
\IfFileExists{fontawesome5.sty}{%
  \usepackage{fontawesome5}%
  \newcommand{\GitHubIcon}{\mbox{{\small\faGithub}}\nobreakspace}%
}{%
  \newcommand{\GitHubIcon}{}%
}
\usepackage{graphicx}
\usepackage[font=normalsize,labelfont=bf,skip=6pt]{caption}
\usepackage{needspace}
\usepackage{colortbl}
\usepackage{tikz}
\usetikzlibrary{arrows.meta,positioning,shapes.geometric}
\usepackage{aliascnt}
\usepackage[hidelinks]{hyperref}
\usepackage[capitalise,nameinlink,noabbrev]{cleveref}
\usepackage{authblk}

\allowdisplaybreaks[2]

\theoremstyle{plain}
\newtheorem{theorem}{Theorem}[section]

\newaliascnt{lemma}{theorem}
\newtheorem{lemma}[lemma]{Lemma}
\aliascntresetthe{lemma}

\newaliascnt{proposition}{theorem}
\newtheorem{proposition}[proposition]{Proposition}
\aliascntresetthe{proposition}

\newaliascnt{corollary}{theorem}
\newtheorem{corollary}[corollary]{Corollary}
\aliascntresetthe{corollary}

\newaliascnt{claim}{theorem}

\aliascntresetthe{claim}

\theoremstyle{definition}
\newaliascnt{definition}{theorem}
\newtheorem{definition}[definition]{Definition}
\aliascntresetthe{definition}

\newaliascnt{procedure}{theorem}

\aliascntresetthe{procedure}

\theoremstyle{remark}
\newaliascnt{remark}{theorem}
\newtheorem{remark}[remark]{Remark}
\aliascntresetthe{remark}

\crefname{theorem}{Theorem}{Theorems}
\Crefname{theorem}{Theorem}{Theorems}
\crefname{lemma}{Lemma}{Lemmas}
\Crefname{lemma}{Lemma}{Lemmas}
\crefname{proposition}{Proposition}{Propositions}
\Crefname{proposition}{Proposition}{Propositions}
\crefname{corollary}{Corollary}{Corollaries}
\Crefname{corollary}{Corollary}{Corollaries}
\crefname{claim}{Claim}{Claims}
\Crefname{claim}{Claim}{Claims}
\crefname{definition}{Definition}{Definitions}
\Crefname{definition}{Definition}{Definitions}
\crefname{procedure}{Procedure}{Procedures}
\Crefname{procedure}{Procedure}{Procedures}
\crefname{remark}{Remark}{Remarks}
\Crefname{remark}{Remark}{Remarks}

\hypersetup{
  pdftitle={A Full Complexity Dichotomy for Complex-Valued Boolean Holant Problems},
  pdfauthor={Chenghua Liu, Boning Meng, Juqiu Wang}
}
\setlist[itemize]{leftmargin=1.6em,itemsep=0.25em,topsep=0.3em}
\setlist[enumerate]{leftmargin=1.8em,itemsep=0.3em,topsep=0.3em}

\newcommand{\AlgNums}{\overline{\mathbb Q}}
\newcommand{\Ftwo}{\mathbb F_2}
\newcommand{\cF}{\mathcal F}
\newcommand{\cG}{\mathcal G}

\newcommand{\cA}{\mathscr A}
\newcommand{\cP}{\mathscr P}
\newcommand{\cL}{\mathscr L}
\newcommand{\cM}{\mathscr M}
\newcommand{\cT}{\mathscr T}
\newcommand{\cC}{\mathscr C}

\newcommand{\angles}[1]{\left\langle #1\right\rangle}
\newcommand{\Holant}{\operatorname{Holant}}
\newcommand{\KHolant}{\operatorname{K\text{-}Holant}}
\newcommand{\EOTract}{\mathsf{EOTract}}
\newcommand{\Tract}{\mathsf{Tract}}
\newcommand{\FP}{\mathsf{FP}}
\newcommand{\SharpP}{\#\mathsf P}
\newcommand{\PGL}{\operatorname{PGL}}
\newcommand{\GL}{\operatorname{GL}}
\newcommand{\Mat}{\operatorname{Mat}}
\newcommand{\diag}{\operatorname{diag}}
\newcommand{\Span}{\operatorname{span}}

\newcommand{\wt}{\operatorname{wt}}
\newcommand{\supp}{\operatorname{supp}}
\newcommand{\arity}{\operatorname{arity}}

\newcommand{\EO}{\mathrm{EO}}
\newcommand{\GE}{\mathrm{GE}}
\newcommand{\CP}{\mathrm{CP}}
\newcommand{\Eq}{=}
\newcommand{\Neq}{\neq}

\newcommand{\one}{\mathbf 1}
\newcommand{\zero}{\mathbf 0}

\newcommand{\DeltaZero}{\Delta_0}
\newcommand{\DeltaOne}{\Delta_1}
\newcommand{\RMcore}{\mathrm{RM}_8}
\newcommand{\leT}{\leq_{\mathrm T}}
\newcommand{\equivT}{\equiv_{\mathrm T}}

\newcolumntype{L}[1]{>{\raggedright\arraybackslash}p{#1}}

\definecolor{pinkfoldink}{HTML}{B74B7D}
\newcommand{\pinkfoldmark}{%
  \raisebox{-0.15em}{%
    \begin{tikzpicture}[x=0.0375em,y=-0.0375em]
      \path[use as bounding box] (0,0) rectangle (24,24);
      \draw[pinkfoldink,line width=0.06375em,line join=miter,line cap=butt]
        (12,2) -- (20.660,7) -- (20.660,17) -- (12,22) --
        (3.340,17) -- (3.340,7) -- cycle
        (12,12) -- (20.660,7)
        (12,12) -- (3.340,7)
        (12,12) -- (12,22);
    \end{tikzpicture}%
  }%
}

\begin{document}

\title{A Full Complexity Dichotomy for Complex-Valued\\[0.15em]
		Boolean Holant Problems}
	\author[1,*]{Chenghua Liu}
	\author[2,\textdagger]{Boning Meng}
	\author[1,\textdaggerdbl]{Juqiu Wang}
	
	\affil[1]{Institute of Software, Chinese Academy of Sciences, Beijing, China}
	\affil[2]{University of Regensburg, Regensburg, Germany}
	
	\affil[ ]{%
		\textsuperscript{*}\texttt{liuch.russell@gmail.com}\quad
		\textsuperscript{\textdagger}\texttt{mengboning2013@gmail.com}\quad
		\textsuperscript{\textdaggerdbl}\texttt{wangjuqiu104@gmail.com}%
	}
	
	\date{}

\hypersetup{pageanchor=false}
\begin{titlepage}
\maketitle
\thispagestyle{empty}

	\begin{abstract}\normalsize
		A Boolean Holant problem evaluates a graph partition function specified
        by local signatures at its vertices.  We prove a complexity dichotomy
        for arbitrary finite sets of algebraic complex-valued signatures on
        general graphs: an explicit, decidable criterion characterizes the
        polynomial-time cases, and every remaining case is \(\SharpP\)-hard
        under polynomial-time Turing reductions.  The proof resolves the
        all-even case by constructing signatures with nonzero endpoints and
        all remaining support on the central layer, then classifying them
        according to half their arity. Binary transfer groups and matching-product
        reconstruction provide the main tools in the small- and large-height cases,
        respectively.
	\end{abstract}
\end{titlepage}

\clearpage
\pagenumbering{arabic}
\hypersetup{pageanchor=true}

\section{Introduction}
\label{short:sec:introduction}

Counting problems ask how many global configurations satisfy a collection
of local constraints. Counting proper graph colourings, for example,
sums over vertex colours subject to a constraint on each edge. In a
weighted problem, each configuration contributes the product of its local
weights, and their sum is a partition function. Although this sum has
exponentially many terms, some choices of local rules permit efficient
evaluation. A complexity dichotomy identifies all such choices and proves
hardness for every remaining choice.

A Boolean Holant problem places local functions, called signatures, at
the vertices of a graph. Each edge receives a bit, and the partition
function sums the product of the vertex signatures over all edge
assignments. We prove a dichotomy for arbitrary finite sets of algebraic
complex-valued Boolean signatures on general graphs, with an explicit,
decidable tractability criterion. The main case addressed here consists
entirely of even-arity signatures, with neither a symmetry restriction nor
a prescribed auxiliary binary signature.

\subsection{From graph homomorphisms through counting CSP to Holant}
\label{short:subsec:intro-history}

Counting graph homomorphisms (\(\#\mathrm{GH}\)) asks how many maps from an
input graph to a fixed target graph preserve adjacency.  In the weighted
version, a fixed matrix specifies the interaction between the labels at
the two ends of each edge.  Counting constraint satisfaction problems
(\(\#\mathrm{CSP}\)) allow constraints of arbitrary arity, so one local
function can constrain several variables at once.  The weighted problem
sums the product of the constraint values over all variable assignments.

Dichotomies for unweighted graph homomorphisms and Boolean counting CSP
\cite{DyerGreenhill2000,CreignouHermann1996} established early classifications
in these frameworks.  Subsequent work extended the allowed domains and
weights, from nonnegative weights to signed-real and complex weights,
where cancellation must also be accounted for.  The general-domain complex
classification of Cai and Chen covers both weighted graph homomorphisms
and counting CSP \cite{CaiChen2017}.
\Cref{short:tab:intro-gh-csp-history} records this development.  A dash means that
the case is covered by a more general entry, with no separate reference
listed.

	\begin{table}[!htbp]
		\centering
		\caption{Development of general-graph dichotomies for weighted graph
			homomorphisms and counting CSP.}
		\label{short:tab:intro-gh-csp-history}
		\normalsize
		\setlength{\tabcolsep}{3.5pt}
		\renewcommand{\arraystretch}{1.16}
		\begin{tabularx}{\textwidth}{@{}l
				>{\raggedright\arraybackslash}X
				>{\raggedright\arraybackslash}X
				>{\raggedright\arraybackslash}X
				>{\raggedright\arraybackslash}X@{}}
			\toprule
			Weight range & Symmetric \(\#\mathrm{GH}\) & \(\#\mathrm{GH}\) &
			Boolean-domain \(\#\mathrm{CSP}\) &
			General-domain \(\#\mathrm{CSP}\) \\
			\midrule
			\(\{0,1\}\)
			& \cite{DyerGreenhill2000}
			& acyclic targets \cite{DyerGoldbergPaterson2007}
			& \cite{CreignouHermann1996}
			& \cite{Bulatov2013,DyerRicherby2013} \\
			\(\mathbb Q_{\ge0}\)
			& -- & -- & --
			& \cite{BulatovDyerGoldbergJalseniusJerrumRicherby2012} \\
			\(\mathbb R_{\ge0}\)
			& \cite{BulatovGrohe2005}
			& \cite{CaiChen2019}
			& \cite{DyerGoldbergJerrum2009}
			& \cite{CaiChenLu2016} \\
			\(\mathbb R\)
			& \cite{GoldbergGroheJerrumThurley2010}
			& --
			& \cite{BulatovDyerGoldbergJalseniusRicherby2009}
			& -- \\
			\(\mathbb C\)
			& \cite{CaiChenLu2013}
			& \cite{CaiChen2017}
			& \cite{CaiLuXia2014}
			& classification \cite{CaiChen2017},
			decidability \cite{LiuMeng2026CSPDecidability} \\
			\bottomrule
		\end{tabularx}
	\end{table}

Holant changes which local constraints are available.  A counting-CSP
instance can be represented by a bipartite graph with variable vertices
and constraint vertices.  At each variable vertex, an equality signature
forces all occurrences of that variable to agree.  In a Holant problem,
these equality signatures need not be available: every vertex must use a
signature from the prescribed set. The signatures needed in a reduction
must be obtained by gadgets or other justified reductions. This distinction brings
problems such as counting perfect matchings into the classification and
changes the constructions available for proving hardness.

Motivated by Valiant's holographic algorithms \cite{Valiant2008}, the
Holant framework was introduced in 2009 \cite{CaiLuXia2009}.
A holographic transformation is an invertible simultaneous change of basis
on the signatures and edge tensors that preserves their contraction.
Such transformations relate apparently different counting problems and
expose tractable signature sets.

On the Boolean domain, Holant classifications developed along two
complementary directions: relaxing restrictions on symmetry and weights,
and reducing the auxiliary signatures assumed available.  Unary
signatures are useful because they allow an input of a signature to be
fixed or summed with a chosen weight.  The auxiliary assumptions range
from all unary signatures (\(\Holant^*\)), through four distinguished
unaries (\(\Holant^+\)), to the two unaries fixing a bit to 0 or 1
(\(\Holant^c\)).  Later work treated any signature set containing a
nonzero odd-arity signature (\(\Holant^{\mathrm{odd}}\)).
\Cref{short:tab:intro-boolean-holant-history} separates these assumptions from
restrictions on symmetry and weights, using the same dash convention as
\cref{short:tab:intro-gh-csp-history}.

	\begin{table}[!htbp]
		\centering
		\caption{Development of Boolean-domain Holant dichotomies on general
			graphs.}
		\label{short:tab:intro-boolean-holant-history}
		\normalsize
		\setlength{\tabcolsep}{3pt}
		\renewcommand{\arraystretch}{1.18}
		\begin{tabularx}{\textwidth}{@{}l
				>{\raggedright\arraybackslash}X
				>{\raggedright\arraybackslash}X
				>{\raggedright\arraybackslash}X@{}}
			\multicolumn{4}{@{}l}{\textit{(a) Symmetric signature sets}}\\[0.2em]
			\toprule
			Weight range & \(\Holant^*\) & \(\Holant^c\) & Unrestricted \(\Holant\) \\
			\midrule
			\(\mathbb R_{\ge0}\) & -- & -- & -- \\
			\(\mathbb R\) & -- & \cite{CaiLuXia2009}
			& \cite{HuangLu2016} \\
			\(\mathbb C\) & \cite{CaiLuXia2009}
			& \cite{CaiHuangLu2012}
			& \cite{CaiGuoWilliams2016} \\
			\bottomrule
		\end{tabularx}
		
		\medskip
		\begin{tabularx}{\textwidth}{@{}l
				>{\raggedright\arraybackslash}X
				>{\raggedright\arraybackslash}X
				>{\raggedright\arraybackslash}X
				>{\raggedright\arraybackslash}X
				>{\raggedright\arraybackslash}X@{}}
			\multicolumn{6}{@{}l}{\textit{(b) Arbitrary signature sets}}\\[0.2em]
			\toprule
			Weight range & \(\Holant^*\) & \(\Holant^+\) & \(\Holant^c\) &
			\(\Holant^{\mathrm{odd}}\) & Unrestricted \(\Holant\) \\
			\midrule
			\(\mathbb R_{\ge0}\) & -- & -- & -- & --
			& \cite{LinWang2018} \\
			\(\mathbb R\) & -- & --
			& \cite{CaiLuXia2018RealHolantc}
			& \cite{CaiFuShao2020Entanglement}
			& \cite{ShaoCai2020} \\
			\(\mathbb C\)
			& \cite{CaiLuXia2011}
			& \cite{Backens2017}
			& \cite{Backens2021,Backens2025}
			& \cite{MengWangXiaZheng2025,GuanShaoShi2026}
			& This paper \\
			\bottomrule
		\end{tabularx}
	\end{table}

Without auxiliary unary signatures, dichotomies were obtained for
symmetric complex-valued signatures \cite{CaiGuoWilliams2016}, arbitrary
nonnegative signatures \cite{LinWang2018}, and arbitrary real-valued
signatures \cite{ShaoCai2020}.

In 2025, Meng, Wang, and Xia proved an \(\FP^{\mathsf{NP}}\) versus
\(\SharpP\)-hard dichotomy for weighted Eulerian orientations
\cite{MengWangXia2025EO}.  Meng, Wang, Xia, and Zheng combined this result
with tensor decomposition to obtain the same complexity bounds for complex
Boolean Holant with a nonzero odd-arity signature
\cite{MengWangXiaZheng2025}.  In 2026, Guan, Shao, and Shi gave a
polynomial-time algorithm based on linear programming for the former
\(\FP^{\mathsf{NP}}\) cases, yielding \(\FP\) versus \(\SharpP\)-hard
dichotomies for both problems \cite{GuanShaoShi2026}.

These results direct attention to the all-even case.  Every gadget built
from even-arity signatures has even arity, since each internal edge consumes
two inputs.  Thus no such gadget supplies a unary signature.  Binary
signatures become the smallest positive-arity objects available for further
constructions.  Xia proposed a framework organized by the finite groups
associated with binary signatures, analyzing several branches in which all
realizable nonzero quaternary signatures decompose \cite{Xia2026Framework}.
Liu and Meng then classified complex Boolean Holant with binary
disequality available, for arbitrary finite sets of algebraic
complex-valued signatures, and gave an explicit decidable tractability
criterion \cite{Liu2026DEQ2Holant}.

Our proof removes the prescribed-binary assumption in the all-even case.
It combines the binary and quaternary dichotomies of Liu and Meng with a
decomposition based on Hamming imbalance.  This decomposition produces
signatures of controlled support whose arity determines the remaining
cases, as described in \cref{short:sec:mountain,short:sec:height-one,short:sec:higher}.  Together with the
odd-arity dichotomy, it completes the unrestricted complex-valued entry in
\cref{short:tab:intro-boolean-holant-history}.

\paragraph{Independent work.}
Huang and Fu established a dichotomy for Holant problems with a single
complex-valued quaternary signature \cite{HuangFu2026FourRegular}.
Guan, Shao, and Tang proved a dichotomy for complex-valued Boolean Holant
problems whose signature sets are closed under complex conjugation
\cite{GuanShaoTang2026ConjugateClosed}.
Zhou proved a dichotomy for finite sets of algebraic complex-valued Boolean
signatures of even arity whose nonzero binary gadget signatures are
nonsingular and whose projective matrix classes form the cyclic group
\(C_2\) of order two under series composition \cite{ZhouC2Holant}.

We were also informed of the following independent work before the
public release of the present paper.
Jiayi Zheng informed us of unpublished work on the complex-valued Holant
dichotomy in the case where the projective binary transfer classes form an
odd dihedral group.
Mingji Xia informed us of work concerning the twisted Klein-group case
and the \(C_1\) case, under the condition that all quaternary signatures
are reducible.
The results and proofs in the present paper were developed independently
of these works.

\subsection{Our result: the full dichotomy}
\label{short:subsec:main}

In their odd-arity dichotomy, Meng, Wang, Xia, and Zheng
\cite[conference version, Theorem~30; full version, Theorem~33]{MengWangXiaZheng2025}
formulate condition~(PC) as seven structural alternatives for signature
sets. These specify support restrictions, tensor decompositions, or
membership in prescribed function classes under a common change of basis.
We write \(\Tract(\cF)\) when \(\cF\) satisfies at least one of these
seven alternatives, using the equivalent formulation in
\cref{sec:predicate}. Each alternative concerns
the whole set: in a transformability case, one basis must work
simultaneously for the edge and every member of \(\cF\).

The algorithms for these seven forms, together with Guan, Shao, and Shi's
\(\FP\) improvement for the cases previously known to lie in
\(\FP^{\mathsf{NP}}\) \cite{GuanShaoShi2026}, give polynomial-time
evaluation whenever this structural condition holds. Our contribution is
to prove hardness outside these seven forms without the odd-arity
assumption.

\begin{theorem}[Full dichotomy]
\label{short:thm:main}
Let \(\cF\) be a finite set of algebraic complex-valued Boolean signatures,
given by exact algebraic numbers. The predicate \(\Tract(\cF)\) is
decidable. If it holds, then \(\Holant(\cF)\in\FP\). Otherwise,
\(\Holant(\cF)\) is \(\SharpP\)-hard under polynomial-time Turing reductions.
\end{theorem}

For evaluation and hardness reductions, \(\cF\) is fixed and running time
is measured in the encoded instance size. For recognition, \(\cF\) itself
is the input. Exact algebraic data may be represented by minimal polynomials
and isolating data for the chosen roots. Decidability of the criterion and
evaluation in its tractable cases are supplied by
\cite[Theorem~2.22]{Liu2026DEQ2Holant}. We prove the remaining hardness
direction. After the known cases are removed, a construction based on
Hamming imbalance produces a signature with two nonzero endpoints and all
other support on its central layer. We call this signature a mountain and
classify the remaining problem according to its height.

The full paper follows the shared references. References with prefix F
point to its detailed statements and proofs.

\section{Known classifications and reductions}
\label{short:sec:known}

We first specify when an existing result completes a branch of the proof.
A signature of arity \(r\) is a function
\(f:\{0,1\}^r\to\AlgNums\), where \(\AlgNums\) is the field of algebraic
complex numbers. Its arguments follow the fixed order of the incident
edges. For a graph \(G=(V,E)\) labelled by signatures \(f_v\), its Holant is
\(\sum_{\sigma:E\to\{0,1\}}\prod_{v\in V}f_v(\sigma|_{E(v)})\).
We allow parallel edges and self-loops. Zero signatures and nullary
constants are handled directly, so subsequent sets contain nonzero
positive-arity signatures.

For counting problems \(A,B\), write \(A\leT B\) if a polynomial-time
algorithm with oracle access to \(B\) computes \(A\), and write
\(A\equivT B\) when both directions hold.
Put \(X=\left(\begin{smallmatrix}0&1\\1&0\end{smallmatrix}\right)\) and
\(K=2^{-1/2}\left(\begin{smallmatrix}1&1\\i&-i\end{smallmatrix}\right)\).
For a matrix \(D\), write \(Df=D^{\otimes r}f\) and
\(D\cF=\{Df:f\in\cF\}\). Since \(K^{\mathsf T}K=X\), the holographic
transformation gives
\[
\Holant(\cF)\equivT\Holant(X\mid K^{-1}\cF)
=\KHolant(K^{-1}\cF).
\]
Here each edge joins opposite bits at its two ends. From now on, \(\cF\)
denotes the current set in these \(K\)-coordinates. The structural
condition is therefore \(\Tract(K\cF)\).
A gadget is a fixed network with ordered dangling edges, whose signature
is obtained by summing over the internal bits. A directly realized
signature can be substituted at every occurrence. More generally, a finite
signature set \(S\) can be \emph{simulated} from \(\cF\) if
\[
\KHolant(\cF\cup S)\leT\KHolant(\cF).
\]
Simulation may use several oracle queries, as in interpolation, without
a gadget realizing each signature of \(S\).
Known nonzero scalar factors are retained in these reductions.

\paragraph{Odd arity and Hamming layers.}
If \(\cF\) contains a nonzero odd-arity signature, the dichotomy of
\cite{MengWangXiaZheng2025}, together with \cite{GuanShaoShi2026}, completes
the classification. In particular, this applies whenever a reduction
supplies a nonzero unary. We may therefore consider only even arities.
For an arity-\(r\) input \(\mathbf x\), define its \emph{half-charge} by
\(\chi_r(\mathbf x)=\wt(\mathbf x)-r/2\), where \(\wt\) is Hamming weight.
Every internal \(X\)-edge removes one zero and one one, so half-charge is
conserved: in a gadget it equals the boundary half-charge, and in a closed
instance the total is zero. An \emph{EO signature} is supported on the
central layer \(\wt(\mathbf x)=r/2\).

If all supported half-charges in the entire set are nonnegative, or all
are nonpositive, only central entries contribute to a closed instance.
The EO dichotomy \cite{MengWangXia2025EO,GuanShaoShi2026} then classifies
the problem. The case in which each signature is supported on one Hamming
layer is also classified \cite[Corollary~68]{MengWangXia2025EO}.
Its reduction appends fixed zero or one bits to balance each signature's
charge. In each charge-balanced component these bits can be paired by
\(X\)-edges, with a unique extension of every original assignment.
The remaining set is \emph{two-sided}, meaning both signs occur among
its supported half-charges, and \emph{multilayer}, meaning at least one
member has support on different Hamming layers.

\paragraph{Factors, equality, and four variables.}
A tensor product uses disjoint variable blocks. A nonzero positive-arity
signature is \emph{irreducible} if it cannot be written as a tensor product
of two positive-arity signatures, after a permutation of variables.
The tensor-factor extraction theorem
\cite[full version, Theorem~32 and Remark~39]{MengWangXiaZheng2025}
allows the factors of a nonzero tensor product to be simulated,
outside the common-one-sided case just handled. Initial factor
extraction preserves two-sidedness and the multilayer property. An odd
factor invokes the odd-arity result. If all factors have arity at most
two, tensor decomposition reduces evaluation to paths and cycles.
Otherwise we start with irreducible even-arity signatures, at least one
of arity greater than two.

Two further results apply to the whole current signature set. If binary
equality \(I\) can be simulated, the dichotomy of
\cite[Theorem~1.2]{Liu2026DEQ2Holant} applies. The same is true for a
nonsingular weighted equality \(\diag(p,q)\): a common diagonal change of
basis makes it \(I\) and preserves \(X\) up to a nonzero scalar.
An \emph{eight-vertex signature} is a quaternary signature supported on
even Hamming weights. If an irreducible such signature has both endpoint
values nonzero, \cite[Theorem~1.3]{Liu2026DEQ2Holant} applies. A reducible
one with these endpoint values has a nonsingular weighted-equality
factor, so that case is classified as well. These consequences are stated in \cref{thm:eq2-kholant-dichotomy} and
\cref{cor:eight-vertex-complete-closure}.

In the remaining hardness proof we assume \(\neg\Tract(K\cF)\).
This assumption persists when the set is enlarged by simulated signatures,
since the seven structural alternatives are inherited by subsets. It also persists under
common changes of basis in
\(GO(X)=\{D\in\GL_2(\AlgNums):D^{\mathsf T}XD=\lambda X,
\ \lambda\ne0\}\). Thus each applicable classification above gives
hardness, which transfers back along the reductions. Tensor reconstruction
transfers any tractability conclusion for the initial factor set back to
the original set (\cref{sec:assembly}).

\section{Obtaining a mountain}
\label{short:sec:mountain}

The remaining signatures may have arbitrary arity and complex
coefficients. We seek a signature with two nonzero endpoint values whose
other support lies on its central Hamming layer. Such a signature retains
both signs of half-charge, while its contractions have much more
restricted support. For \(d\ge1\), a \emph{mountain of height \(d\)} is an
arity-\(2d\) signature \(M\) satisfying
\[
\supp(M)\subseteq HW_{2d}^{=}\cup\{\zero^{2d},\one^{2d}\},
\qquad M(\zero^{2d})M(\one^{2d})\ne0,
\]
where \(HW_{2d}^{=}\) is the set of weight-\(d\) strings.

\begin{lemma}[Mountain construction]
\label{short:lem:mountain}
Let \(\cF\) be a finite two-sided set of nonzero even-arity signatures, and
let \(d\) be the gcd of all its nonzero supported half-charges. Then either
a mountain \(M\) of height \(d\) can be simulated, or both
pinning signatures \(\Delta_0=(1,0)\) and \(\Delta_1=(0,1)\) can be simulated.
In the mountain case, every supported half-charge in
\(\cF\cup\{M\}\) belongs to \(d\mathbb Z\).
\end{lemma}

\begin{proof}[Proof outline]
Let \(C\) be the finite set of nonzero supported half-charges. Tensor
products add half-charges, so the first task is to express both \(d\) and
\(-d\) as sums of elements of \(C\) with nonnegative multiplicities.
The gcd is an integer linear combination
\(\sum_{c\in C}u_cc=d\), but some coefficients may be negative.
For every \(a,-b\in C\), with \(a,b>0\), the positive multiplicities
\(b/\gcd(a,b)\) and \(a/\gcd(a,b)\) give total charge zero.
Adding these zero-sum combinations over all positive--negative pairs
gives \(\mathbf v\in\mathbb Z_{>0}^C\) with \(\sum_cv_cc=0\).
Consequently \(\mathbf u+k\mathbf v\) and \(-\mathbf u+k\mathbf v\)
are nonnegative for sufficiently large \(k\), while their total charges
are \(d\) and \(-d\).

Choose a supported input witnessing each charge used in these two tensor
products. The endpoint-preserving self-loop lemma
\cite[full version, Lemma~40]{MengWangXiaZheng2025} gives two
arity-\(2d\) gadget signatures \(P,N\) with
\(P(\one^{2d})N(\zero^{2d})\ne0\). Charge conservation keeps every
supported charge divisible by \(d\). At this arity only \(-d,0,d\)
are possible, so both signatures already have the required support.
It remains to put the two nonzero endpoints into the same signature.

If either signature already has both endpoints, we are done. Otherwise
\(P(\zero^{2d})=N(\one^{2d})=0\). If both have nonzero central
restrictions, permute supported central strings to
\((\zero^d,\one^d)\) in \(P\) and \((\one^d,\zero^d)\) in \(N\).
Join their last \(d\) variables by \(X\)-edges and expose the first
\(d\) variables of each. The resulting signature \(Q\) satisfies
\[
\begin{split}
Q(\zero^{2d})&=P(\zero^d,\one^d)N(\zero^{2d}),\\
Q(\one^{2d})&=P(\one^{2d})N(\one^d,\zero^d).
\end{split}
\]
Each endpoint has a unique nonzero extension, so cancellation is
excluded. If just one central restriction is nonzero, use two copies of
that signature and the pure endpoint of the other to force the connected
bits. The same unique-extension argument gives both nonzero endpoints.
Charge conservation retains the support restriction.

If both central restrictions are zero, \(P\) and \(N\) are nonzero
scalar multiples of tensor powers of \(\Delta_1\) and \(\Delta_0\).
Factor extraction supplies both pinnings. The original two-sided set
remains present, so the one-sided exception is excluded. The full
construction is \cref{lem:charge-bridge}.
\end{proof}

The pinning outcome is classified by \cref{short:sec:known}. We therefore
retain the mountain and the whole-set divisibility condition.
At height one the mountain is binary, so the available binary operations
control the next steps. At larger heights, contractions below arity
\(2d\) have half-charge zero and hence are EO. This gives two different
ways to constrain the remaining signatures.

\section{Height one: binary transfers and contractions}
\label{short:sec:height-one}

A height-one mountain is a binary
\(B=\left(\begin{smallmatrix}a&b\\c&e\end{smallmatrix}\right)\) with
\(ae\ne0\). If a nonzero singular binary can be realized, its unary
factors invoke the odd-arity classification. Binary interpolation gives
the same conclusion when a nonsingular binary transfer has infinite
projective order (\cref{lem:binary-interpolation}).
Here the \emph{transfer} of a binary \(B\) is \(BX\): series composition
satisfies \((BXC)X=(BX)(CX)\). A projective class \([T]\) identifies
matrices differing by a nonzero scalar, and its order is the least
\(m\ge1\) with \(T^m\) scalar, if such an \(m\) exists.

In the remaining case the \emph{complete transfer group} is
\[
G(\cF)=\{[BX]: B\ne0\text{ is an ordered binary directly realizable over }\cF\}.
\]
All these binaries are nonsingular and their transfers have finite order.
The algebraicity of the fixed signature set makes \(G(\cF)\) finite
\cite[Lemma~5.2]{Liu2026DEQ2Holant}. The group analysis in \cref{thm:finite-group-split} gives either a nonsingular
weighted equality, which completes the classification, or a common
\(GO(X)\) change of basis after which the group is odd cyclic or odd
dihedral. We now explain how these two remaining group forms constrain
signatures of higher arity.

\subsection{From restricted contractions to tensor factors}

Fix one directly realized binary representative of each transfer class.
Contract two variables of any current signature with such a
representative, and extract the factors of the nonzero output. A new
binary may enlarge the complete transfer group. A new nonbinary
irreducible factor may lower the minimum such arity. Both changes supply
new information and restart the analysis with an equivalent set.
An odd factor or a binary covered by \cref{short:sec:known} completes
the classification immediately.
If neither occurs, recursive contractions give a finite set
\(\cF^{\mathrm{cl}}\): every indicated contraction has been processed,
and all its irreducible factors are retained. This closure is finite
because a new nonbinary factor has arity at least two less than its
parent (\cref{lem:height-one-card-saturation}).

Choose a nonbinary irreducible \(f\) of minimum arity \(n\) in this
set. Every nonzero two-variable contraction of \(f\) must be a tensor
product of binary factors whose transfer classes lie in the current
group. We call a tensor product of binaries on the pairs of a perfect
matching a \emph{matching product}. The difficulty is to recover information about \(f\) itself:
different contractions could initially use different perfect matchings
of the remaining variables.

Walsh coordinates make the group representatives diagonal or
anti-diagonal. With
\(P=\left(\begin{smallmatrix}1&1\\1&-1\end{smallmatrix}\right)\), put
\(\widetilde f=(P^{-1})^{\otimes n}f\). Each binary factor of a
nonzero contraction is supported on two complementary assignments, with
the ratio of its ordered nonzero coefficients in \(-\mu_N\).
Here \(\mu_N=\{s:s^N=1\}\) and \(N\ge3\) is odd.
Fix two contracted variables \(i,j\), written first, and let
\(\mathbf y\) denote the remaining variables. After removing known
nonzero scalars, the rotation classes in either group realize
\[
q_s^+(\mathbf y)=\widetilde f(0,0,\mathbf y)
-s\widetilde f(1,1,\mathbf y),\qquad s\in\mu_N.
\]
The reflection classes in the dihedral group additionally realize
\[
q_s^-(\mathbf y)=\widetilde f(0,1,\mathbf y)
-s\widetilde f(1,0,\mathbf y),\qquad s\in\mu_N.
\]
Each family has the form \(q_s=A-sB\), where \(A,B\) are the indicated
slices obtained by fixing \(i,j\). We call it a \emph{contraction pencil}.
Thus the dihedral case permits comparisons between the equal-bit and
opposite-bit restrictions, while the cyclic case has only the first family.

In the dihedral case, we distinguish \(n=4,6\) from \(n\ge8\).
The low-arity cases are treated directly. At arity four, coefficient
analysis gives either a tensor factorization or a signature covered by
the equality reductions in \cref{short:sec:known}. At arity six, a finite
support classification, including computer-assisted enumeration for the
order-three group, followed by coefficient analysis forces a tensor
factorization (\cref{thm:odd-dih-q6-base} and
\cref{proofsubsec:cert-odd-dih-q6}). For \(n\ge8\), we instead recover
a binary factor of \(f\) from the compatibility of its contractions.
Although every nonzero contraction is already a product of binaries,
its matching---which variables are paired---and the coefficients of its
binary factors may initially depend on the contraction.

Fix a contracted pair \(i,j\) and one pencil \(q_s=A-sB\). Its linear
dependence on \(s\) links these otherwise separate factorizations.
Comparing three phases shows that their matchings either coincide or
differ only on four variables, with common binary factors elsewhere.
In the latter case, closing the common external pairs leaves a
six-variable proper contraction of \(f\), which is itself a matching
product. Oddness of \(N\) makes these closures nonzero. Its fixed
matching cannot produce different residual pairings as the phase varies,
excluding this possibility. Thus the matching of \(q_s\) is fixed.
On that matching, the coefficients form a rank-one tensor; linear
dependence on \(s\) then allows at most one binary factor to vary.
Consequently,
\[
q_s=H\otimes h_s,
\]
where \(H\) is a fixed product of binaries and \(h_s\) occupies one
remaining pair. Recovering \(A,B\) from two phases shows that \(H\)
factors out of both slices. Applying this argument to the equal-bit and
opposite-bit pencils gives factorizations of the two corresponding
restrictions of \(\widetilde f\). To compare them, we use contractions
that are nonzero on both restrictions: their residual decompositions
must come from the same matching product. This compatibility identifies
a binary factor shared by the two restrictions, and hence by their sum
\(\widetilde f\). Inverting the Walsh coordinates gives a factorization
of \(f\), contradicting irreducibility (\cref{prop:odd-dih-support-circuit-rigidity} and
\cref{lem:odd-dih-common-binary-factor}).

In the cyclic case, only the equal-bit pencils are available. We use them
to prove that the minimum nonbinary arity is four, then handle that
arity by binary attachments. For \(n\ge6\), the contraction supports
consist of assignments equal on each matched pair. The pencil relation
again fixes the matching as the phase varies. A pair whose entire pencil
vanishes would force its two bits to be opposite on every supported input
of \(\widetilde f\); comparing this constraint with contractions on other
pairs either gives incompatible matchings or a proper tensor
factorization. Thus every pair has a nonzero phase contraction. These
comparisons also force \(\widetilde f(\zero^n)\ne0\), so we normalize it
to one. The weight-two coefficients now reveal which original variables
must be paired. Put
\[
W=\{\{p,q\}:\widetilde f(\mathbf e_p+\mathbf e_q)\ne0\},
\]
where \(\mathbf e_p\) is the \(p\)-th unit vector. Any edge of \(W\)
survives in the fixed matching of a pencil contracted outside that edge.
Two adjacent edges would therefore have to belong to one matching,
which is impossible. Moreover, each edge \(pq\in W\) forces
\(x_p=x_q\) throughout the support of \(\widetilde f\): a violating
assignment could be preserved by a contraction outside \(p,q\), in
conflict with that contraction's matching.

If \(W\ne\varnothing\), compatibility with contractions meeting its
edges makes \(W\) a perfect matching. We may then regard
\(\widetilde f\) as a coefficient tensor \(h(t_1,\ldots,t_{n/2})\),
with one bit for each equal pair. Contracting variables from two pairs
forces their bits to agree and gives a rank-one restriction of \(h\).
These overlapping restrictions determine a product decomposition of
\(h\), and hence of \(f\). The six-variable case requires one further
coefficient check using the fact that squaring permutes \(\mu_N\).
If \(W=\varnothing\), a different obstruction occurs: the constant
coefficient of every contraction is one, while the rank-one identity on
two residual pairs would equate a linear polynomial in \(s\) with a
nonzero quadratic at all \(N\ge3\) phases. This is impossible. Both
alternatives exclude \(n\ge6\). At arity four, binary attachments first
align the coefficients on complementary Walsh inputs. The resulting
forms give an equality reduction, a classified eight-vertex signature,
or a tensor factorization, apart from one exceptional coefficient
pattern. A two-copy construction in that last case yields a binary
that is singular or has a transfer outside the allowed finite cyclic
group. It therefore invokes an earlier classification or contradicts
completeness of the group. This completes the cyclic branch (\cref{thm:odd-cyc-higher-collapse} and
\cref{thm:odd-cyc-quaternary-closure}).

\subsection{Completing the height-one classification}

These arguments complete a fixed cyclic or dihedral case once the
contraction procedure has stabilized. Either a known classification
applies, or no nonbinary irreducible signature remains. In the latter
case, the same Walsh matrix works for every signature: all binary
factors belong to the product class \(\cP\), consisting of products of
unary functions and equality or disequality constraints, and the edge
becomes \(P^{\mathsf T}XP=2\diag(1,-1)\in\cP\). Reconstructing the
signatures from these factors gives one common \(\cP\)-transformation
of the whole set, satisfying one of the seven structural alternatives.

It remains to ensure that the restarts preceding this analysis terminate.
A contraction may reveal a new binary transfer or a smaller nonbinary
irreducible factor. Express every transfer group
back in the original coordinates over one fixed number field. These
groups are nested, and their orders have a uniform bound depending only
on that field. Meanwhile, the minimum nonbinary irreducible arity never
increases. Each strict update increases the former or decreases the
latter. Between updates, the finite contraction recursion above
terminates by arity descent. Thus a classification is eventually reached
(\cref{lem:height-one-progress-well-founded} and
\cref{thm:height-one-transfer-interface}). All changes retain the
original signatures up to one common invertible transformation, so
hardness or a common tractable basis transfers back to the initial set.

\section{Higher heights: local contractions and reconstruction}
\label{short:sec:higher}

Write the mountain as
\(M=a\delta_{\zero^{2d}}+b\delta_{\one^{2d}}+Q\), where \(ab\ne0\),
\(\delta_{\mathbf x}\) is the indicator of the single assignment
\(\mathbf x\), and \(Q\) is its central restriction. Height two is
already covered by the eight-vertex classification in
\cref{short:sec:known}. For \(d\ge3\), the whole-set divisibility
condition makes every gadget output of positive arity less than \(2d\)
EO: its half-charge lies in \(d\mathbb Z\), but has absolute value less
than \(d\). In particular, an \(X\) self-loop gives
\(\partial_{ij}M=\partial_{ij}Q\), where
\((\partial_{ij}f)(\mathbf x)=f(0,1,\mathbf x)+f(1,0,\mathbf x)\)
and the other variables retain their order.

We seek an EO output whose presence forces hardness. Height three uses
a finite family of local constructions. At larger heights, constraints
on all pair contractions reconstruct the central restriction, finishing
heights at least five and leaving one exception at height four.

\subsection{Which EO outputs suffice?}

An EO output sits inside the current set, which also contains the
non-EO mountain. We therefore need a conclusion about this entire set.
Let \(\mathcal B_{\rm ARS}\) consist of nonzero scalar multiples of
tensor products, on oriented pairs of a perfect matching, of binaries
\((0,z,\overline z,0)\) with \(z\ne0\). The subscript refers to
arrow-reversal symmetry, the identity
\(h(\mathbf x)=\overline{h(\overline{\mathbf x})}\), where
\(\overline{\mathbf x}\) is bitwise complement.

\begin{lemma}[An EO output in the ambient set]
\label{short:lem:relative-eo}
Let \(\cF\) be a finite two-sided set of nonzero even-arity signatures
with \(\neg\Tract(K\cF)\). If \(\cF\) contains a nonzero positive-arity
EO signature outside \(\mathcal B_{\rm ARS}\), then
\(\KHolant(\cF)\) is \(\SharpP\)-hard.
\end{lemma}

The proof uses the constructions in the EO classifications
\cite{MengWangXia2025EO,CaiFuShao2020ARS} to obtain either a signature
classified by \cref{short:sec:known} or the asserted binary factorization.
The first stage forces arrow-reversal symmetry up to a scalar; the
second handles signatures with that symmetry that fail to factor.
To use this argument in the present problem, each fixed gadget must
remain valid in the ambient set, which also contains the non-EO
mountain. Gadget substitution supplies this extension for direct
constructions, and two-sidedness permits the separate factor-extraction
reductions. The extension is proved in
\cref{subsec:hc-relative-eo}. Thus, unless hardness
is obtained, every nonzero low-arity EO output has the specified
matching-product form. This turns a complexity classification into
structural information about the mountain's contractions.

\subsection{Height three and a two-copy construction}

Let \(\mathcal B_\times\) be the nonzero scalar multiples of matching
products of nonsingular anti-diagonal binaries
\(\left(\begin{smallmatrix}0&p\\q&0\end{smallmatrix}\right)\), with
\(pq\ne0\). It contains \(\mathcal B_{\rm ARS}\).
A matching-product central restriction already suffices to finish the
branch. If \(Q\in\mathcal B_\times\), join two copies of the complete
mountain so that their central matching edges and connecting edges form
two odd-length alternating paths between four exposed variables. Central
assignments supply four boundary patterns, with one bit equal to one on
each exposed pair. The two complementary endpoint choices supply two
additional patterns. Mixed endpoint--central assignments are excluded
by an internal matching edge, and each supported boundary has a unique
nonzero extension. Thus the quaternary EO output has all six weight-two
entries nonzero, whereas a quaternary member of \(\mathcal B_\times\)
has only four supported strings. By \cref{short:lem:relative-eo} this
gives hardness. The construction uses the whole mountain and works at
every height \(d\ge3\) (\cref{lem:hc-transversal-cube-terminal}).

Height three is handled directly, without first reconstructing \(Q\).
Self-loops test its four-variable restrictions; to obtain stronger
constraints, we also join two copies of the six-variable mountain along
four pairs of variables, again leaving four variables exposed. A finite
coefficient classification shows that if all fifteen self-loops and all
\(5400\) labelled two-copy outputs are zero or in \(\mathcal B_\times\),
then \(Q\) is supported on at most one complementary pair. This support
is too small to satisfy the same tests: if \(Q\ne0\), a self-loop
joining positions with opposite bits gives a nonzero output supported on
at most two strings; if \(Q=0\), a two-copy output has that property.
Either contradicts the four-string support of a quaternary member of
\(\mathcal B_\times\). One of the tested outputs must therefore give
hardness by \cref{short:lem:relative-eo}.
The finite classification uses computer-assisted exhaustive enumeration
and exact polynomial solver checks for arbitrary complex coefficients
and every \(ab\ne0\). The polynomial infeasibility steps rely on Z3's
\texttt{UNSAT} answers, as documented in \cref{thm:hc-q6-complete} and
\cref{proofsubsec:cert-high-charge-q6}.

\subsection{Normalization and reconstruction at larger heights}

For \(d\ge4\), we reconstruct the central restriction from its pair
contractions. We use the reconstruction theorem of Shao and Cai
\cite[full version, Lemma~2.11]{ShaoCai2020} after making the central
tensor real and complement-invariant and all its contractions products
of \(X\). The task is to obtain these simultaneous properties from our
complex-valued mountain. If \(Q=0\), the mountain is a generalized
equality, already classified by the auxiliary reductions (\cref{lem:pure-ge}). Otherwise, \cref{short:lem:relative-eo} makes
every nonzero contraction a product of EO--ARS binaries. Factor
extraction simulates these binaries. Attaching an appropriately oriented
copy of a factor \(B(z)=(0,z,\overline z,0)\) to one of its variables
turns that factor into \(|z|^2X\). The issue is to normalize all
contractions of the same mountain, since an attachment can affect
several of them. The resulting central restriction \(F\) must satisfy
\[
\partial_eF=\lambda_e X_{P_e},\qquad
X_{P_e}:=\bigotimes_{t\in P_e}X_t,\qquad \lambda_e\ne0
\quad\left(e\in\binom{[2d]}2\right),
\]
where \(P_e\) is a perfect matching of the remaining variables. To reach
this form, first normalize one nonzero contraction. Contracting a
disjoint pair in either order gives the same tensor, so the other
contraction must also be nonzero, and any binary factor surviving the
comparison must be proportional to \(X\). These comparisons constrain
the remaining factors sufficiently to choose further attachments that
preserve the first normalized contraction. They then propagate the
\(X\)-matching form to all pairs, including those whose original
contractions might have vanished. This simultaneous normalization is
proved in \cref{lem:hc-port-normalization}. All
attachments act on the complete mountain; its central restriction is
the tensor analyzed here.

The scalars \(\lambda_e\) may still be complex. To recover symmetry of
the central tensor, we need one common scalar normalization. Commuting
contractions on disjoint pairs \(e,f\) compares \(\lambda_e\) and
\(\lambda_f\): closing a pair in an \(X\)-matching product contributes
only a factor one or two, so their ratio is positive real. Since the
graph of disjoint pairs is connected, all \(\lambda_e\) have the same
phase. Dividing by one scalar makes every contraction real and
complement-invariant. These properties also hold for the central tensor
itself: its differences from its conjugate and from its bit-reversed
tensor are EO and have all pair contractions zero, which forces both
differences to vanish \cite[Lemma~2.7]{CaiFuShao2020ARS}. This establishes
the required symmetry from the contractions (\cref{lem:hc-global-phase}).

We can now apply reconstruction. At arity at least ten, it makes the
central restriction itself an \(X\)-matching product. The two-copy
construction above then gives hardness, completing every height
\(d\ge5\). At arity eight, reconstruction needs an additional overlap:
two distinct contracted pairs \(A,B\) must satisfy
\(A\cap B\ne\varnothing\) and their residual matchings share an edge,
\(P_A\cap P_B\ne\varnothing\). Such an overlap gives the same
matching-product conclusion and finishes that height-four branch.
The only remaining case has no such overlap.

\subsection{The height-four exception}

At height four, we first determine the remaining matching family, then
recover the coefficients of the mountain. Without the preceding overlap,
commutation of double contractions still requires the matchings to agree
after deleting any two disjoint variable pairs in either order. A
computer-assisted exhaustive classification of these compatible families
identifies the eight variables with \(\mathbb F_2^3\): contracting a pair
with difference \(t\ne0\) leaves the other three pairs \(\{x,x+t\}\).
This specifies the contraction matchings, but the coefficients of the
central tensor remain to be recovered. A contraction coefficient is the
sum of two entries whose supports differ by swapping the deleted bits.
The prescribed matching determines which such sums vanish and which
equal its nonzero scalar. The resulting linear relations force all
coefficients outside the fourteen affine planes of \(\mathbb F_2^3\)
to vanish and all plane coefficients to agree (\cref{lem:q8-collision-free-atlas} and
\cref{lem:q8-canonical-support}).

Consequently the mountain has the form
\(M^\circ=A\delta_{\zero^8}+B\delta_{\one^8}+cR_8^{\rm cen}\), with
\(ABc\ne0\), where \(R_8^{\rm cen}\) indicates those fourteen planes.
These fourteen weight-four strings, together with the two constant
strings, are the truth tables of the affine functions on \(\mathbb F_2^3\).
Their set is the Reed--Muller code
\[
RM(1,3)=\{(\alpha\cdot x+\beta)_{x\in\mathbb F_2^3}:
\alpha\in\mathbb F_2^3,\ \beta\in\mathbb F_2\}.
\]
Write \(\RMcore\) for its indicator. To obtain this signature up to
scale, we must also control the endpoint coefficients \(A,B\).
Join two copies of \(M^\circ\) along four coordinates that determine
each affine function uniquely. The remaining coefficients distinguish
endpoint contributions \(AB\) from central contributions \(c^2\).
One further closure produces a six-variable EO tensor whose coefficients
cannot factor over a matching unless \(AB=c^2\). Failure gives hardness
by \cref{short:lem:relative-eo}; equality permits a common diagonal
change of basis in \(GO(X)\) taking the mountain to \(c\RMcore\)
while preserving the whole-set charge condition (\cref{lem:q8-canonical-endpoint-repair}).

We now simulate binary equality, which allows the known dichotomy in
\cref{short:sec:known} to finish the branch. The Reed--Muller identity of
\cite[full version, Lemma~8.9]{ShaoCai2020} replaces three occurrences
of \(I\) by one occurrence and one Reed--Muller vertex, up to a known
nonzero scalar. Repetition leaves at most two occurrences, with only
linear growth in the instance size. Thus it suffices to evaluate
instances with one or two equality vertices. After removing these
vertices, the remaining network is over the current set \(\cF\), whose
supported half-charges all belong to \(4\mathbb Z\). Its boundary tensor
is therefore anti-diagonal on two variables and EO on four variables.
One equality occurrence contributes zero, since it requires equal
boundary bits. For two occurrences, only the weight-two entries matter,
and on that support
\[
I_{12}I_{34}
=\frac{X_{13}X_{24}+X_{14}X_{23}-X_{12}X_{34}}2.
\]
Here \(I_{ij}\) and \(X_{ij}\) act on the indicated pair of variables.
The right-hand side uses only existing \(X\)-edges, so three oracle
queries evaluate the two-occurrence instance. Adjacent equality vertices
collapse directly, and equality self-loops have value zero, covering
arbitrary instances. We have therefore proved
\(\KHolant(\cF,I)\leT\KHolant(\cF)\), and the binary-equality
dichotomy completes the classification. The full version proves the
stronger simulation of all four Bell signatures in
\cref{thm:q8-normalized-bell-lift}.

\section{Completing the dichotomy}
\label{short:sec:completion}

\begin{proof}[Proof of \cref{short:thm:main}]
Decidability and polynomial-time evaluation follow from
\cite[Theorem~2.22]{Liu2026DEQ2Holant}. Suppose none of the seven
structural alternatives holds. By \cref{short:sec:known}, it remains to
treat a two-sided, multilayer set of irreducible even-arity signatures
in \(K\)-coordinates. \Cref{short:lem:mountain} either simulates
pinnings, invoking the odd-arity classification, or supplies a mountain
of height \(d\) and the whole-set charge condition.
\Cref{short:sec:height-one} handles \(d=1\), and
\cref{short:sec:higher} handles \(d\ge2\). The simulations and common
changes of basis preserve failure of the structural condition, so the
applicable dichotomies give hardness. Composing the reductions with
factor reconstruction and the initial holographic transformation proves
the theorem (\cref{sec:synthesis}).
\end{proof}

Appendix~\ref{sec:exact-finite-certificates} specifies the finite classifications used
in the odd-dihedral, height-three, and height-four arguments, including
their coverage and computational dependencies. It distinguishes
enumeration and solver checks used in the proofs from supplementary
computations.

\section*{AI methodology statement}
\addcontentsline{toc}{section}{AI methodology statement}

We used OpenAI's GPT-5.6 Sol to assist with mathematical exploration and
proof development throughout this paper, including exact symbolic
calculations and finite-case checks. All companion verification code was
generated by OpenAI Codex and subsequently reviewed by the authors.

We used OpenAI's GPT-6 Astra to improve the writing and presentation and
to assist with the \href{https://github.com/liuchliuch/holant-dichotomy-lean}{\pinkfoldmark\nobreak\hspace{0.2em}Lean formalization}. We completed the Lean formalization
of the proof. The authors take full responsibility for the correctness
and originality of the paper.

\section*{Acknowledgments}
\addcontentsline{toc}{section}{Acknowledgments}

Chenghua Liu thanks Zhengfeng Ji and Mingsheng Ying for their continued
encouragement and support. Boning Meng and Juqiu Wang thank Mingji Xia for
his guidance and mentorship.

\label{s:end:main}


\clearpage
\phantomsection
\addcontentsline{toc}{section}{References}
\newcommand{\etalchar}[1]{$^{#1}$}


\clearpage
\setcounter{section}{0}
\setcounter{equation}{0}
\setcounter{figure}{0}
\setcounter{table}{0}
\setcounter{footnote}{0}
\setcounter{page}{1}
\renewcommand{\thepage}{F-\arabic{page}}
\renewcommand{\thesection}{F.\arabic{section}}
\renewcommand{\theHsection}{full.\arabic{section}}
\renewcommand{\theequation}{F.\arabic{equation}}
\renewcommand{\theHequation}{full.\arabic{equation}}
\renewcommand{\thefigure}{F.\arabic{figure}}
\renewcommand{\theHfigure}{full.\arabic{figure}}
\renewcommand{\thetable}{F.\arabic{table}}
\renewcommand{\theHtable}{full.\arabic{table}}
\section*{Full paper}
\section{Introduction}
\label{sec:introduction}

Counting problems ask how many global configurations satisfy a collection
of local constraints.  For example, counting proper colourings of a graph
means counting assignments of colours to vertices such that adjacent
vertices receive different colours.  In a weighted version, each local
configuration contributes a weight, and the weight of a global
configuration is the product of its local weights.  Summing these products
over all configurations gives a partition function.  Ordinary counting
is the special case in which the local weights are zero or one.

A short description of the local rules can therefore specify a sum with
exponentially many terms.  Some rules allow this sum to be evaluated
efficiently without enumerating the configurations, whereas others lead
to computationally hard counting problems. A complexity dichotomy
identifies all tractable choices of local rules within a given framework
and proves hardness for every remaining choice. Holant provides a
graph-based framework for studying this boundary.

A Boolean Holant problem evaluates a graph partition function specified by
local functions, called signatures, at its vertices.  Each edge receives a
bit, and the partition function sums the product of the vertex signatures
over all edge assignments.  The classification problem asks which fixed
signature sets permit polynomial-time evaluation and which give
\(\SharpP\)-hard problems.

We prove a dichotomy for arbitrary finite sets of algebraic complex-valued
Boolean signatures on general graphs, with an explicit, decidable
tractability criterion.  The main case addressed here consists entirely of
even-arity signatures, with neither a symmetry restriction nor a prescribed
auxiliary binary signature.  Its resolution completes a line of
classification results extending from graph homomorphisms through counting
constraint satisfaction problems to Holant problems.

\subsection{From graph homomorphisms through counting CSP to Holant}
\label{subsec:intro-history}

Counting graph homomorphisms (\(\#\mathrm{GH}\)) asks how many maps from an
input graph to a fixed target graph preserve adjacency.  In the weighted
version, a fixed matrix specifies the interaction between the labels at
the two ends of each edge.  Counting constraint satisfaction problems
(\(\#\mathrm{CSP}\)) allow constraints of arbitrary arity, so one local
function can constrain several variables at once.  The weighted problem
sums the product of the constraint values over all variable assignments.

Dichotomies for unweighted graph homomorphisms and Boolean counting CSP
\cite{DyerGreenhill2000,CreignouHermann1996} established early classifications
in these frameworks.  Subsequent work extended the allowed domains and
weights, from nonnegative weights to signed-real and complex weights,
where cancellation must also be accounted for.  The general-domain complex
classification of Cai and Chen covers both weighted graph homomorphisms
and counting CSP \cite{CaiChen2017}.
\Cref{tab:intro-gh-csp-history} records this development.  A dash means that
the case is covered by a more general entry, with no separate reference
listed.

	\begin{table}[!htbp]
		\centering
		\caption{Development of general-graph dichotomies for weighted graph
			homomorphisms and counting CSP.}
		\label{tab:intro-gh-csp-history}
		\small
		\setlength{\tabcolsep}{3.5pt}
		\renewcommand{\arraystretch}{1.16}
		\begin{tabularx}{\textwidth}{@{}l
				>{\raggedright\arraybackslash}X
				>{\raggedright\arraybackslash}X
				>{\raggedright\arraybackslash}X
				>{\raggedright\arraybackslash}X@{}}
			\toprule
			Weight range & Symmetric \(\#\mathrm{GH}\) & \(\#\mathrm{GH}\) &
			Boolean-domain \(\#\mathrm{CSP}\) &
			General-domain \(\#\mathrm{CSP}\) \\
			\midrule
			\(\{0,1\}\)
			& \cite{DyerGreenhill2000}
			& acyclic targets \cite{DyerGoldbergPaterson2007}
			& \cite{CreignouHermann1996}
			& \cite{Bulatov2013,DyerRicherby2013} \\
			\(\mathbb Q_{\ge0}\)
			& -- & -- & --
			& \cite{BulatovDyerGoldbergJalseniusJerrumRicherby2012} \\
			\(\mathbb R_{\ge0}\)
			& \cite{BulatovGrohe2005}
			& \cite{CaiChen2019}
			& \cite{DyerGoldbergJerrum2009}
			& \cite{CaiChenLu2016} \\
			\(\mathbb R\)
			& \cite{GoldbergGroheJerrumThurley2010}
			& --
			& \cite{BulatovDyerGoldbergJalseniusRicherby2009}
			& -- \\
			\(\mathbb C\)
			& \cite{CaiChenLu2013}
			& \cite{CaiChen2017}
			& \cite{CaiLuXia2014}
			& classification \cite{CaiChen2017},
			decidability \cite{LiuMeng2026CSPDecidability} \\
			\bottomrule
		\end{tabularx}
	\end{table}

Holant changes which local constraints are available.  A counting-CSP
instance can be represented by a bipartite graph with variable vertices
and constraint vertices.  At each variable vertex, an equality signature
forces all occurrences of that variable to agree.  In a Holant problem,
these equality signatures need not be available: every vertex must use a
signature from the prescribed set. The signatures needed in a reduction
must be obtained by gadgets or other justified reductions. This distinction brings
problems such as counting perfect matchings into the classification and
changes the constructions available for proving hardness.

Motivated by Valiant's holographic algorithms \cite{Valiant2008}, the
Holant framework was introduced in 2009 \cite{CaiLuXia2009}.
A holographic transformation is an invertible simultaneous change of basis
on the signatures and edge tensors that preserves their contraction.
Such transformations relate apparently different counting problems and
expose tractable signature sets.

On the Boolean domain, Holant classifications developed along two
complementary directions: relaxing restrictions on symmetry and weights,
and reducing the auxiliary signatures assumed available.  Unary
signatures are useful because they allow an input of a signature to be
fixed or summed with a chosen weight.  The auxiliary assumptions range
from all unary signatures (\(\Holant^*\)), through four distinguished
unaries (\(\Holant^+\)), to the two unaries fixing a bit to 0 or 1
(\(\Holant^c\)).  Later work treated any signature set containing a
nonzero odd-arity signature (\(\Holant^{\mathrm{odd}}\)).
\Cref{tab:intro-boolean-holant-history} separates these assumptions from
restrictions on symmetry and weights, using the same dash convention as
\cref{tab:intro-gh-csp-history}.

	\begin{table}[!htbp]
		\centering
		\caption{Development of Boolean-domain Holant dichotomies on general
			graphs.}
		\label{tab:intro-boolean-holant-history}
		\small
		\setlength{\tabcolsep}{3pt}
		\renewcommand{\arraystretch}{1.18}
		\begin{tabularx}{\textwidth}{@{}l
				>{\raggedright\arraybackslash}X
				>{\raggedright\arraybackslash}X
				>{\raggedright\arraybackslash}X@{}}
			\multicolumn{4}{@{}l}{\textit{(a) Symmetric signature sets}}\\[0.2em]
			\toprule
			Weight range & \(\Holant^*\) & \(\Holant^c\) & Unrestricted \(\Holant\) \\
			\midrule
			\(\mathbb R_{\ge0}\) & -- & -- & -- \\
			\(\mathbb R\) & -- & \cite{CaiLuXia2009}
			& \cite{HuangLu2016} \\
			\(\mathbb C\) & \cite{CaiLuXia2009}
			& \cite{CaiHuangLu2012}
			& \cite{CaiGuoWilliams2016} \\
			\bottomrule
		\end{tabularx}
		
		\medskip
		\begin{tabularx}{\textwidth}{@{}l
				>{\raggedright\arraybackslash}X
				>{\raggedright\arraybackslash}X
				>{\raggedright\arraybackslash}X
				>{\raggedright\arraybackslash}X
				>{\raggedright\arraybackslash}X@{}}
			\multicolumn{6}{@{}l}{\textit{(b) Arbitrary signature sets}}\\[0.2em]
			\toprule
			Weight range & \(\Holant^*\) & \(\Holant^+\) & \(\Holant^c\) &
			\(\Holant^{\mathrm{odd}}\) & Unrestricted \(\Holant\) \\
			\midrule
			\(\mathbb R_{\ge0}\) & -- & -- & -- & --
			& \cite{LinWang2018} \\
			\(\mathbb R\) & -- & --
			& \cite{CaiLuXia2018RealHolantc}
			& \cite{CaiFuShao2020Entanglement}
			& \cite{ShaoCai2020} \\
			\(\mathbb C\)
			& \cite{CaiLuXia2011}
			& \cite{Backens2017}
			& \cite{Backens2021,Backens2025}
			& \cite{MengWangXiaZheng2025,GuanShaoShi2026}
			& This paper \\
			\bottomrule
		\end{tabularx}
	\end{table}

Without auxiliary unary signatures, dichotomies were obtained for
symmetric complex-valued signatures \cite{CaiGuoWilliams2016}, arbitrary
nonnegative signatures \cite{LinWang2018}, and arbitrary real-valued
signatures \cite{ShaoCai2020}.

In 2025, Meng, Wang, and Xia proved an \(\FP^{\mathsf{NP}}\) versus
\(\SharpP\)-hard dichotomy for weighted Eulerian orientations
\cite{MengWangXia2025EO}.  Meng, Wang, Xia, and Zheng combined this result
with tensor decomposition to obtain the same complexity bounds for complex
Boolean Holant with a nonzero odd-arity signature
\cite{MengWangXiaZheng2025}.  In 2026, Guan, Shao, and Shi gave a
polynomial-time algorithm based on linear programming for the former
\(\FP^{\mathsf{NP}}\) cases, yielding \(\FP\) versus \(\SharpP\)-hard
dichotomies for both problems \cite{GuanShaoShi2026}.

These results direct attention to the all-even case.  Every gadget built
from even-arity signatures has even arity, since each internal edge consumes
two inputs.  Thus no such gadget supplies a unary signature.  Binary
signatures become the smallest positive-arity objects available for further
constructions.  Xia proposed a framework organized by the finite groups
associated with binary signatures, analyzing several branches in which all
realizable nonzero quaternary signatures decompose \cite{Xia2026Framework}.
Liu and Meng then classified complex Boolean Holant with binary
disequality available, for arbitrary finite sets of algebraic
complex-valued signatures, and gave an explicit decidable tractability
criterion \cite{Liu2026DEQ2Holant}.

Our proof removes the prescribed-binary assumption in the all-even case.
It combines the binary and quaternary dichotomies of Liu and Meng with a
decomposition based on Hamming imbalance.  This decomposition produces
signatures of controlled support whose arity determines the remaining
cases, as described in \cref{subsec:proof-sketch}.  Together with the
odd-arity dichotomy, it completes the unrestricted complex-valued entry in
\cref{tab:intro-boolean-holant-history}.

\paragraph{Independent work.}
Huang and Fu established a dichotomy for Holant problems with a single
complex-valued quaternary signature \cite{HuangFu2026FourRegular}.
Guan, Shao, and Tang proved a dichotomy for complex-valued Boolean Holant
problems whose signature sets are closed under complex conjugation
\cite{GuanShaoTang2026ConjugateClosed}.
Zhou proved a dichotomy for finite sets of algebraic complex-valued Boolean
signatures of even arity whose nonzero binary gadget signatures are
nonsingular and whose projective matrix classes form the cyclic group
\(C_2\) of order two under series composition \cite{ZhouC2Holant}.

We were also informed of the following independent work before the
public release of the present paper.
Jiayi Zheng informed us of unpublished work on the complex-valued Holant
dichotomy in the case where the projective binary transfer classes form an
odd dihedral group.
Mingji Xia informed us of work concerning the twisted Klein-group case
and the \(C_1\) case, under the condition that all quaternary signatures
are reducible.
The results and proofs in the present paper were developed independently
of these works.

\subsection{Our result: the full dichotomy}
\label{subsec:intro-main}

The predicate \(\Tract(\cF)\), defined in \cref{sec:predicate}, is an
equivalent formulation of the seven conditions in the odd-arity dichotomy
of Meng, Wang, Xia, and Zheng
\cite[conference version, Theorem~30; full version, Theorem~33]{MengWangXiaZheng2025}.
Our result applies this criterion without an odd-arity assumption.
Each condition concerns the signature set as a whole.  In each transformability
case, one basis must work simultaneously for the edge and every
member of \(\cF\).

	\begin{theorem}[Full dichotomy theorem]
		\label{thm:main}
		Let \(\cF\) be a finite set of algebraic complex-valued Boolean signatures,
		given in the exact algebraic representation fixed in
		\cref{subsec:complexity-conventions}.  The predicate \(\Tract(\cF)\)
		is decidable, and the following complexity classification holds:
		\[
		\begin{aligned}
			\Tract(\cF)
			&\quad\Longrightarrow\quad \Holant(\cF)\in\FP,\\
			\neg\Tract(\cF)
			&\quad\Longrightarrow\quad
			\Holant(\cF)\text{ is }\SharpP\text{-hard}
		\end{aligned}
		\]
		under polynomial-time Turing reductions.
	\end{theorem}

For evaluation and hardness reductions, \(\cF\) is fixed and running time
is measured in the encoded instance size. For recognition, \(\cF\) itself
is the input.
The proof of \cref{thm:main} is completed in \cref{sec:assembly}, where
\cref{thm:K-structural} is transferred back to ordinary coordinates.

\subsection{Proof organization}
\label{subsec:proof-sketch}

We first apply the known odd-arity dichotomy.  For the all-even case,
we reduce to known classifications or construct a \emph{mountain}, whose
height determines the remaining analysis.
An input \({\mathbf x}\) to an arity-\(r\) signature has \emph{half-charge}
\(\chi_r({\mathbf x})=\wt({\mathbf x})-r/2\), where \(\wt({\mathbf x})\) is its number of 1-bits.
In the \(K\)-basis, edges carry disequality, so the vertex charges in a
closed instance sum to zero.  If all supported charges are nonnegative,
or all are nonpositive, only the central Hamming layers contribute, and
the weighted Eulerian-orientation classification applies.  That
classification also handles sets in which each signature is supported on
a single Hamming layer, meaning inputs of one fixed weight.

The remaining set is \emph{multilayer}: some signature has nonzero entries at
different Hamming weights.  It is also \emph{two-sided}: positive and
negative half-charges both occur, possibly in different signatures.
We extract tensor factors and handle the cases covered by known
classifications.  For the remaining set, let \(d\) be the gcd of its
nonzero supported half-charges.  The construction in \cref{lem:charge-bridge}
first uses tensor
products and self-loops to obtain two signatures, one nonzero at its
all-zero input and the other at its all-one input.  We combine these
endpoint properties by further connections, or obtain unary pins and
apply the odd-arity dichotomy.  In the first case, the output is a
\emph{mountain of height \(d\)}: an arity-\(2d\) signature with both
endpoint coefficients nonzero and all other support on the weight-\(d\)
layer.  Every supported half-charge in the current set remains divisible
by \(d\), a condition used in the higher-height arguments.

At height one, the mountain is binary.  We use interpolation and finite
projective groups to analyze binary gadgets.  Binary equality is handled
by Liu and Meng's dichotomy \cite[Theorem~1.2]{Liu2026DEQ2Holant}.
The remaining groups are odd cyclic or odd dihedral.  Contraction analysis
and a terminating sequence of signature additions settle these two cases.

Height two is resolved by the eight-vertex dichotomy and factor extraction.
Height three uses self-loop and two-copy constructions and computer-assisted
exhaustive enumeration.

For the remaining cases at heights at least four, we normalize the
mountain and reconstruct its central restriction from its self-loop
contractions.  At heights at least five, the recovered matching structure
yields a hard two-copy gadget.  Height four has an exceptional
Reed--Muller structure, which we handle by a reduction adding binary
equality.  Some base cases in the odd-dihedral and higher-height arguments
depend on exact enumeration and solver checks.

Sections~\ref{sec:preliminaries}--\ref{sec:factor-charge} present the basic
conventions, tractability criterion, auxiliary reductions, and mountain
construction, respectively.
Sections~\ref{sec:binary-dynamics}--\ref{sec:height-at-least-four} treat
the height cases, \cref{sec:synthesis} completes the main proof, and
\cref{sec:conclusion} concludes.
Appendix~\ref{sec:exact-finite-certificates} records the finite
calculations and their verification.

\subsection{Related work}
\label{subsec:related-work}

The historical progression above concerns exact counting on general
graphs.  Further work studies recognition of the tractable cases, larger
domains, specific vertex models, and planar restrictions.

The last cell of \cref{tab:intro-gh-csp-history} distinguishes a
classification for each fixed signature set from an algorithm that decides
whether the tractability criterion holds.  Cai and Chen's 2017 theorem expressed tractability
through three conditions on an unbounded generated family of partial sums
\cite{CaiChen2017}.  Liu and Meng subsequently gave a total exact algorithm
deciding these conditions from a finite set of algebraic complex-valued
constraints.  They proved that Block Orthogonality implies the other two
conditions and also decided the corresponding degree-multiple
counting-CSP criterion \cite{LiuMeng2026CSPDecidability}.

For Holant on larger domains, classifications are available under further
restrictions on the signatures.  A domain-three theorem treated one
symmetric ternary signature with all unary signatures available
\cite{CaiLuXia2013Domain3}. Its journal version appeared in 2026
\cite{CaiLuXia2026Domain3}.  Subsequent classifications covered special
edge-colouring signatures and restricted domain-three and domain-four
settings \cite{CaiGuoWilliams2016EdgeColoring,LiuFanCai2025}.  In 2025,
Cai and Ihm proved a dichotomy for real-valued symmetric signatures on
domain three with all unary signatures available, allowing an arbitrary
finite signature set \cite{CaiIhm2025}.

More specific Boolean models include a single ternary signature on
3-regular graphs, the six-vertex and eight-vertex models, and weighted
Eulerian orientations with arrow-reversal symmetry
\cite{YangHuangFu2024,CaiFuXia2018,CaiFu2023,CaiFuShao2020ARS}.

	Restricting the input graph to be planar can make additional counting
	problems tractable.  For Boolean counting CSP, Cai, Lu, and Xia classified
	real-valued symmetric constraints \cite{CaiLuXia2017PlanarCSP}.
	Guo and Williams extended the classification to symmetric constraints with
	complex weights \cite{GuoWilliams2020PlanarCSP}, and Cai and Fu removed the
	symmetry assumption \cite{CaiFu2022PlanarCSP}.  In this planar model, the
	incidence graph is embedded in the plane and the argument order follows
	the cyclic order around each constraint vertex.  Every case that is hard
	on general graphs but tractable in this planar setting admits a holographic
	reduction to matchgates.  These reductions evaluate the partition function
	through weighted perfect-matching sums computed by the FKT algorithm.
	
	Graph homomorphism problems are a subclass of counting CSP, so these
	Boolean-domain results also classify planar graph homomorphisms with
	complex symmetric weight matrices.  For larger domains, Cai and Maran
	proved a dichotomy for real symmetric \(3\times3\) matrices in 2023
	\cite{CaiMaran2023PlanarGH}.  Their 2024 preprint classified nonnegative
	real symmetric \(4\times4\) matrices of full rank
	\cite{CaiMaran2024PlanarGH}.  In a 2026 preprint, Cai, Maran, and Young
	obtained a dichotomy for nonnegative real symmetric matrices of arbitrary
	finite size whose diagonal entries are pairwise distinct
	\cite{CaiMaranYoung2026}.
	
	For planar Holant with symmetric complex-valued Boolean signatures,
	Cai, Fu, Guo, and Williams gave a classification that includes tractable
	cases beyond holographic reductions to the FKT algorithm
	\cite{CaiFuGuoWilliams2022PlanarHolant}.  In 2026, Cai, Fan, Shao, and Tang
	classified the eight-vertex model into three cases: tractable on general
	graphs, hard on general graphs but tractable on planar graphs, and hard
	even on planar graphs \cite{CaiFanShaoTang2026}.
	
	\section{Preliminaries}
\label{sec:preliminaries}

\subsection{Basic notation}
\label{subsec:foundational-algebra}
\label{subsec:graph-code-walsh-terminology}

We use \(\mathbb Q,\mathbb R\),
and \(\mathbb C\) for the rationals, reals, and complex numbers, and write
\[
\AlgNums:=\{\alpha\in\mathbb C:p(\alpha)=0
\text{ for some }0\ne p\in\mathbb Q[t]\},
\qquad
\Ftwo:=\mathbb Z/2\mathbb Z
\]
for the algebraic complex numbers and the two-element field.
For a field \(F\), put \(F^\times=F\setminus\{0\}\). For a finite set
\(V\), \(F^V\) denotes the functions \(V\to F\).
We write \([m]=\{1,\ldots,m\}\), with \([0]=\varnothing\), and
\(\Mat_{m\times n}(F)\) for matrices over \(F\), abbreviated to
\(\Mat_n(F)\) when \(m=n\).  Transposes are denoted by \(M^{\mathsf T}\)
and do not include complex conjugation.  The group of invertible matrices
is \(\GL_n(F)\), and \(I_n\) is its identity.
For a vector space \(V\) over \(F\), tensor products use the convention
\(V^{\otimes0}=F\), with empty product \(1\).

Vector symbols are bold, while scalar coordinates and variable or vertex labels
are not.  For a Boolean vector \({\mathbf x}\in\Ftwo^V\), put
\(\wt({\mathbf x})=|\{v:x_v=1\}|\).  A Hamming layer fixes \(\wt({\mathbf x})\), and parity
means \(\wt({\mathbf x})\bmod2\).  All binary-vector arithmetic is over \(\Ftwo\).
We use \(\oplus\) for coordinatewise addition and \(\bar {\mathbf x}\) for bitwise complement.  For \({\mathbf y}\in\Ftwo^V\), the dot product
\({\mathbf x}\cdot {\mathbf y}=\sum_{v\in V}x_vy_v\) is taken in \(\Ftwo\).
Binary strings on ordered index sets are ordered lexicographically with
\(0<1\).  For \(A\subseteq V\), \({\mathbf x}_{\widehat A}\) denotes the restriction
of \({\mathbf x}\) to \(V\setminus A\), in the inherited order when \(V\) is ordered.
In particular, \({\mathbf x}_{\widehat{pq}}\) deletes positions \(p,q\).

\subsection{Signatures and Holant instances}
\label{subsec:holant-definitions}

For an integer \(r\ge0\), an arity-\(r\)
	Boolean signature is a
function \(f:\{0,1\}^r\to\AlgNums\). Its truth table, listed as a column
vector in lexicographic order, identifies it with a tensor in
\((\AlgNums^2)^{\otimes r}\). Its variables have labelled positions whose
order is part of the signature. We
	write \(\arity(f)=r\) and \(\supp(f)=\{{\mathbf x}:f({\mathbf x})\ne0\}\), and use
	\(\one[\mathsf P]\) for the indicator of a proposition \(\mathsf P\).
	A signature set is \emph{all-even} when each member has even
arity. The terms nullary, unary, binary, ternary, and quaternary mean
arity \(0,1,2,3,4\), respectively.
	For an ordered binary signature \(B\), we identify \(B\) with the matrix
	\((B(a,b))_{a,b\in\Ftwo}\), with rows indexed by \(a\) and columns by \(b\).

The symbols \(\zero^r\) and \(\one^r\) denote the all-zero and all-one
strings. At \(r=0\), both denote the unique empty string.  More generally,
\(\zero^V,\one^V\in\Ftwo^V\) are the constant vectors on a finite index set
\(V\).  The values \(f(\zero^r)\) and \(f(\one^r)\) are the two
\emph{endpoint coefficients}.  For \(T\subseteq[r]\),
\(\one_T\in\{0,1\}^r\) is the incidence vector of \(T\), so that
\(\one_\varnothing=\zero^r\).  We write \({\mathbf e}_j=\one_{\{j\}}\) for the
\(j\)-th standard basis vector when the ambient arity is clear.

Graphs may have parallel edges and loops. Each loop contributes two to
\(\deg_G(v)\).  A perfect matching of a finite set \(V\) is a partition
into pairs.  The mate of a vertex is the other endpoint of its matching
edge.  We write \(K_n\) for the complete graph and \(\binom Vk\) for the
\(k\)-element subsets of \(V\).

For a signature set \(\cF\), a signature grid over \(\cF\) is a finite
	multigraph \(G=(V,E)\),
	together with a signature \(f_v\in\cF\) satisfying
	\(\arity(f_v)=\deg_G(v)\) and an ordering of the incident
	edge-ends at every vertex.  Let \(\mathbf x_v\) list the local variables at \(v\) in this order.
	An edge assignment \(\sigma:E\to\{0,1\}\) induces the string
	\(\left.\mathbf x_v\right|_\sigma\) by assigning each variable the bit on its edge.
	The two ends of a loop receive the same edge bit.
	For the resulting grid \(\Omega\),
	\begin{equation*}
		\Holant_\Omega
		=
		\sum_{\sigma:E\to\{0,1\}}
		\prod_{v\in V}
		f_v(\left.\mathbf x_v\right|_\sigma).
	\end{equation*}
	The problem \(\Holant(\cF)\) asks for this value on an input grid over
	\(\cF\).
	Throughout, comma-separated signature or signature-set arguments mean adding
	their union. For example, \(\Holant(\cF,a,b)\) abbreviates
	\(\Holant(\cF\cup\{a,b\})\).

We use the unary pinnings
	\[
	\DeltaZero=[1,0]^{\mathsf T},
	\qquad
	\DeltaOne=[0,1]^{\mathsf T},
	\]
	and the binary equality and disequality tensors
	\[
	I=I_2=(\Eq_2)=
	\begin{pmatrix}1&0\\0&1\end{pmatrix},
	\qquad
	X=(\Neq_2)=
	\begin{pmatrix}0&1\\1&0\end{pmatrix}.
	\]
	We also fix the sign and skew matrices
	\[
	Z=\begin{pmatrix}1&0\\0&-1\end{pmatrix},
	\qquad
	Y=ZX=\begin{pmatrix}0&1\\-1&0\end{pmatrix}.
	\]
	More generally, for \(r\ge1\), let
	\[
	(\Eq_r)(x_1,\ldots,x_r)=\one[x_1=\cdots=x_r].
	\]
	For signature sets \(\cG\) and \(\cF\), the bipartite problem
	\(\Holant(\cG\mid\cF)\) is defined on a bipartite graph
	\(H=(U,V,E)\): vertices in the left part \(U\) are assigned signatures
	from \(\cG\), and vertices in the right part \(V\) are assigned signatures
	from \(\cF\).

Subdividing every ordinary edge by an \(\Eq_2\)-vertex preserves the
partition value, as does suppressing these degree-two equality vertices.
Under this correspondence between instances, we identify
\[
\Holant(\cF)=\Holant(\Eq_2\mid\cF).
\]

For an integer \(r\ge0\) and \({\mathbf a}\in\{0,1\}^r\), define the
\emph{indicator signature} \(\delta_{\mathbf a}\) by
	\[
	\delta_{\mathbf a}({\mathbf x})=\one[{\mathbf x}={\mathbf a}].
	\]

\subsection{Zero and nullary preprocessing}
\label{subsec:zero-nullary}

An instance containing an identically zero signature has value zero.
A nullary signature is a scalar: delete its vertex and multiply that scalar
back into the final value.  The empty residual instance has value \(1\).
After this preprocessing, we assume that the input signature set \(\cF\)
contains neither zero nor nullary signatures.  Intermediate contractions
may still be zero or nullary and are treated as specified in each argument.

\subsection{Complexity and exact-input conventions}
\label{subsec:complexity-conventions}

We treat counting problems as functions from finite encoded instances to
exact values.  We use \emph{tractable} to mean membership in \(\FP\), the
functions computable exactly in deterministic polynomial time.
The class \(\SharpP\) consists of functions counting accepting paths of
nondeterministic polynomial-time Turing machines.  A possibly complex-valued counting
problem is \emph{hard} if it is \(\SharpP\)-hard under polynomial-time
Turing reductions.

For counting problems \(A,B\), write \(A\leT B\) if an algorithm computes
\(A\) using polynomially many exact evaluations of \(B\) on instances of
polynomial size, with polynomial time for all other work.  This time bound
includes any known nonzero algebraic rescaling and exact algebraic
postprocessing.  Write \(A\equivT B\) when reductions exist in both
directions.  The notation \(\FP^{\mathsf{NP}}\) permits an \(\mathsf{NP}\)
decision oracle.

We fix an \emph{exact algebraic encoding} of \(\alpha\in\AlgNums\) by its
monic irreducible polynomial over \(\mathbb Q\) and an isolating rectangle
in \(\mathbb C\): its vertices have rational real and imaginary parts,
and it contains \(\alpha\) and no other complex root of the polynomial.
Equality, zero tests, conjugation, and field arithmetic are exact.
A \emph{number field} is a finite extension \(L/\mathbb Q\), and
\([L:\mathbb Q]\) is its dimension as a \(\mathbb Q\)-vector space.

For evaluation of \(\Holant(\cF)\), the signature set \(\cF\) is fixed,
and running time is measured in the encoded instance size.  For recognition,
the input is a finite list of signatures given by exact algebraic entries.
No polynomial-time bound
in the length of this list is asserted.  

\subsection{Tensor operations and projective notation}
\label{subsec:algebra-group-terminology}

For an integer \(r\ge0\), let \(S_r\) denote the symmetric group on
	\([r]\).  For an arity-\(r\) signature \(f\) and \(\tau\in S_r\),
	\((\tau f)({\mathbf x})=f(x_{\tau^{-1}(1)},\ldots,x_{\tau^{-1}(r)})\).  For signatures
	\(f:\Ftwo^r\to\AlgNums\) and \(g:\Ftwo^s\to\AlgNums\) on disjoint ordered
	variables, and \({\mathbf x}\in\Ftwo^r,{\mathbf y}\in\Ftwo^s\), define
	\((f\otimes g)({\mathbf x},{\mathbf y})=f({\mathbf x})g({\mathbf y})\).
	The signature \(f\) is \emph{symmetric} when \(\tau f=f\) for every
	\(\tau\in S_r\), equivalently when \(f({\mathbf x})\) depends only on \(\wt({\mathbf x})\).

For \({\mathbf b}\in\Ftwo^r\), the bit-flip operation sends \(f\) to
\(f^{\oplus {\mathbf b}}({\mathbf x})=f({\mathbf x}\oplus {\mathbf b})\).  Bit-flip/variable-permutation orbits use
these flips and permutations of argument positions.
A nonzero tensor has rank one when it is a product of one vector on each
of its specified factors. For two factors this is matrix rank one.
We say \(a\) divides \(f\), with cofactor \(b\), when
\(f=a\otimes b\) up to variable permutation and nonzero scalar redistribution.
Every tensor is taken to divide the zero tensor.  A proper tensor factor
has positive arity strictly smaller than its parent.

A nonzero positive-arity signature is \emph{irreducible} if, up to variable permutation and nonzero scalar, it cannot
	be written as \(a\otimes b\) with both \(a\) and \(b\) of positive arity.
	A nonzero positive-arity signature that is not irreducible is
	\emph{reducible}.
	Factorization into irreducibles is unique up to these operations and factor order
	\cite[full version, Lemma~2.13]{CaiFuShao2020ARS}.  For any class \(\cC\) of
	signatures, its tensor closure \(\angles{\cC}\) consists of tensor products
	of members of \(\cC\), up to variable permutation and nonzero scalar.  By
	convention, it also contains zero and the arity-zero empty product.

For an integer \(r\ge1\) and \({\mathbf a}\in\{0,1\}^r\), a
\emph{complementary-pair signature} (\(\CP\)) is a signature
\begin{equation*}
    c_0\delta_{\mathbf a}+c_1\delta_{\bar {\mathbf a}},
    \qquad c_0,c_1\in\AlgNums^\times.
\end{equation*}
When \({\mathbf a}=\zero^r\), it is a \emph{generalized equality} (\(\GE\)).
As a tensor identity, individual bit flips turn a CP signature into a GE signature.
The notation \(\diag(p,q)\) denotes the \(2\)-by-\(2\) diagonal matrix with diagonal entries \(p,q\). For \(t\in\AlgNums^\times\), put \(D(t):=\diag(1,t)\).  For \(p,q\in\AlgNums\), in arity two,
\(p\delta_{00}+q\delta_{11}=\diag(p,q)\) is a \emph{weighted binary
equality}, and \(p\delta_{01}+q\delta_{10}\) is a \emph{weighted binary
disequality}.  Either is nonsingular when \(pq\ne0\).

The notation \(\Span S\) denotes linear span.  An affine subspace is a
	translate \({\mathbf a}+L\) of a linear subspace. Its translation space is \(L\), and an
	affine plane has dimension two.

For a nonzero vector \(\mathbf v\) over a field \(F\), write
\([\mathbf v]=F^\times\mathbf v\) for its projective class.  We likewise
write \([A]=F^\times A\) for a nonzero tensor or matrix \(A\).  In particular,
\(\PGL_n(F)=\GL_n(F)/(F^\times I_n)\).
For tensors or matrices \(A,B\) over \(\AlgNums\), write \(A\doteq B\)
when \(A=\lambda B\) for some \(\lambda\in\AlgNums^\times\).
In gadget constructions we suppress overall nonzero algebraic factors.
Their contribution to a Holant value can be removed by exact algebraic
postprocessing.  This convention does not suppress relative coefficients
in sums of signatures.
Angle brackets denote generated subgroup in a group context and tensor closure for a signature class.

\subsection{Holographic transformations and the
  \texorpdfstring{\(K\)}{K}-basis}
\label{subsec:holographic}

In \(\Holant(\cG\mid\cF)\), represent each left-side signature \(g\)
of arity \(r\) as a row vector and each right-side signature \(f\) of
arity \(s\) as a column vector, both in lexicographic order.  For a fixed
\(M\in\GL_2(\AlgNums)\), replace them by
\[
gM^{\otimes r},\qquad (M^{-1})^{\otimes s}f,
\]
respectively.  The matrices on the two ends of each edge cancel by
\(MM^{-1}=I\), so this replacement preserves the value of every instance.
Write \(gM:=gM^{\otimes r}\) for the left-side action and
\(Mf:=M^{\otimes s}f\) for the right-side action, and put
\[
\cG M:=\{gM:g\in\cG\},\qquad M\cF:=\{Mf:f\in\cF\}.
\]
The value-preserving transformation in both directions gives
\[
\Holant(\cG\mid\cF)\equivT\Holant(\cG M\mid M^{-1}\cF).
\]
For ordinary Holant, take \(\cG=\{\Eq_2\}\).  The transformed row vector
\((1,0,0,1)M^{\otimes2}\) represents the binary signature with matrix
\(M^{\mathsf T}M\), since its \((a,b)\)-entry is
\(\sum_{x\in\Ftwo}M_{xa}M_{xb}\).  Thus
\[
\Holant(\cF)=\Holant(\Eq_2\mid\cF)
\equivT\Holant(M^{\mathsf T}M\mid M^{-1}\cF).
\]

Fix the K-coordinate matrix
	\begin{equation}
		K=\frac1{\sqrt2}
		\begin{pmatrix}1&1\\ i&-i\end{pmatrix},
		\qquad
		K^{-1}=\frac1{\sqrt2}
		\begin{pmatrix}1&-i\\1&i\end{pmatrix}.
		\label{eq:K-matrix}
	\end{equation}
	Since \(K^{\mathsf T}K=X\), putting
	\(\widehat{\cF}=K^{-1}\cF\) gives
	\begin{equation}
		\Holant(\cF)
		\equivT
		\Holant(X\mid\widehat{\cF})
		=:\KHolant(\widehat{\cF}).
		\label{eq:K-equivalence}
	\end{equation}
	For an original-coordinate signature \(f\), we likewise write
	\(\widehat f=K^{-1}f\). The hat records passage to \(K\)-coordinates.
	Once we work with a signature set already in \(K\)-coordinates, we suppress
	these hats and denote its member signatures simply by \(f\).  Unless stated
	otherwise, signatures and signature sets appearing in a K-Holant problem
	are understood to be in \(K\)-coordinates, with edge \(X\).

Write \(Z_X(\Omega)\) for the partition function of a closed grid
	whose edges carry \(X\).  To describe simultaneous changes of basis
	that preserve this edge up to scalar, define the general orthogonal
	similitude group of \(X\) by
	\begin{equation*}
		GO(X)
		=
		\{D\in\GL_2(\AlgNums):D^{\mathsf T}XD=\lambda X
		\text{ for some }\lambda\in\AlgNums^\times\}.
	\end{equation*}
	If \(E(\Omega)\) is the edge set,
	\(D\Omega\) applies \(D\) at every vertex, and
	\(D^{\mathsf T}XD=\lambda X\), then
	\begin{equation}
		Z_X(D\Omega)=\lambda^{|E(\Omega)|}Z_X(\Omega).
		\label{eq:GOX-global-scalar}
	\end{equation}
	The known edge-count scalar gives a polynomial-time equivalence.  This
	operation transforms the whole signature set. It does not by itself
	realize \(D\) as a local gadget.

\subsection{Gadgets and reductions}
\label{subsec:gadget-proof-vocabulary}
\label{subsec:reduction-conventions}

A gadget is a finite network over the allowed signatures, with ordered,
labelled dangling edges.  Its signature is obtained by summing over the
internal variables, including the edge signatures in the product.  Edges
carry equality in ordinary coordinates and disequality in \(K\)-coordinates.
A \emph{direct} construction uses one fixed gadget.  Below we work
in \(K\)-coordinates.

A slice fixes some arguments of a tensor. Implementing it by pinning
requires the corresponding pinning signatures. Interpolation gives an
oracle reduction by solving a linear system from exact evaluations.
Formal tensor operations make a signature available only when a gadget or
a stated reduction supplies it. Whole-set changes of basis act on the
problem as described in \cref{subsec:holographic}.

Using bit flips in a gadget or reduction requires available signatures or
a stated whole-set transformation. Algebraic support analysis may also
use the relabellings defined above.

\paragraph{Equivalence convention.}
Let \(\cF\) be the current finite signature set.  Whenever a fixed
gadget construction, polynomial-time interpolation, or a cited reduction
supplies a finite set \(S\) such that
\begin{equation}
\KHolant(\cF\cup S)\leT\KHolant(\cF),
\label{eq:finite-equivalent-enlargement}
\end{equation}
the reverse reduction is immediate by restriction.  We call the signatures
in \(S\) \emph{available by reduction over \(\cF\)}.  This does not assert
membership in the original set or direct gadget realizability.
When an argument replaces the current set by an equivalent set and retains
the notation \(\cF\), that replacement is stated explicitly.
The same convention applies to ordinary Holant when
\(\Holant(\cF\cup S)\leT\Holant(\cF)\).
A use that depends on local topology or variable order still requires an
explicit direct realization over the stated signature set.  Overall nonzero
scalars are handled as in \cref{subsec:algebra-group-terminology}.
These are the gadget and interpolation conventions of
\cite[Section~2.5]{Liu2026DEQ2Holant} and
\cite[Section~2.3]{MengWangXiaZheng2025}.

If a direct \(X\)-edge gadget over \(\cF\) realizes \(h\), replacing
each vertex carrying \(h\) by that gadget gives
\begin{equation}
\KHolant(\cF\cup\{h\})\leT\KHolant(\cF).
\label{eq:direct-gadget-reduction}
\end{equation}
A fixed gadget has constant size.  If it realizes \(\rho h\), with
\(\rho\in\AlgNums^\times\), an instance with \(N\) occurrences is
rescaled by \(\rho^{-N}\).  Gadget replacement is thus one case
of \cref{eq:finite-equivalent-enlargement}.

Let \(f\) have arity \(r\).  For an ordered pair of distinct variables
	\(p=(i,j)\), a binary kernel
	\(C:\Ftwo^2\to\AlgNums\), and a residual string \({\mathbf z}\) indexed by
	\([r]\setminus\{i,j\}\) in increasing variable order, let
	\({\mathbf z}^{i\leftarrow a,j\leftarrow b}\) be the length-\(r\) string obtained by
	inserting \(a,b\) at positions \(i,j\).  Define the contraction of two variables by
	\begin{equation*}
		(\partial_p^C f)({\mathbf z})
		=
		\sum_{a,b\in\Ftwo}C(a,b)
		f\bigl({\mathbf z}^{i\leftarrow a,j\leftarrow b}\bigr).
	\end{equation*}
We also write \(\partial_{ij}^C f=\partial_{(i,j)}^C f\).

In the notation of Shao--Cai
	\cite[full version, Section~2.4]{ShaoCai2020}, their
	\(\widehat\partial_{ij}\widehat f\) and \(\partial_{ij}f\) are,
	respectively, our \(\partial_{(i,j)}^X\widehat f\) and
	\(\partial_{(i,j)}^I f\).

Let an ordered binary signature \(B\) be directly realized over the current
signature set.  Connecting its two inputs through \(X\)-edges to variables
\(i,j\) of \(f\), in that order, gives the contraction kernel
\(C=XBX\), with \(C(a,b)=B(1-a,1-b)\).  Swapping the inputs of \(B\)
gives the kernel \(XB^{\mathsf T}X\).

We also write
	\(f_{ij}^{ab}({\mathbf z})=f({\mathbf z}^{i\leftarrow a,j\leftarrow b})\) for the residual
	slice.  If an ordered binary \(B\) is placed on the ordered variable pair
	\(e=(u,v)\), we write \(B_e\).

For tensors \(A,B\) on the same ordered variables, define the bilinear
coefficient pairing
\[
\langle A,B\rangle=\sum_{\mathbf x}A({\mathbf x})B({\mathbf x}),
\]
without conjugation. For binary matrices this is the Frobenius pairing
\(\langle U,V\rangle=\operatorname{tr}(U^{\mathsf T}V)\).
Two tensors are \emph{orthogonal} for this pairing when it is zero.

For any \(T\in\Mat_2(\AlgNums)\), define
\[
T^\#:=XT^{\mathsf T}X.
\]
For an ordered binary signature \(B\), its \emph{transfer} is \(T_B:=BX\).
Swapping the two inputs of \(B\) transposes its matrix, so the resulting
transfer is
\[
T_{B^{\mathsf T}}=B^{\mathsf T}X=X(T_B)^{\mathsf T}X=(T_B)^\#.
\]
Joining the second input of \(B\) to the first input of an ordered binary
signature \(C\) through an \(X\)-edge gives \(BXC\).  Hence
\[
T_{BXC}=(BXC)X=(BX)(CX)=T_BT_C.
\]
For all \(S,T\in\Mat_2(\AlgNums)\), we have
\((ST)^\#=T^\#S^\#\) and \((T^\#)^\#=T\).
For \(T\in\GL_2(\AlgNums)\), the \emph{projective order} of \([T]\)
is the least \(m\ge1\) such that \(T^m=\lambda I\) for some
\(\lambda\in\AlgNums^\times\), and is infinite if no such \(m\) exists.

	\section{The Tractability Criterion}
	\label{sec:predicate}
	
The criterion in \cref{thm:main} is a disjunction of seven conditions on
the whole signature set.  The first two use support restrictions and EO
signatures, the next two use tensor decompositions, and the last three
require a common change of basis for the edge and all vertex signatures.
We define these conditions and then state their decidability and
polynomial-time evaluation result.  The hardness direction is proved in
\cref{sec:assembly}.

	\subsection{Charge and Hamming layers}
	\label{subsec:charge}
	
	All charge terminology in this subsection is in \(K\)-coordinates.
	For an integer \(r\ge0\), define
	the doubled-charge function \(q_r:\Ftwo^r\to\mathbb Z\) by
	\[
	q_r({\mathbf x})=2\wt({\mathbf x})-r.
	\]
	For an even integer \(r\), we also use the integer half-charge
	\[
	\chi_r({\mathbf x})=\frac{q_r({\mathbf x})}2=\wt({\mathbf x})-\frac r2.
	\]
	
	Put
	\[
	HW_r^{\ge}=\{{\mathbf x}\in\Ftwo^r:q_r({\mathbf x})\ge0\},
	\quad
	HW_r^{\le}=\{{\mathbf x}\in\Ftwo^r:q_r({\mathbf x})\le0\},
	\quad
	HW_r^{=}=\{{\mathbf x}\in\Ftwo^r:q_r({\mathbf x})=0\},
	\]
	and define \(HW_r^{>},HW_r^{<}\) by strict inequalities.  We omit the
	subscript when the arity is clear.  A signature has one of these properties when its support lies in the
	corresponding set.  A signature is \emph{single-layer} when its support lies
	in one Hamming layer. A signature set is single-layer when each member is so
	supported, with member-dependent layers allowed, and is \emph{multilayer}
	otherwise.  A signature set is \emph{common
		one-sided} when all members are \(HW^{\ge}\) or all are \(HW^{\le}\).
	Otherwise it is \emph{two-sided}, equivalently both charge signs occur among
	support strings of its members.  Here \(\EO\) abbreviates Eulerian orientation.  An
	arity-\(2d\) \emph{EO signature} is supported in \(HW_{2d}^=\). This central
	support condition is the local balance constraint for an Eulerian
	orientation.
	
	For a closed \(X\)-edge grid \(\Omega\) with vertex set \(V(\Omega)\), every edge
	contributes one \(0\)-bit and one \(1\)-bit.  Hence each contributing
	assignment \({\mathbf x}\), with induced local string \({\mathbf x}_v\) at vertex \(v\), satisfies
	\begin{equation*}
		\sum_{v\in V(\Omega)}q_{\arity(f_v)}({\mathbf x}_v)=0.
	\end{equation*}
	For a gadget with dangling edges, the same cancellation occurs on internal
\(X\)-edges: the sum of local charges equals the charge of the external
string.  This is the charge conservation used in contractions.  In a closed
instance over a common one-sided set, all local charges have the same weak
sign and sum to zero.  Hence every contributing assignment uses only
central-layer entries. Setting all other entries to zero therefore
preserves the Holant value.
	
	\begin{remark}[Notation relative to prior work]
		\label{rem:mwxz-notation}
		We follow \cite[Sections~2.1--2.3]{MengWangXiaZheng2025} in writing
		\(\widehat{\cF}=K^{-1}\cF\), \(HW^{\ge}\), \(HW^{\le}\), \(HW^=\),
		and \emph{EO signature}, as well as their equality and disequality
        notation \(\Eq_r\) and \(\Neq_2\).
		
	\end{remark}
	
	\subsection{The standard signature classes}
	\label{subsec:standard-classes}
	
	Zero signatures are included whenever their arity is permitted by the
	class.  The affine class
	\(\cA\) consists of the zero signature and the nonzero arity-\(r\) signatures
	\begin{equation}
		f({\mathbf x})=\lambda\,\one[A{\mathbf x}={\mathbf b}]\,i^{Q({\mathbf x})},
		\label{eq:affine-class}
	\end{equation}
	where \(r,m\ge0\) are integers,
	\(\lambda\in\AlgNums^\times\),
	\(A\in\Mat_{m\times r}(\Ftwo)\), \({\mathbf b}\in\Ftwo^m\), and
	\(Q:\Ftwo^r\to\mathbb Z_4\), where
	\(\mathbb Z_4=\mathbb Z/4\mathbb Z\), is represented by a multilinear
	polynomial over \(\mathbb Z_4\) of degree at most two whose mixed
	coefficients are even.
	Equivalently, the phase is a product of fourth-root phases of affine linear
	forms interpreted as \(\{0,1\}\)-valued integers.
	Here a \emph{mixed coefficient} is the coefficient of \(x_ix_j\) for
	\(i\ne j\), and an affine linear form over \(\Ftwo\) is
	\(a_0+\sum_{j=1}^r a_jx_j\) with all coefficients in \(\Ftwo\).
	
	The product class \(\cP\) consists of zero and all signatures of the form
	\(\lambda\prod_{j=1}^m h_j\), where \(m\ge0\) is an integer,
	\(0\ne\lambda\in\AlgNums\), and every factor
	\(h_j\) is a unary, binary equality, or binary disequality constraint applied
	to specified signature variables.  This is a pointwise product of
	constraints that may share variables, whereas tensor products use disjoint
	variables.  Empty products and repeated variables are allowed. An inconsistent
	product denotes the zero signature.
	
	Put \(\zeta_8:=e^{\pi i/4}=(1+i)/\sqrt2\).
	The local-affine class \(\cL\) consists of the signatures \(f\) such that
	\begin{equation*}
		\left(\bigotimes_{j=1}^{\arity(f)}D(\zeta_8)^{\alpha_j}\right)f\in\cA
		\qquad\text{for every }{\boldsymbol{\alpha}}\in\supp(f).
	\end{equation*}
	Each support assignment \({\boldsymbol{\alpha}}\) specifies a phase change at every
	variable, and the condition tests all of these choices on the original support.
	
	Let \(\cT\) be the set of all unary and binary signatures, and let
	\begin{equation*}
		\cM=\{f:\text{\(f\) is a Boolean signature and }
		f({\mathbf x})=0\text{ for every \({\mathbf x}\) with }\wt({\mathbf x})>1\}.
	\end{equation*}

	\begin{definition}[Transformability]
		\label{def:common-transformability}
		\label{def:common-K-presentation}
		Let \(\cC\) be closed under nonzero scalars and variable permutations.  A
		finite signature set \(\cF\) is \(\cC\)-transformable if one
		\(M\in\GL_2(\AlgNums)\) satisfies
		\begin{equation*}
			M^{\mathsf T}M\in\cC,\qquad M^{-1}\cF\subseteq\cC.
		\end{equation*}
		Equivalently, for \(\widehat{\cF}=K^{-1}\cF\), one matrix
		\(U=K^{-1}M\) satisfies
		\begin{equation*}
			U^{\mathsf T}XU\in\cC,\qquad U^{-1}\widehat{\cF}\subseteq\cC.
		\end{equation*}
		Thus one matrix must satisfy the edge condition and the conditions for
		all signatures.  The tractability criterion uses this definition for \(\cC\in\{\cA,\cP,\cL\}\).
	\end{definition}

The equivalence in \cref{def:common-transformability} follows directly
from \(K^{\mathsf T}K=X\).  For \(U=K^{-1}M\),
\[
U^{\mathsf T}XU=M^{\mathsf T}M,
\qquad U^{-1}\widehat{\cF}=M^{-1}\cF.
\]
Conversely, \(M=KU\), which proves the equivalence.

	\subsection{The Eulerian-orientation predicate}
	\label{subsec:eo-predicate}
	
	The first two tractable cases use two different operations on
	signatures in \(K\)-coordinates.  The central restriction keeps the arity
	and retains only the zero-charge entries.  For an arity-\(r\) signature
	\(g\), define
	\begin{equation*}
		g|_{\EO}({\mathbf x})=g({\mathbf x})\one[q_r({\mathbf x})=0].
	\end{equation*}
	It is zero for odd \(r\), and zero restrictions are deleted.  For a
	signature set \(\cF\), set
	\begin{equation*}
		\cF|_{\EO}=\{g|_{\EO}:g\in\cF,\ g|_{\EO}\ne0\}.
	\end{equation*}
	  When applying the
	central restriction to a closed instance over a common one-sided set,
	an occurrence with zero central restriction makes the instance value
	zero.  The deletion above is a convention for the signature set.
	
	Let \(g\ne0\) have arity \(k\) and be supported on the single Hamming
	layer of weight \(d\).  To balance its charge while retaining its
	coefficients, its \emph{EO completion} appends pinning signatures in the
	following order:
	\begin{equation}
		g_{\to\EO}
		=
		\begin{cases}
			g\otimes\DeltaZero^{\otimes(2d-k)},
			& 2d\ge k,\\
			g\otimes\DeltaOne^{\otimes(k-2d)},
			& 2d<k.
		\end{cases}
		\label{eq:eo-completion}
	\end{equation}
	The original variables precede the added unaries. A zeroth power is the nullary
	unit.  The completed signature has even arity and central support, hence is
	EO\@.  For a signature set of nonzero single-layer signatures, write
	\begin{equation*}
		\cF_{\to\EO}=\{g_{\to\EO}:g\in\cF\}.
	\end{equation*}
  In the
evaluation reduction, a component whose fixed local charges do not sum
to zero contributes zero.  Otherwise, the appended zero and one variables are paired by
\(X\)-edges within each charge-balanced component, giving a unique
extension of each original assignment
\cite[Section~6.2, proof of Corollary~68]{MengWangXia2025EO}.  This construction uses
the explicit completion tensors in an auxiliary EO instance. It requires
no assumption that the input signature set realizes the added pinnings.
	
	For arity-\(2d\) EO signature \(g\), an oriented perfect pairing
	\(\pi=((a_1,b_1),\ldots,(a_d,b_d))\) has unordered pairs partitioning
	\([2d]\).  For \({\mathbf z}\in\Ftwo^d\), let \({\mathbf y}\in\Ftwo^{2d}\) be the unique string
	with \(y_{a_j}=z_j\) and \(y_{b_j}=1-z_j\). The pairing induces
	\begin{equation*}
		g^\pi(z_1,\ldots,z_d)=g({\mathbf y}).
	\end{equation*}
	
	\begin{definition}[The predicate \(\EOTract\)]
		\label{def:eo-tract}
		For a finite EO signature set \(\cF\), \(\EOTract(\cF)\) holds if there exist a
		sign \(\varepsilon\in\{+1,-1\}\) and a class
		\(\cC\in\{\cA,\cP\}\), both shared by the entire signature set, such that:
		\begin{enumerate}
			\item for every \(g\in\cF\) and all
			\({\boldsymbol{\alpha}},{\boldsymbol{\beta}},{\boldsymbol{\gamma}}\in\supp(g)\), if
			\({\boldsymbol{\delta}}={\boldsymbol{\alpha}}\oplus{\boldsymbol{\beta}}\oplus{\boldsymbol{\gamma}}\), then
			\({\boldsymbol{\delta}}\in\supp(g)\) or
			\(\varepsilon q_{\arity(g)}({\boldsymbol{\delta}})>0\).
			\item for every \(g\in\cF\) and every oriented perfect pairing \(\pi\)
			of its variables, \(g^\pi\in\cC\).
		\end{enumerate}
        In particular, \(\EOTract(\varnothing)\) holds.
	\end{definition}
	
The common-sign condition follows \cite[Theorem~26]{MengWangXia2025EO}.
The compressed signature \(g^\pi\) expresses the opposite-pair restriction
used to define the affine and product-type EO classes in
\cite[Definition~22]{MengWangXia2024TypicalEO}.  To compare the two
forms, reorder the original variables as
\((a_1,\ldots,a_d,b_1,\ldots,b_d)\), with assignments \(({\mathbf z},{\mathbf w})\).
The full restriction, extended by zero away from opposite pairs, is
\[
R_\pi g({\mathbf z},{\mathbf w}):=g({\mathbf z},{\mathbf w})\prod_{j=1}^dX(z_j,w_j)
=g^\pi({\mathbf z})\prod_{j=1}^dX(z_j,w_j).
\]
For \(\cC\in\{\cA,\cP\}\), adding these disequality factors gives
\(g^\pi\in\cC\Rightarrow R_\pi g\in\cC\).  Conversely,
\(\sum_{{\mathbf w}\in\Ftwo^d}R_\pi g({\mathbf z},{\mathbf w})=g^\pi({\mathbf z})\): each factor
\(X(z_j,w_j)\) admits exactly one value of \(w_j\).
The affine and product normal forms are preserved under this elimination,
so \(R_\pi g\in\cC\Rightarrow g^\pi\in\cC\), as detailed in
\cite[Section~2.10.2]{Liu2026DEQ2Holant}.
Reversing a pair flips the corresponding compressed variable, and
reordering pairs permutes the compressed variables.  Both operations
preserve \(\cA\) and \(\cP\), so using oriented pairings retains the
source condition and its single shared class.

	\subsection{The tractability predicate}
	\label{subsec:tract-predicate}
	
	For an arbitrary finite signature set, evaluate \(\Tract\) after the
	zero and nullary preprocessing of \cref{subsec:zero-nullary}.
	For a preprocessed signature set \(\cF\), put
	\(\widehat{\cF}=K^{-1}\cF\).  Define \(\Tract(\cF)\) to be the disjunction
	of the following seven alternatives, which may overlap.  The first two
test the \(K\)-coordinate set \(\widehat{\cF}\). The remaining five are
stated for the original-coordinate set \(\cF\).
Transformability and the EO condition are defined in
\cref{def:common-transformability,def:eo-tract}, respectively.
	
	\begin{enumerate}[label=\textnormal{(\arabic*)}]
		\item \label{case:one-sided}
		The signature set \(\widehat{\cF}\) is common one-sided, that is,
		\[
		[\,\forall\widehat f\in\widehat{\cF},\
		\supp(\widehat f)\subseteq HW_{\arity(\widehat f)}^{\ge}\,]
		\quad\text{or}\quad
		[\,\forall\widehat f\in\widehat{\cF},\
		\supp(\widehat f)\subseteq HW_{\arity(\widehat f)}^{\le}\,],
		\]
		and \(\EOTract(\widehat{\cF}|_{\EO})\) holds.
		
		\item \label{case:single-weight}
		Every \(\widehat f\in\widehat{\cF}\) is supported on one Hamming
		weight \(d_f\).  Across the signature set, some support string has negative
		charge and some (possibly other) support string has positive charge, and
		\(\EOTract(\widehat{\cF}_{\to\EO})\) holds, with the completion
		operation defined in \cref{eq:eo-completion}.
		
		\item \label{case:arity-two}
		\(\cF\subseteq\angles{\cT}\).
		
		\item \label{case:matching}
		\(\cF\subseteq\angles{K\cM}\), or
		\(\cF\subseteq\angles{KX\cM}\).
		
		\item \label{case:affine}
		\(\cF\) is \(\cA\)-transformable.
		
		\item \label{case:product}
		\(\cF\) is \(\cP\)-transformable.
		
		\item \label{case:local-affine}
		\(\cF\) is \(\cL\)-transformable.
	\end{enumerate}
	
In \cref{case:matching}, one of the two inclusions must hold for the
entire signature set.

	\begin{remark}
        For a nonzero positive-arity signature \(f\),
		\[
		\begin{aligned}
			f\in\angles{K\cM}
			&\Longleftrightarrow
			\text{ every irreducible factor of }K^{-1}f\text{ lies in }\cM,\\
			f\in\angles{KX\cM}
			&\Longleftrightarrow
			\text{ every irreducible factor of }XK^{-1}f\text{ lies in }\cM.
		\end{aligned}
		\]
		Zero and nullary signatures are covered by the tensor-closure
		conventions of \cref{subsec:algebra-group-terminology} and by the
		preprocessing in \cref{subsec:zero-nullary}.
	\end{remark}
	
	\subsection{Decidability and polynomial-time evaluation}
	\label{subsec:predicate-algorithms}

	\begin{theorem}[Recognition and evaluation]
		\label{thm:predicate-algorithms}
		For every finite algebraic signature set \(\cF\), the predicate \(\Tract(\cF)\),
		with the preprocessing convention of \cref{subsec:zero-nullary}, is
		decidable.  If it holds, then \(\Holant(\cF)\in\FP\).
	\end{theorem}
	
	\begin{proof}
		This is \cite[Theorem~2.22]{Liu2026DEQ2Holant}.
	\end{proof}

	\section{Auxiliary Reductions}
	\label{sec:mountain-terminals}
	
The height arguments use the reductions in this section to turn particular
gadget outputs into a classification of the whole signature set.  Factor
extraction and binary interpolation supply additional signatures, while
access to a nonzero unary gives a dichotomy.  The EO and equality results
then handle the corresponding outputs of the height constructions.
Availability by reduction is understood as in
\cref{eq:finite-equivalent-enlargement}.

	\subsection{Factor Extraction and Preservation of Tractability}
\label{subsec:dichotomy-decomposition-tools}
\label{sec:retained-normalization}

The support alternatives of \cref{subsec:tract-predicate} determine when
factors may be added to the current signature set.  We first state the EO
classification needed for the common-one-sided exception, then record factor
extraction and preservation of the tractability criterion under the reductions.
\begin{theorem}[Weighted Eulerian orientations
		{\cite[Theorem~26]{MengWangXia2025EO}, \cite[Theorem~1.1]{GuanShaoShi2026}}]
		\label{thm:external-eo}
		For every finite algebraic EO signature set \(\cF\), if \(\EOTract(\cF)\) holds,
		then \(\KHolant(\cF)\in\FP\). Otherwise
		\(\KHolant(\cF)\) is \(\SharpP\)-hard.
	\end{theorem}
	
	\begin{theorem}[Tensor-factor extraction
		{\cite[full version, Theorem~32 and Remark~39; conference version,
			Theorem~29]{MengWangXiaZheng2025}}]
		\label{thm:external-factor}
		Let \(\cF\) be a finite signature set and let
		\(h=a\otimes b\) be a nonzero tensor product.  Unless
		\(\cF\cup\{h\}\) is common one-sided,
		\[
		\KHolant(\cF,a,b)\equivT\KHolant(\cF,h).
		\]
	\end{theorem}

In the common-one-sided exception, \(\KHolant(\cF,h)\) is in \(\FP\)
		when \(\EOTract((\cF\cup\{h\})|_{\EO})\) holds and is
		\(\SharpP\)-hard otherwise, by charge conservation and
		\cref{thm:external-eo}.
	
	\subsubsection{Irreducible factors}
	\label{subsec:factor-saturation}

The common-one-sided case is already classified by \cref{thm:external-eo}.
In the remaining case, \cref{thm:external-factor} permits tensor-factor
extraction, so the analysis of individual signatures reduces to their
irreducible factors.  The iterations below also retain previously
constructed signatures.  For those arguments, we use the following finite
closure.
	
	For a preprocessed finite signature set \(\cF\), choose an exact
	irreducible factorization of every member and let \(\Pi(\cF)\) be the
	finite set of all irreducible factors, with their ordered variable blocks.  We call
	\(\cF\) \emph{factor-closed} when \(\Pi(\cF)\subseteq\cF\).
	
	\begin{lemma}[Factor closure]
		\label{lem:factor-saturation}
		If \(\cF\) is not common one-sided, then
		\begin{equation}
			\KHolant(\cF)
			\equivT
			\KHolant(\cF\cup\Pi(\cF)).
			\label{eq:factor-saturation}
		\end{equation}
	\end{lemma}
	
	\begin{proof}
		Iterate \cref{thm:external-factor}, retaining the original set at each
		step so that the common-one-sided exception cannot occur.
		This is the argument of \cite[Lemma~2.23]{Liu2026DEQ2Holant}.
	\end{proof}

Repeated factorization terminates because every proper factor has smaller
arity than its parent.  Each original signature can be reconstructed as
the tensor product of its factors, with the recorded scalar and variable
order.  Substituting these products in an instance therefore gives
\[
\KHolant(\cF)\equivT\KHolant(\Pi(\cF))
\]
outside the common-one-sided case.  Thus initial preprocessing may replace
the signature set by its irreducible factors, still denoted by \(\cF\).
This assumption applies to the initial set: later gadget outputs may be
reducible and must be processed again.

Some later iterations retain both earlier signatures and their factors.
When such an iteration requires factor closure, we instead use
\(\cF\cup\Pi(\cF)\), retaining \(\cF\) for the current set after this
explicit replacement.
	
	\subsubsection{Taking subsets and changing basis}

The next two lemmas transfer the tractability criterion through taking
subsets and through a common change of basis.
	
	\begin{lemma}[Preservation under taking subsets]
		\label{lem:downward-inheritance}
		If \(\cF\subseteq\cF'\) are preprocessed signature sets and
		\(\Tract(K\cF')\) holds, then \(\Tract(K\cF)\) holds.
	\end{lemma}
	
	\begin{proof}
		This is \cite[Lemma~2.24]{Liu2026DEQ2Holant}.
	\end{proof}
	
	\begin{lemma}[\(GO(X)\)-invariance of tractability]
		\label{lem:GOX-tract-invariance}
		Let \(\cF\) be a finite preprocessed signature set and let
		\(D\in GO(X)\).  Then
		\begin{equation*}
			\Tract(K\cF)
			\quad\Longleftrightarrow\quad
			\Tract\bigl(K(D\cF)\bigr).
		\end{equation*}
	\end{lemma}
	
	\begin{proof}
		This is \cite[Lemma~2.25]{Liu2026DEQ2Holant}.
	\end{proof}

	\subsection{Binary Interpolation}
	\label{proofsubsec:central-descent}

For an ordered binary \(B\), the transfer \(T_B=BX\) records composition
along \(X\)-edges, as in \cref{subsec:reduction-conventions}.
The lemma below has two ways to obtain a unary: the singular case uses
factor extraction, while the infinite-order case first supplies a rank-one
binary by interpolation.  Both conclusions are reductions over the
original signature set.
	
	\begin{lemma}[Binary interpolation]
		\label{lem:binary-interpolation}
		Let \(\cF\) be a finite algebraic signature set whose supported doubled
		charges include both signs, let
		\(0\ne B\in\cF\) be ordered binary, with transfer \(T_B=BX\).
		\begin{enumerate}
			\item If \(B=u\otimes v\) is singular, where \(u,v\ne0\), then
			\[
			\KHolant(\cF,u,v)\leT\KHolant(\cF).
			\]
			\item If \(T_B\) is nonsingular and has infinite projective order, there is
			a nonzero rank-one binary \(H\) such that
			\[
			\KHolant(\cF,H)\leT\KHolant(\cF).
			\]
		\end{enumerate}
	\end{lemma}
	
	\begin{proof}
		This is \cite[Lemma~3.1]{Liu2026DEQ2Holant}.
	\end{proof}
	
	\subsection{Single-occurrence Reductions and Nonzero Unaries}
	\label{subsec:single-occurrence-carriers}

	\begin{theorem}[Odd-arity Holant
		{\cite[full version, Theorem~33; conference version, Theorem~30]{MengWangXiaZheng2025}}]
		\label{thm:external-odd}
		Let \(\cF\) be a finite preprocessed algebraic signature set containing an
		odd-arity signature.  Then \(\Holant(\cF)\in\FP\) if
		\(\Tract(\cF)\) holds, and is \(\SharpP\)-hard otherwise.
	\end{theorem}
	
	Its first two alternatives were originally in \(\FP^{\mathsf{NP}}\). The
	stated \(\FP\) conclusion follows from
	\cite[Theorem~1.1 and Corollary~1.2]{GuanShaoShi2026}.

The single-occurrence replication argument below originated in an
expanded version of \cite{MengWangXiaZheng2025} that is not publicly
available.  The version we use is also stated in
\cite[Lemma~3.3]{Liu2026DEQ2Holant}. We reproduce it with its ordered
variable connections explicit.

If a nonzero odd-arity signature is realized or made available by a
reduction, this dichotomy settles that case.  The same holds in
\(K\)-coordinates by \cref{eq:K-equivalence}.  After zero and nullary
preprocessing, the remaining analysis may therefore assume that every
signature in \(\cF\) has even arity.  In particular, there are no unaries.
	Write
	\[
	\Delta=\DeltaZero\otimes\DeltaOne,
	\]
	with the displayed order of its two variables.  For a signature \(g\), the
	superscripts \(g^{=1}\) and \(g^{\leq1}\) mean that an oracle instance contains,
	respectively, exactly one or at most one \(g\)-vertex.
	
	\begin{lemma}[Reduction with at most one \(\Delta\)-vertex]
		\label{lem:universal-directed-pair}
		For every fixed finite algebraic signature set \(\cF\),
		\[
		\KHolant(\cF,\Delta^{\leq1})
		\leT
		\KHolant(\cF).
		\]
	\end{lemma}
	
	\begin{proof}
		This is \cite[Lemma~3.2]{Liu2026DEQ2Holant}.
	\end{proof}
	
The next lemma turns one-occurrence access to \(g\) into unrestricted
access to \(h\).  Each copy of a fixed gadget propagates the boundary
signature \(g\) and exposes an output block with signature \(h\).
A fixed nonzero evaluation then closes
the retained end of the chain, as illustrated in
\cref{fig:single-occurrence-chain}.
The lemma applies separately to ordinary Holant and to K-Holant.

	\begin{lemma}[Single-occurrence replication]
		\label{lem:single-occurrence-replication}
		Let \(\cF\) be a fixed finite algebraic signature set, and let \(g\) and
		\(h\) be nonzero signatures of arities \(r\) and \(s\), respectively.  For ordinary Holant,
		all connections below use \(\Eq_2\).  Suppose that
		\begin{enumerate}
			\item \(\Holant(\cF,g^{=1})\leT\Holant(\cF)\).
			\item there is a fixed gadget \(C\) over \(\cF\) whose dangling edges are
			the three disjoint ordered blocks
			\[
			A=(a_1,\ldots,a_r),\qquad
			P=(p_1,\ldots,p_r),\qquad
			Y=(y_1,\ldots,y_s).
			\]
			Connecting one \(g\)-vertex coordinatewise to \(A\) realizes
			\[
			\rho\,g(P)h(Y)
			\]
			on \(P,Y\), for a known \(\rho\in\AlgNums^\times\).  Moreover, for
			consecutive copies \(C_i,C_{i+1}\), the block \(P_i\) can be connected
			coordinatewise to \(A_{i+1}\), with no other connections.
			Thus a chain of \(k\) copies, beginning with one \(g\)-vertex,
			contains exactly that one \(g\)-vertex and realizes
			\[
			\rho^k g(P_k)\prod_{i=1}^k h(Y_i),
			\]
			where the blocks \(Y_1,\ldots,Y_k\) are pairwise disjoint.
			\item there is a fixed instance \(J\) over \(\cF\cup\{g\}\), containing
			exactly one \(g\)-vertex, whose exactly computable partition value
			\(Z(J)\) is nonzero.
		\end{enumerate}
		Then
		\[
		\Holant(\cF,h)\leT\Holant(\cF).
		\]
		For K-Holant, if all connections use \(X\) and assumption~1 is
		\(\KHolant(\cF,g^{=1})\leT\KHolant(\cF)\), then the same
		gadget and nonzero-evaluation hypotheses give
		\[
		\KHolant(\cF,h)\leT\KHolant(\cF).
		\]
	\end{lemma}
	
	\begin{proof}
		This is \cite[Lemma~3.3]{Liu2026DEQ2Holant}.  An input with no \(h\)-vertex
		is already an instance over \(\cF\).  Let \(\Omega\) be an input with
		\(k\ge1\) occurrences of \(h\).  Delete those vertices, retain their incident
		edges as \(k\) ordered boundary blocks, and attach the blocks
		\(Y_1,\ldots,Y_k\) of the displayed chain coordinatewise in the corresponding
		orders.  Let \(J^\circ\) be the network obtained by deleting the
		unique \(g\)-vertex of \(J\), and join its boundary coordinatewise to
		\(P_k\) in the variable order inherited from that removed vertex.  Denote the
		resulting instance by \(\Omega'\).  It contains exactly one
		\(g\)-vertex, and the two tensor identities give
		\[
		Z(\Omega')=\rho^k Z(J)Z(\Omega).
		\]
		Divide by the known nonzero scalar \(\rho^kZ(J)\) and apply assumption~1.
The gadgets \(C\) and \(J\) and the scalars \(\rho\) and \(Z(J)\)
are fixed.  The replacement uses \(k\) copies of \(C\), so its size
grows linearly with the number of replaced vertices, and the scalar
recovery uses exact algebraic arithmetic.
	\end{proof}
	
\begin{figure}[htbp]
\centering
\begin{tikzpicture}[font=\small,box/.style={draw,minimum width=10mm,minimum height=8mm}]
\node[draw,circle,minimum size=7mm] (g) at (0,0) {$g$};
\node[box] (c1) at (1.65,0) {$C_1$};
\node (dots) at (3.15,0) {$\cdots$};
\node[box] (ck) at (4.65,0) {$C_k$};
\node[box] (j) at (6.55,0) {$J^\circ$};
\draw (g) -- node[above] {$A_1$} (c1);
\draw (c1) -- node[above] {$P_1$} (dots);
\draw (dots) -- node[above] {$A_k$} (ck);
\draw (ck) -- node[above] {$P_k$} (j);
\draw (c1.south) -- node[right] {$Y_1$} (1.65,-1.1);
\draw (ck.south) -- node[right] {$Y_k$} (4.65,-1.1);
\node at (3.15,-1.1) {$\cdots$};
\node[align=center] at (3.15,-1.65)
  {ordered boundary blocks of the removed \(h\)-vertices};
\end{tikzpicture}
\caption{The replacement in \cref{lem:single-occurrence-replication}.
Each horizontal connection represents \(r\) edges in the specified
variable order. Each \(Y_i\) represents \(s\) edges attached to the original
boundary of one removed \(h\)-vertex.  The network contains only the one
displayed \(g\)-vertex.  Here \(J^\circ\) is obtained by removing the
unique \(g\)-vertex from \(J\).}
\label{fig:single-occurrence-chain}
\end{figure}

	\begin{lemma}[Dichotomy with a nonzero unary]
		\label{lem:common-unary-exit}
		Let \(\cF\) be a finite algebraic signature set, and let
		\(u\ne0\) be unary.  If \(u\in\cF\), then
		\(\Holant(\cF)\) is \(\SharpP\)-hard unless \(\Tract(\cF)\) holds, and
		\(\KHolant(\cF)\) is \(\SharpP\)-hard unless \(\Tract(K\cF)\) holds.
	\end{lemma}
	
	\begin{proof}
		The ordinary-coordinate assertion is the unary case of
		\cref{thm:external-odd}.  For the K-coordinate assertion, apply the
		same theorem after the transformation in \cref{eq:K-equivalence},
		which preserves the arity and nonzeroness of the unary signature.
		This is the argument recorded in \cite[Lemma~3.4]{Liu2026DEQ2Holant}.
	\end{proof}
	
	\begin{corollary}[Higher equality supplies unaries]
		\label{cor:higher-equality-unary}
		Let \(d\ge4\), and let \(\cF\) be a finite algebraic signature set
		containing \((\Eq_d)\).  Then
		\[
		\begin{aligned}
		\Holant(\cF\cup\{\DeltaZero,\DeltaOne\})&\leT\Holant(\cF),\\
		\KHolant(\cF\cup\{\DeltaZero,\DeltaOne\})&\leT\KHolant(\cF).
		\end{aligned}
		\]
	\end{corollary}
	
	\begin{proof}
		Both reductions follow from \cite[Corollary~3.6]{Liu2026DEQ2Holant}.
		Restriction supplies the reverse reductions.
	\end{proof}

Together with \cref{lem:common-unary-exit}, these reductions settle the
		complexity of the original problem in either setting.
	
	\subsection{Hardness from an Eulerian-orientation Signature}
	\label{subsec:hc-relative-eo}

The height constructions produce EO signatures inside a set that may
also contain non-EO signatures. We show that, under the assumptions stated
before \cref{thm:hc-relative-eo-terminal}, an EO signature outside the
matching-product class defined below forces hardness.

	All gadget connections in this subsection use the edge \(X\).

	Arrow-reversal symmetry (ARS) means
	\(h({\mathbf x})=\overline{h(\bar {\mathbf x})}\) for every input \({\mathbf x}\).
	An EO signature \(h\) has \emph{projective arrow-reversal symmetry}
	(projective ARS) if, for some \(c\in\AlgNums^\times\),
	\[
	(c^{-1}h)({\mathbf x})=\overline{(c^{-1}h)(\bar {\mathbf x})}
	\qquad\text{for every \({\mathbf x}\in\{0,1\}^{\arity(h)}\)}.
	\]
	Here \(\bar {\mathbf x}\) is bitwise complement and the overline on the coefficient is
	complex conjugation.  An EO--ARS binary is both EO and ARS without rescaling.

	Let \(M\) be a perfect matching of a finite variable set \(V\).
	After choosing an orientation of every edge, write \(\vec M\) for
	the resulting set of ordered pairs.  A nonzero \emph{matching product} on
	\(V\) is a function \(f:\Ftwo^V\to\AlgNums\) of the form
	\[
	f({\mathbf x})=\lambda\prod_{(u,v)\in\vec M}b_{uv}(x_u,x_v),
	\qquad
	\lambda\in\AlgNums^\times,
	\quad b_{uv}:\Ftwo^2\to\AlgNums.
	\]

	Let \(\mathcal B_{\rm ARS}\) consist of the nonzero scalar multiples of
	products on arbitrary perfect matchings, with an orientation chosen for each
	matched pair, of binary EO--ARS signatures
	\begin{equation}
		(0,a,\overline a,0),\qquad a\in\AlgNums^\times.
		\label{eq:hc-rel-B}
	\end{equation}
	We write
	\((\Neq_4)=\delta_{0011}+\delta_{1100}\).
	Nonzero scalars and permutations of the output variables are suppressed according to the
	equivalence convention.  The required matching-product factorization constrains both the support and the coefficients.

We first record how fixed direct gadgets over a subset apply to the ambient signature set.
	
	\begin{lemma}[Replacing vertices by fixed gadgets in a larger signature set]
		\label{lem:ambient-eo-lifting}
		Let \(\cF_0\subseteq\cF\) be finite algebraic signature sets in
		K-coordinates.  Suppose that a finite sequence of fixed gadget constructions starting from
		\(\cF_0\) adds a finite set \(S\) using only fixed direct \(X\)-gadgets.
		Then the same derivation is valid over \(\cF\), and
		\[
		\KHolant(\cF\cup S)\leT\KHolant(\cF).
		\]
		If, at some stage, the problem \(\KHolant(\cF_0')\), with
		\(\cF_0'\subseteq\cF_0\cup S\), is \(\SharpP\)-hard, then
		\(\KHolant(\cF)\) is \(\SharpP\)-hard.
	\end{lemma}
	
	\begin{proof}
		Each direct gadget is substituted independently at every occurrence of its
		output signature, leaving vertices labelled from \(\cF\setminus\cF_0\)
		unchanged.  Induction over the finite sequence of fixed gadget constructions proves the displayed
		reduction.  For the final claim, restriction and that reduction give
		\[
		\KHolant(\cF_0')
		\leT\KHolant(\cF\cup S)
		\leT\KHolant(\cF).
		\]
        The lemma applies to the fixed direct gadgets in a literature proof.
        Tensor-factor extraction is supplied separately by \cref{thm:external-factor}.
	\end{proof}
	
	The common-one-sided case and the tractable cases have already been handled.
For the following theorem and
\cref{lem:hc-rel-complementary-unary,lem:hc-rel-mwx,lem:hc-rel-cfs},
fix a finite preprocessed algebraic set \(\cF\) of even-arity signatures in
\(K\)-coordinates satisfying:
	\begin{enumerate}[label=\textnormal{(\roman*)},leftmargin=*,labelindent=0pt]
		\item Among the values \(\chi_{\arity(f)}({\mathbf x})\), with \(f\in\cF\)
		and \({\mathbf x}\in\supp(f)\), at least one is positive and at least one
		is negative.
		\item \(\neg\Tract(K\cF)\).
	\end{enumerate}
Finite enlargements retain witnesses of both charge signs, so (i) persists.
By \cref{lem:downward-inheritance}, (ii) also persists under enlargement.
A nonzero unary then gives hardness by \cref{lem:common-unary-exit}.

	\begin{theorem}[Hardness from an EO signature outside \(\mathcal B_{\rm ARS}\)]
		\label{thm:hc-relative-eo-terminal}
		Assume (i)--(ii). Let \(h\in\cF\) be an EO signature.  Suppose there are no
		\(\lambda\in\AlgNums^\times\), perfect matching \(P_h\) of its variables
		together with an orientation \((i,j)\) of each matched pair, and coefficients
		\(a_{ij}\in\AlgNums^\times\) such that
		\[
		h=\lambda
		\bigotimes_{(i,j)\in P_h}(0,a_{ij},\overline{a_{ij}},0)_{ij}.
		\]
		Then \(\KHolant(\cF)\) is \(\SharpP\)-hard.
	\end{theorem}
	
The proof first obtains projective ARS or hardness, then proves hardness when the signature satisfying projective ARS lies outside \(\mathcal B_{\rm ARS}\).

	\begin{lemma}[Quaternary EO CP signatures supply a unary]
		\label{lem:hc-rel-complementary-unary}
		Under the assumptions preceding \cref{thm:hc-relative-eo-terminal},
		suppose \(D\in\cF\), where
		\[
		D=p\delta_{\mathbf s}+q\delta_{\bar {\mathbf s}},
		\qquad
		{\mathbf s}\in\{0,1\}^4,\quad \wt({\mathbf s})=2,\quad p,q\in\AlgNums^\times.
		\]
		Then there is a nonzero unary \(u\) such that
		\[
		\KHolant(\cF,u)\leT\KHolant(\cF).
		\]
	\end{lemma}
	
	\begin{proof}
		First reduce \(D\) to the unweighted quaternary CP signature.  After a
		variable permutation take \({\mathbf s}=0101\), and connect the same last two variable
		positions of two copies through \(X\)-edges.  The resulting ordered
		quaternary is
		\[
		\begin{aligned}
			Q(x_1,x_2,y_1,y_2)
			&=\sum_{u,v\in\Ftwo}
			D(x_1,x_2,u,v)D(y_1,y_2,1-u,1-v)\\
			&=pq\bigl(\delta_{0110}+\delta_{1001}\bigr).
		\end{aligned}
		\]
		Thus, in output order \((x_1,y_2,x_2,y_1)\), the output is the nonzero
		scalar multiple \(pq\,(\Neq_4)\). Hence \((\Neq_4)\) is directly
		realizable over \(\cF\), up to the displayed nonzero scalar.
		
		Put \(E=(\Neq_4)\).  Joining the fourth variables of
		\(E(a_1,a_2,a_3,u)\) and \(E(b_1,b_2,b_3,v)\) through \(X(u,v)\)
		realizes, in output order \((a_1,a_2,b_3,a_3,b_1,b_2)\), the six-ary
		generalized disequality
		\[
		\delta_{000111}+\delta_{111000}.
		\]
		Attach the \(\DeltaZero\)-variable of one \(\Delta\)-vertex through
		\(X\) to \(a_1\), and its \(\DeltaOne\)-variable through \(X\) to \(a_3\).
		This selects the string \(111000\). In remaining-variable order
		\((b_1,a_2,b_2,b_3)\), the output is
		\(\delta_{0101}=\Delta\otimes\Delta\).  In the notation of
		\cref{lem:single-occurrence-replication}, delete the unique
		\(\Delta\)-vertex to obtain a gadget with input block \(A=(a_1,a_3)\)
		and output blocks \(P=(b_1,a_2)\), \(Y=(b_2,b_3)\).
		Connecting a \(\Delta\)-vertex to \(A\) realizes \(\Delta(P)\Delta(Y)\).
		Connect \(P\) of each copy to \(A\) of the next, in the displayed orders,
		with no other connections.  A chain of \(k\) copies therefore
		uses exactly one \(\Delta\)-vertex and realizes
		\(\Delta(P_k)\prod_{i=1}^k\Delta(Y_i)\) on pairwise disjoint output blocks.
		Closing a single \(\Delta\)-vertex by
		\(X\) has value one.  The reductions in
		\cref{lem:universal-directed-pair,lem:single-occurrence-replication}
		therefore give
		\[
		\KHolant(\cF,\Delta)\leT\KHolant(\cF).
		\]
		
		Now use \(\Delta=\DeltaZero\otimes\DeltaOne\).
		Since both charge signs occur, the common-one-sided exception in
		\cref{thm:external-factor} is impossible.  Factor extraction and the
		preceding reduction give
		\[
		\KHolant(\cF,\DeltaZero,\DeltaOne)\leT\KHolant(\cF).\qedhere
		\]
	\end{proof}
	
	Write \(\mu_m=\{\omega\in\AlgNums:\omega^m=1\}\) for \(m\ge1\).

	\begin{lemma}[Projective ARS or hardness]
		\label{lem:hc-rel-mwx}
		Under the assumptions preceding \cref{thm:hc-relative-eo-terminal}, let \(h\in\cF\) be an
		EO signature.  Then either there is a
		\(c\in\AlgNums^\times\) such that
		\[
		(c^{-1}h)({\mathbf x})=\overline{(c^{-1}h)(\bar {\mathbf x})}
		\qquad({\mathbf x}\in\{0,1\}^{\arity(h)}),
		\]
		or \(\KHolant(\cF)\) is \(\SharpP\)-hard.
	\end{lemma}
	
	\begin{proof}
		We first show that projective ARS can be witnessed by an algebraic
		scalar.  Suppose an algebraic signature \(h_0\) is
		ARS after multiplication by an arbitrary nonzero complex scalar.  Then, for
		some \(\rho\in\mathbb C^\times\),
		\[
		h_0({\mathbf x})=\rho\,\overline{h_0(\bar {\mathbf x})}.
		\]
		Choose \({\mathbf x}\) with \(h_0({\mathbf x})\ne0\).  The same identity implies
		\(h_0(\bar {\mathbf x})\ne0\), and hence
		\[
		\rho=\frac{h_0({\mathbf x})}{\overline{h_0(\bar {\mathbf x})}}\in\AlgNums^\times,
		\qquad \rho\overline\rho=1.
		\]
		Set \(c=1+\rho\) if \(\rho\ne-1\), and set \(c=i\) if
		\(\rho=-1\).  Then \(c\in\AlgNums^\times\) and
        \(c/\overline c=\rho\), so \(c^{-1}h_0\) is ARS with the required
        algebraic scalar.
		
		We next handle the binary case.  Write \(h=(0,a,b,0)\).  If \(ab=0\), then \(h\) is a
		nonzero singular binary, so \cref{lem:binary-interpolation} and factor
		extraction give a nonzero unary.  By \(\neg\Tract(K\cF)\) and
		\cref{lem:common-unary-exit}, this is hard.  Suppose \(ab\ne0\), and put
		\[
		T_h=hX=\operatorname{diag}(a,b).
		\]
		If \(T_h\) has infinite projective order, \cref{lem:binary-interpolation}
		first gives a rank-one binary. Applying its singular case and factor
		extraction again gives a unary and hence hardness.  Otherwise \(a/b\) is a
		root of unity, so \(|a|=|b|\).  With
		\(\rho=a/\overline b\), one has
		\(h({\mathbf x})=\rho\overline{h(\bar {\mathbf x})}\) for all \({\mathbf x}\), and the preceding scalar
		argument proves the required projective ARS conclusion.
		
        Assume now that \(\arity(h)\ge4\). Let \(L\) be the number field generated
		over \(\mathbb Q\) by all entries of \(\cF\) and their complex conjugates.
		It is a finite extension of \(\mathbb Q\).  Every output of a fixed direct
		\(X\)-gadget over signatures already generated from \(h\) has entries in
		\(L\), and every nonzero ratio used to normalize such an output also lies in
		\(L\).
		
		Run the generating process of
		\cite[arXiv v2, Definition~29]{MengWangXia2025EO} starting from \(\{h\}\),
		with \(X\) as its initial binary.  All numbered references to this
		paper below use the same version.  Each self-loop and each binary path in that
		definition is a fixed direct \(X\)-gadget.  We therefore add its output to
		\(\cF\) by \cref{lem:ambient-eo-lifting}. Zero outputs are
		discarded.  If a generated nonzero binary is singular, the singular case of
		\cref{lem:binary-interpolation} gives unary factors.  If it is nonsingular
		but its transfer has infinite projective order, the other case of that lemma
		gives a rank-one binary and then unary factors.  Either event is hard by
		\(\neg\Tract(K\cF)\) and \cref{lem:common-unary-exit}.  The original witnesses
		of both charge signs exclude the factor-extraction exception at every
		stage.
		
		Otherwise every generated nonzero binary is nonsingular and,
		up to a nonzero scalar, has the form
		\[
		B_r=\begin{pmatrix}0&1\\ r&0\end{pmatrix},
		\qquad r\in L^\times,
		\qquad B_rX=\operatorname{diag}(1,r),
		\]
		where \(r\) is a root of unity.  Let \(R_i\) be the parameter set represented
		by the family \(B_i(h)\) at stage \(i\).  The path clause in
		Definition~29 makes each nonempty \(R_i\) the subgroup generated by the
		root-of-unity parameters available at that stage. Repeated self-loop descent
		\cite[Lemma~39]{MengWangXia2025EO} supplies nonemptiness.  To make the
		monotonicity explicit, write the stage operation as \(\Phi\).
		Its self-loop and path construction is monotone under inclusion of its
		input parameter set.  Finite order and the positive-power path clause give
		\(1\in R_1\), so \(R_0=\{1\}\subseteq R_1\).  Inductively,
		\(R_i=\Phi(R_{i-1})\subseteq\Phi(R_i)=R_{i+1}\).
		Equality at one stage forces equality thereafter by the same operation.
		Every \(R_i\) lies in the group of roots of unity of \(L\), which is finite.
		Consequently this increasing chain stabilizes at some finite stage \(i_0\).
		
        For a unique \(k\ge1\), the stabilized binary family is
		\[
		B(h)=B_{i_0}(h)=\{B_\omega:\omega\in\mu_k\}.
		\]
		Thus every gadget construction invoked below is finite, as
		required by \cref{lem:ambient-eo-lifting}.
		
		If \(k\ge3\), \cite[Lemmas~40--42]{MengWangXia2025EO} show that \(h\) is
		ARS up to a nonzero scalar.  The opening scalar argument completes this case.
		
		Suppose \(k=2\), and write \(Y=B_{-1}=(0,1,-1,0)\).  Both \(X\) and \(Y\)
		are available in this construction, and \(Y\) is realized by a direct
        gadget. The inductive argument in the proof of
        \cite[Lemma~47]{MengWangXia2025EO} uses fixed \(X\)- or
        \(Y\)-self-loops. At arity four, the proofs of
		\cite[Lemmas~45--46]{MengWangXia2025EO}, before invoking their general
		problem reduction, construct either an EO quaternary
		\(q\) for which the singleton problem is hard by
		\cite[Theorem~3.1]{CaiFuXia2018}, or a nonzero quaternary CP signature
		\[
		D=p\delta_{\mathbf s}+q'\delta_{\bar {\mathbf s}},
		\qquad {\mathbf s}\in\Ftwo^4,\quad\wt({\mathbf s})=2,\quad pq'\ne0.
		\]
		The two-copy subcase realizes \(\Neq_4\) itself.  All these signatures arise
		by the fixed self-loops and two-copy gadgets just identified, so
		\cref{lem:ambient-eo-lifting} adds them relative to \(\cF\).  The hard
		quaternary gives hardness by restriction.  The quaternary CP signature gives a
		unary by \cref{lem:hc-rel-complementary-unary}, and then hardness by
		\cref{lem:common-unary-exit}.  In the remaining case, the proof of
		Lemma~47 gives ARS up to a nonzero scalar, to which the opening argument applies.
		
		Finally suppose \(k=1\).  Follow the proof of
		\cite[Lemma~48]{MengWangXia2025EO}, and let \(C(h)\) denote its set of
		nonzero quaternary signatures obtained by successive \(X\)-self-loops.
		Every member of \(C(h)\) is therefore available by
		\cref{lem:ambient-eo-lifting}.  Moreover, every binary generated from such a
		signature is also generated from \(h\).  Hence its generated binary set is
		contained in \(B(h)=\{X\}\), and Lemma~39 makes that set nonempty. It is
		therefore exactly \(\{X\}\).  If a member of \(C(h)\) lies outside
		\(\{X^{\otimes2},Y^{\otimes2}\}\), the proof of Lemma~45 reaches exactly
		one of the hard quaternaries or quaternary CP signatures described in the preceding
		paragraph.  If \(Y^{\otimes2}\in C(h)\), add this signature and apply
		\cref{thm:external-factor}.  The presence of both charge signs excludes its exceptional case, so
		\(Y\) is available. The proof of Lemma~47 then applies because
        \(B(h)\subseteq\{X,Y\}\) and both binaries are available. Here factor
        extraction supplies the reduction otherwise obtained from
        \cite[Lemma~38]{MengWangXia2025EO}. The case \(C(h)=\varnothing\) is
		impossible for nonzero \(h\) by Lemma~39.  In the remaining case
		\(C(h)=\{X^{\otimes2}\}\), the reorganizing argument in Lemma~48, together
		with \cite[Lemma~43]{MengWangXia2025EO}, makes \(h\) a scalar multiple of a
		real complement-invariant signature, hence projectively ARS.  Its real and
		imaginary parts enter only the vector identities in that argument.  The opening scalar argument again
		makes the scalar algebraic.
	\end{proof}
	
	\begin{lemma}[Hardness outside \(\mathcal B_{\rm ARS}\) under projective ARS]
		\label{lem:hc-rel-cfs}
		Under the assumptions preceding \cref{thm:hc-relative-eo-terminal}, let \(h\in\cF\) be an
		EO signature for which there is a
		\(c\in\AlgNums^\times\) satisfying
		\[
		(c^{-1}h)({\mathbf x})=\overline{(c^{-1}h)(\bar {\mathbf x})}
		\qquad({\mathbf x}\in\{0,1\}^{\arity(h)}).
		\]
		If there are no \(\lambda\in\AlgNums^\times\), perfect matching \(P_h\)
		together with an orientation \((i,j)\) of each matched pair, and
		\(a_{ij}\in\AlgNums^\times\) such that
		\[
		h=\lambda
		\bigotimes_{(i,j)\in P_h}(0,a_{ij},\overline{a_{ij}},0)_{ij},
		\]
		then \(\KHolant(\cF)\) is \(\SharpP\)-hard.
	\end{lemma}
	
	\begin{proof}
		Put \(h'=c^{-1}h\), which is algebraic, EO, and ARS.  By uniqueness of irreducible factorization
		and the ARS factorization lemma
		\cite[Lemmas~2.13--2.14]{CaiFuShao2020ARS}, choose an irreducible
		factorization of \(h'\), up to an overall scalar, into factors
		\(g_1,\ldots,g_m\) on disjoint variable blocks, each satisfying ARS exactly.
		The coordinate restrictions called \emph{pinnings} in the proof of
		\cite[Lemma~2.14]{CaiFuShao2020ARS} are slices used to prove
		tensor-divisibility identities.  They do not supply direct pinning gadgets.
		
		We next choose algebraic representatives and restore exact ARS.  For each \(j\), choose
		\({\boldsymbol{\alpha}}_j\in\operatorname{supp}(g_j)\), let
		\({\boldsymbol{\alpha}}={\boldsymbol{\alpha}}_1\cdots{\boldsymbol{\alpha}}_m\), and define on the variable block of \(g_j\)
		\[
		q_j({\mathbf z})=
		\frac{h'({\boldsymbol{\alpha}}_1,\ldots,{\boldsymbol{\alpha}}_{j-1},{\mathbf z},
			{\boldsymbol{\alpha}}_{j+1},\ldots,{\boldsymbol{\alpha}}_m)}{h'({\boldsymbol{\alpha}})}.
		\]
		The chosen support assignments make the denominator nonzero, and
		\(q_j=g_j/g_j({\boldsymbol{\alpha}}_j)\).  Hence every \(q_j\) is algebraic, EO, irreducible,
		and projectively ARS, and
		\[
		h'=h'({\boldsymbol{\alpha}})\bigotimes_{j=1}^m q_j.
		\]
		For each \(q_j\), the algebraic-scalar argument in
		\cref{lem:hc-rel-mwx} gives
		\(c_j\in\AlgNums^\times\) such that \(p_j=c_j^{-1}q_j\) satisfies ARS
		exactly.  Thus
		\[
		h'=h'({\boldsymbol{\alpha}})\Bigl(\prod_{j=1}^m c_j\Bigr)
		\bigotimes_{j=1}^m p_j,
		\]
		with an algebraic overall scalar and algebraic EO--ARS irreducible factors \(p_j\).
        If every irreducible factor were binary, we would have
        \(h\in\mathcal B_{\rm ARS}\), contrary to the hypothesis. Thus there is a nonbinary irreducible EO--ARS factor
		\(f\notin\mathcal B_{\rm ARS}\).  Iterated use of
		\cref{thm:external-factor} adds \(f\) to the current set. Since
		both charge signs occur, its common-one-sided exception is impossible.  Whenever
		the cited descent subsequently extracts another factor, we repeat the
		displayed ratio construction and rescale by a nonzero algebraic scalar to
		obtain an ARS representative. Thus each signature used in the descent
		is a member of the current set or an extracted tensor factor.
		
		Apply the proof of
		\cite[Lemmas~3.2 and~4.1--4.7; Theorem~4.8]{CaiFuShao2020ARS} to the
		EO--ARS signature \(f\). The signature-producing operations are local: the tensor connections in Lemma~3.2 and the contractions in Lemma~4.6 are
		fixed direct \(X\)-gadgets, while the selected-variable attachments of the fixed
		binary \(iY=(0,i,-i,0)\) in Lemmas~4.5--4.6 are direct gadgets after that
		binary factor has been extracted.  These direct gadgets can be substituted over
		\(\cF\) by
		\cref{lem:ambient-eo-lifting}. Every factor-extraction step is supplied
		separately by \cref{thm:external-factor}, whose exception is again excluded
		by the presence of both charge signs.  Sections~3--4 of the cited proof use no polynomial
		interpolation.
		
		Concretely, Lemma~3.2 produces a direct EO quaternary whose singleton problem
		is hard, directly realizes \(\Neq_4\), or leaves the case treated in
		Lemmas~4.1--4.7. Using Lemmas~4.1--4.4, the descent in Lemmas~4.5--4.6
		preserves nonmembership in \(\mathcal B_{\rm ARS}\).
		Lemma~4.7 completes the base arities \(4,6,8\).
		Hence the descent adds either a hard quaternary
		or \(\Neq_4\) by the reductions just described.  The first gives hardness by restriction.
		The second gives a unary by \cref{lem:hc-rel-complementary-unary} and then
		hardness by \cref{lem:common-unary-exit}.
	\end{proof}
	
	\begin{proof}[Proof of \cref{thm:hc-relative-eo-terminal}]
		Apply \cref{lem:hc-rel-mwx}.  Either \(\KHolant(\cF)\) is
		\(\SharpP\)-hard, or \(h\) is projectively ARS.  In the latter case,
		\(h\notin\mathcal B_{\rm ARS}\), so \cref{lem:hc-rel-cfs} gives hardness.
	\end{proof}

	\Needspace{8\baselineskip}
	\subsection{Walsh Coordinates}
\label{subsec:walsh-local}

The equality reductions below and the height-one analysis use the same
Walsh basis to describe tensor coefficients.
For an arity-\(r\) signature \(f\), with the column-vector convention of
\cref{subsec:holant-definitions}, we fix
	the unnormalized Walsh matrix and the Walsh coordinates
	\begin{equation}
		P=\begin{pmatrix}1&1\\1&-1\end{pmatrix},
		\qquad
		\widetilde f=(P^{-1})^{\otimes r}f.
		\label{eq:Walsh-coordinates}
	\end{equation}
	Thus inversion is \(f=P^{\otimes r}\widetilde f\).  In coordinates,
	\[
	f({\mathbf x})=\sum_{{\mathbf y}\in\Ftwo^r}(-1)^{{\mathbf x}\cdot {\mathbf y}}\widetilde f({\mathbf y}).
	\]
	Consequently, if \(\widetilde f\) has even-parity support then \(f\) is
\emph{complement-invariant}, meaning \(f(\bar {\mathbf x})=f({\mathbf x})\).
If \(\widetilde f\) has odd-parity support, then \(f(\bar {\mathbf x})=-f({\mathbf x})\).  We write \(\mathbf1_{\rm even}\) for the even-parity indicator tensor,
    with its arity inferred from context.

For an ordered binary \(B\), its Walsh transfer matrix is
\[
P^{-1}T_BP=P^{-1}(BX)P=2\widetilde BZ,
\qquad \widetilde B=P^{-1}B(P^{-1})^{\mathsf T},
\]
since \(P^{\mathsf T}XP=2Z\). Its Walsh transfer class is therefore
\([P^{-1}T_BP]=[\widetilde BZ]\) when \(B\) is nonsingular.

	These coordinates describe the same tensor contractions with edge kernel \(2Z\),
    whose factor \(2\) is retained in the gadget calculations.
    This use of Walsh coordinates is distinct from a global transformation of the problem.

    \subsection{Equalities and Complementary-Pair Signatures}
	\label{proofsubsec:equality-anchor}

Three lemmas handle the equality signatures needed below.  The first
concerns a quaternary that becomes
\(\Eq_4\) under a common transformation taking the edge to equality.
The second is a generalized equality already available in
\(K\)-coordinates.  The third starts from a tensor factor with two
complementary support strings in Walsh coordinates. Its proof reduces to
the first form or obtains a
nonzero unary by interpolation.

	\begin{lemma}[Dichotomy from a transformed quaternary equality]
		\label{lem:equality-anchor}
		Let \(\cF\) be a finite algebraic signature set containing an ordered
		quaternary signature \(g\).  Suppose there are
		\[
		U\in\GL_2(\AlgNums),
		\qquad \kappa,\lambda\in\AlgNums^\times,
		\]
		such that
		\[
		(U^{-1})^{\mathsf T}XU^{-1}=\kappa(\Eq_2),
		\qquad
		U^{\otimes4}g=\lambda(\Eq_4).
		\]
		Then either \(\KHolant(\cF)\) is \(\SharpP\)-hard or
		\(\Tract(K\cF)\) holds.
	\end{lemma}
	
	\begin{proof} This is \cite[Lemma~3.8]{Liu2026DEQ2Holant}. \end{proof}

	\begin{lemma}[Dichotomy with a generalized equality]
		\label{lem:pure-ge}
		Let \(\cF\) be a finite algebraic signature set containing
		\[
		G_q=c_0\delta_{\zero^q}+c_1\delta_{\one^q},
		\qquad c_0,c_1\in\AlgNums^\times,\qquad q\ge4\text{ an integer}.
		\]
		Then either \(\KHolant(\cF)\) is \(\SharpP\)-hard or
		\(\Tract(K\cF)\) holds.
	\end{lemma}
	
	\begin{proof}
		This is \cite[Lemma~3.9]{Liu2026DEQ2Holant}.
	\end{proof}
	
	\begin{lemma}[Dichotomy from a CP factor in Walsh coordinates]
		\label{lem:walsh-ge-terminal}
		Let \(\cF\) be a finite algebraic signature set whose supported doubled
		charges include both signs.  Suppose a finite equivalent enlargement of
		\(\cF\) contains a nonzero signature with a tensor factor \(h\) of even
		arity \(2k\ge4\) with
		\[
		\widetilde h=(P^{-1})^{\otimes 2k}h
		=c_0\delta_{\mathbf a}+c_1\delta_{\bar {\mathbf a}},
		\qquad {\mathbf a}\in\Ftwo^{2k},\quad c_0c_1\ne0.
		\]
		Then \(\KHolant(\cF)\) is \(\SharpP\)-hard or \(\Tract(K\cF)\)
		holds.  If \(c_1/c_0\) is a root of unity, a finite equivalent enlargement
		contains a quaternary \(g\) satisfying the two transformation conditions of
		\cref{lem:equality-anchor} with \(U=K\).  
	\end{lemma}
	
	\begin{proof}
		Add the specified signature and then its factor \(h\) using
		\cref{eq:finite-equivalent-enlargement,thm:external-factor}. The supported
		charges of both signs exclude the common-one-sided exception.  All subsequent gadgets are
		\(K\)-gadgets.  In Walsh coordinates an edge is
		\(P^{\mathsf T}XP=2Z\).
		
        Choose two positions at which \({\mathbf a}\) has the same bit.
        Relabel the two support strings, and their coefficients, so that this
        bit is zero in \({\mathbf a}\). Take two copies of \(h\), leave these
        two positions exposed in each copy, and join all other
		corresponding positions between the copies.  Write \(J\) for the joined
		positions.  The diagonal edges allow only the two equally oriented pairs
		of support strings.  Since \(|J|=2k-2\) is even, their signs agree.  Thus the
		quaternary gadget output \(q\) has
		\[
		\widetilde q=\lambda(\delta_{0000}+\rho\delta_{1111}),
		\quad \rho=(c_1/c_0)^2,
		\quad \lambda=2^{2k-2}(-1)^{\wt({\mathbf a}|_J)}c_0^2\ne0.
		\]
		Closing two variables of \(q\) gives a binary \(B\) with
		\[
		\widetilde B=2\lambda\diag(1,-\rho),
		\qquad P^{-1}(BX)P=4\lambda\diag(1,\rho).
		\]
		The multiplicative order of \(\rho\) now determines which
		reduction applies.  If it is infinite, \cref{lem:binary-interpolation} supplies a rank-one binary.  Factor
		extraction supplies a nonzero unary, and \cref{lem:common-unary-exit}
		gives the conclusion.
		
		Otherwise let \(\ell\) be the order of \(\rho\).  A chain of \(\ell-1\)
		copies of \(B\) has Walsh transfer class \([\diag(1,\rho^{-1})]\).
		When \(\ell=1\) use the empty chain.  Attaching this chain to one variable of
		\(q\) gives a gadget output \(g\) with
		\(\widetilde g=c(\delta_{0000}+\delta_{1111})\), \(c\ne0\).
        Walsh inversion gives
		\[
		g=2c\mathbf1_{\rm even},\qquad K^{\otimes4}g=4c(\Eq_4),
		\qquad (K^{-1})^{\mathsf T}XK^{-1}=I.
		\]
		These are the conditions of \cref{lem:equality-anchor} with \(U=K\), so
		that lemma applies.  The equivalent enlargements transfer
		hardness back to \(\cF\). Tractability passes to the original
		subset by \cref{lem:downward-inheritance}.\looseness=-1
	\end{proof}
	
	\section{Obtaining a Mountain}
	\label{sec:factor-charge}

Starting from an even-arity signature set with supported half-charges of
both signs, we seek either both pinning signatures or a signature with two
nonzero endpoints and all remaining support on the central EO layer.
The latter is a mountain, defined below.  Its height will be the gcd of
the nonzero supported half-charges of the entire input set.  This choice
also retains the divisibility condition on the input signatures needed
in the later height analysis.

	\begin{definition}[Mountain]
		\label{def:mountain}
		For an integer \(d\ge1\), an arity-\(2d\) signature \(M\) in
		\(K\)-coordinates is a mountain of height \(d\) if
		\[
		\supp(M)\subseteq HW_{2d}^=\cup\{\zero^{2d},\one^{2d}\}
		\qquad\text{and}\qquad
		M(\zero^{2d})M(\one^{2d})\ne0.
		\]
		Equivalently,
		\[
		\{\chi_{2d}({\mathbf x}):M({\mathbf x})\ne0\}\subseteq\{-d,0,d\},
		\qquad
		\{-d,d\}\subseteq\{\chi_{2d}({\mathbf x}):M({\mathbf x})\ne0\}.
		\]
		Thus the height is both half the arity and the absolute integer half-charge
		of either endpoint.  No restriction is placed on the coefficients on the central Hamming layer.
	\end{definition}

	\begin{lemma}[Producing the gcd-height signature or both pinning signatures]
		\label{lem:charge-bridge}
        Let \(\cF\) be a finite preprocessed set of even-arity algebraic
        signatures.
		Assume that among the supported half-charges
		\[
        \chi_{\arity(f)}({\mathbf x})=\wt({\mathbf x})-\frac{\arity(f)}2,
		\qquad f\in\cF,\quad f({\mathbf x})\ne0,
		\]
		there is at least one positive value and at least one negative value.  Put
		\[
		d=\gcd\left\{
        \left|\chi_{\arity(f)}({\mathbf x})\right|:
        f\in\cF,\ f({\mathbf x})\ne0,\
        \chi_{\arity(f)}({\mathbf x})\ne0
		\right\}.
		\]
        Then \(d\ge1\), every supported half-charge
        \(\chi_{\arity(f)}({\mathbf x})\) lies in \(d\mathbb Z\), and at least one of the
		following reductions holds.
		\begin{enumerate}
			\item There is an arity-\(2d\) signature \(M\) such that
			\[
			\KHolant(\cF,M)\leT\KHolant(\cF),
			\]
			\[
			M(\zero^{2d})M(\one^{2d})\ne0,
			\qquad
			M({\mathbf x})=0\quad\text{if }\wt({\mathbf x})\notin\{0,d,2d\}.
			\]
			\item Both pinning signatures are available by reduction:
			\[
			\KHolant(\cF,\DeltaZero,\DeltaOne)
			\leT\KHolant(\cF).
			\]
		\end{enumerate}
	\end{lemma}
	
	\begin{proof}
		Let
		\[
		C=\{\chi_{\arity(f)}({\mathbf x}):
		f\in\cF,\ f({\mathbf x})\ne0,\ \chi_{\arity(f)}({\mathbf x})\ne0\}
		\]
		be the finite set of nonzero supported half-charges of the input set.
		The gcd assertion follows from the definition of \(d\).
		The next argument expresses \(d\) and \(-d\) as sums of elements of
		\(C\) with nonnegative multiplicities.  Tensor products add half-charges, but their factor
		multiplicities cannot be negative.  Because \(C\) contains both signs, take every pair
		\(a,-b\in C\) with \(a,b>0\).  If \(g=\gcd(a,b)\), assigning multiplicity
		\(b/g\) to \(a\) and \(a/g\) to \(-b\) gives total half-charge zero.  Summing
		these assignments over all positive--negative pairs gives a strictly
		positive multiplicity for each element of \(C\), since each participates
		in at least one pair.  We therefore obtain a vector
		\({\mathbf v}\in\mathbb Z_{>0}^C\) satisfying
		\[
		\sum_{c\in C}v_c c=0.
		\]
		Bézout's identity, iterated over the finite set \(C\), supplies an integer
		vector \({\mathbf u}\in\mathbb Z^C\) with \(\sum_{c\in C}u_c c=d\)
		\cite[Section~1.3, Theorem~1.2]{Apostol1976}.  The entries of \({\mathbf u}\) need not be
		nonnegative, but every entry of \({\mathbf v}\) is positive.  Consequently, for all sufficiently
		large nonnegative integers \(k\), both vectors
		\[
		{\mathbf n}^+={\mathbf u}+k{\mathbf v},\qquad {\mathbf n}^-=-{\mathbf u}+k{\mathbf v}
		\]
		are nonnegative integer vectors, while their total half-charges remain \(d\)
		and \(-d\), respectively.  Thus \({\mathbf n}^+\) and \({\mathbf n}^-\) are the required
		multiplicity lists.
		
For each \(c\in C\), choose a signature \(f_c\in\cF\) and an input
\({\mathbf x}_c\in\supp(f_c)\) with \(\chi_{\arity(f_c)}({\mathbf x}_c)=c\).
Form the two tensor products using \(n_c^+\) and \(n_c^-\) copies of
\(f_c\), respectively.  Concatenating the chosen inputs in these copies
gives a nonzero product of signature values, with half-charge \(d\) in
the first tensor product and \(-d\) in the second.  Every supported
half-charge of either tensor product is divisible by \(d\), since the
same is true of each input signature.

For these two supported inputs, the excess of ones over zeros, or of
zeros over ones, is \(2d\).  The endpoint-preserving self-loop lemma
\cite[full version, Lemma~40; conference version, Lemma~37]{MengWangXiaZheng2025}
therefore supplies sequences of \(X\) self-loops whose outputs
\(P,N\) have arity \(2d\) and satisfy
\(P(\one^{2d})N(\zero^{2d})\ne0\).
Both outputs are realized by fixed gadgets over \(\cF\).
Each \(X\) self-loop removes one zero and one one from every
contributing input, so the supported half-charges remain divisible by
\(d\).  At arity \(2d\), they lie between \(-d\) and \(d\), giving
\[
\supp(P)\cup\supp(N)
\subseteq
HW_{2d}^=\cup\{\zero^{2d},\one^{2d}\}.
\]
		If either \(P\) or \(N\) has both endpoints, take that signature as the
		required \(M\).
		We may therefore assume
		\[
		P(\zero^{2d})=0,
		\qquad
		N(\one^{2d})=0.
		\]
		Apart from the two nonzero endpoints already identified, all supported
		inputs have weight \(d\). Their central restrictions are \(P|_{\EO}\) and \(N|_{\EO}\).
		
If both central restrictions are nonzero, permute chosen central strings to
\((\zero^d,\one^d)\) in \(P\) and \((\one^d,\zero^d)\) in \(N\).
Expose the first \(d\) variables of each signature and join their last
\(d\) variables in the same order by \(X\) edges.  Call the resulting
arity-\(2d\) signature \(Q\).  At the all-zero boundary, the support
restriction and \(P(\zero^{2d})=0\) force the last block of \(P\) to be
\(\one^d\), and hence the last block of \(N\) to be \(\zero^d\).
At the all-one boundary, \(N(\one^{2d})=0\) similarly forces the last
block of \(N\) to be \(\zero^d\), and that of \(P\) to be \(\one^d\).
Thus each endpoint sum has exactly one nonzero term, and
\[
Q(\zero^{2d})=P(\zero^d,\one^d)N(\zero^{2d}),\qquad
Q(\one^{2d})=P(\one^{2d})N(\one^d,\zero^d),
\]
both nonzero.

If only \(P\) has a nonzero central restriction, write
\(N=n\delta_{\zero^{2d}}\) for \(n\in\AlgNums^\times\), and permute a
supported central string of \(P\) to \((\zero^d,\one^d)\).  Join the two
\(d\)-variable blocks of \(N\) to the last \(d\) variables of two copies of
\(P\), again by \(X\) edges, leaving the first \(d\) variables of
both copies exposed.  The support of \(N\) forces its variables to zero
and the connected variables of both copies of \(P\) to one.  The endpoint
extensions are therefore unique, with values
\(nP(\zero^d,\one^d)^2\) and \(nP(\one^{2d})^2\), both nonzero.
If only \(N\) has a nonzero central restriction, use the same construction
with the roles of \(P,N\) and of zero and one interchanged.
In either construction, \(X\) charge conservation preserves
divisibility by \(d\).  Since the output has arity \(2d\), its supported
half-charges belong to \(\{-d,0,d\}\). Take that output as \(M\).

It remains to consider the case in which both central restrictions are zero.
Then \(P=p\DeltaOne^{\otimes2d}\) and
\(N=n\DeltaZero^{\otimes2d}\) for \(p,n\in\AlgNums^\times\).
The set \(\cF\cup\{P,N\}\) retains the original witnesses of positive and
negative half-charge, so the common-one-sided exception in
\cref{thm:external-factor} cannot occur.  Applying that theorem to the
pure tensor powers, with nonzero scalars handled by preprocessing, and
substituting the gadgets for \(P,N\) proves the second reduction.
	\end{proof}

In the pinning outcome of \cref{lem:charge-bridge}, if
\(\neg\Tract(K\cF)\), then
\cref{lem:downward-inheritance,lem:common-unary-exit} and the displayed
reduction give \(\SharpP\)-hardness of \(\KHolant(\cF)\).
Thus the remaining analysis uses the mountain outcome.
	
\begin{remark}
The pinning case in \cref{lem:charge-bridge} is necessary.
Consider \(\{\operatorname{EXACT}^4_3,\operatorname{EXACT}^4_1\}\), where
\(\operatorname{EXACT}^4_j({\mathbf x})=\one[\wt({\mathbf x})=j]\).
These signatures have single supported half-charges \(1\) and \(-1\),
respectively.  In a fixed gadget built directly from them using
\(X\) edges, charge conservation fixes the half-charge of every
supported boundary input: it is the number of
\(\operatorname{EXACT}^4_3\)-vertices minus the number of
\(\operatorname{EXACT}^4_1\)-vertices.  Such a gadget cannot have both
endpoints of a mountain.  The lemma instead gives both pinning signatures
through factor extraction.
	\end{remark}
	\section{Mountain Height One}
	\label{sec:binary-dynamics}
	
	A height-one mountain is an ordered binary signature whose two diagonal
	entries are nonzero. We prove the dichotomy for any finite algebraic
	signature set containing one. First, a nonzero unary settles the case by
	the odd-arity dichotomy. A nonzero singular binary, or a nonsingular binary
	whose transfer has infinite projective order, leads to the same conclusion
	by \cref{lem:binary-interpolation}.

	In the remaining case, the projective transfers of all nonzero directly
	realizable binaries form a finite group. By \cref{thm:finite-group-split},
	we either obtain a nonsingular weighted equality or normalize this group
	to an odd cyclic or odd dihedral form. A weighted equality can be
	normalized to \(I=(\Eq_2)\) by a global diagonal transformation in
	\(GO(X)\), so \cref{thm:eq2-kholant-dichotomy} settles that branch.

	For the two remaining group forms, contractions and factor extraction
	either give hardness, identify a common tractable basis, or yield a strict
	update: the transfer group grows in the original coordinates, or the
	minimum nonbinary irreducible arity decreases. A uniform bound over the
	initial number field makes these updates finite. We establish the common
	contraction tools, treat the dihedral and cyclic cases separately, and
	then combine them with the termination argument.

	\begin{theorem}[Dichotomy with a height-one mountain]
		\label{cor:height-one-complete}
		Let \(\cF\) be a finite algebraic signature set containing a height-one
		mountain.  Then
		\[
		\Tract(K\cF)
		\qquad\text{or}\qquad
		\KHolant(\cF)\text{ is }\SharpP\text{-hard}.
		\]
	\end{theorem}
	
	We prove this theorem at the end of the section, after establishing the
	finite-group alternatives and the termination of their updates.  
	
	\subsection{Binary equality}
	\label{subsec:equality-quaternary-interfaces}
	
    Binary equality gives the following classification, due to
    Liu and Meng \cite{Liu2026DEQ2Holant}.
	
	\begin{theorem}[Binary-equality \(K\)-Holant dichotomy]
		\label{thm:eq2-kholant-dichotomy}
		Let \(\cF\) be a finite algebraic signature set containing
		\(I=(\Eq_2)\).  If \(\Tract(K\cF)\) holds, then
		\(\KHolant(\cF)\in\FP\). Otherwise
		\(\KHolant(\cF)\) is \(\SharpP\)-hard.
		The predicate \(\Tract(K\cF)\) is decidable.
	\end{theorem}
	
	\begin{proof}
		\cite[Theorem~1.2]{Liu2026DEQ2Holant} is stated with a reduced four-clause
		predicate.
		Since \(I\in\cF\), its coordinate image contains
		\(Z=K^{\otimes2}I\in K\cF\). Hence the binary-anchor lemma
		\cite[Lemma~2.13]{Liu2026DEQ2Holant} identifies that reduced predicate with
		\(\Tract(K\cF)\).  The displayed formulation follows.
	\end{proof}
	
	\subsection{Finite Transfer Groups}
\label{subsec:finite-transfer-groups}

Series composition turns binary gadgets into matrix products. Passing to
projective classes removes the nonzero scalars already recorded in the
reductions. We use finite-group structure to find a weighted equality or
restrict the remaining transfers to cyclic and dihedral forms. The edge
and sign matrices remain fixed throughout this classification.

We use \(H\le G\) for a subgroup, \(H\triangleleft G\) for a normal
	subgroup, and \([G:H]\) for its index.  An involution has order two.
	Our finite-group notation is
	\[
	C_m,\qquad
	D_{2m}=\langle r,s\mid r^m=s^2=1,\ srs=r^{-1}\rangle,\qquad
	V_4=C_2\times C_2,
	\]
	where \(C_m\) is cyclic of order \(m\) and \(D_{2m}\) has order \(2m\).
	The symmetric and alternating groups are \(S_m\) and \(A_m\).
	For \(S\subseteq G\), the centralizer is
    \(C_G(S)=\{g\in G:gs=sg\text{ for every }s\in S\}\).
    If \(\zeta\) is a primitive \(m\)-th root of unity,
	then \([\mathbb Q(\zeta):\mathbb Q]=\varphi(m)\), where \(\varphi\) is
	Euler's totient.
	
	\subsubsection{Ordered Bell signatures}
	\label{subsec:pauli-notation}
	
	The ordered Bell signatures are \(\mathcal B=\{I,X,Y,Z\}\).
    With truth-table order \(00,01,10,11\), they are
	\begin{equation}
		\begin{aligned}
			I&=(1,0,0,1),&
			Z&=(1,0,0,-1),\\
			X&=(0,1,1,0),&
			Y&=(0,1,-1,0).
		\end{aligned}
		\label{eq:bell-matrix-dictionary}
	\end{equation}
	Our real skew matrix \(Y\) differs from the conventional complex Pauli
	\(Y\)-matrix by a nonzero scalar.
	
	\begin{remark}[Bell notation in Shao--Cai]
		The tuple in \cref{eq:bell-matrix-dictionary} is Shao--Cai's
		\((=^+_2,=^-_2,\neq^+_2,\neq^-_2)\) \cite[Section~2]{ShaoCai2020}.
		\(Y^{\mathsf T}=-Y\).  Their holographic matrix named \(Z\) is our \(K\),
		whereas our \(Z\) is only the diagonal sign matrix.
	\end{remark}

Recall from \cref{subsec:reduction-conventions} that \(T_B=BX\),
\(T_{BXC}=T_BT_C\), and \(T_{B^{\mathsf T}}=(T_B)^\#\).
	For a general \(T\in\Mat_2(\AlgNums)\), the operation defined in
	\cref{subsec:reduction-conventions} satisfies
	\begin{equation}
		T^\#=XT^{\mathsf T}X,
		\qquad
		T^\#=(\det T)ZT^{-1}Z\quad\text{if }T\text{ is nonsingular}.
		\label{eq:transfer-reversal}
	\end{equation}
	
	\begin{remark}[Relation to binary-signature classifications]
        \Cref{lem:binary-interpolation} handles nonzero singular binaries and
        nonsingular binaries whose transfer \(T_B=BX\) has infinite projective
        order, leaving the finite-order case. This is the K-Holant form of the
        binary classification used by Cai--Lu--Xia
        \cite{CaiLuXia2009,CaiLuXia2018RealHolantc}. The four ordered Bell
		signatures in \(\mathcal B\) represent
		the projective Bell classes of Shao--Cai \cite{ShaoCai2020}.  Beyond
		the binary classification, we need the complete transfer group
		and uniform control of its higher-arity contractions.
	\end{remark}
	
	\begin{lemma}[Bounded torsion]
		\label{lem:bounded-torsion}
		Fix a number field \(L\), and let \(G\) be the set of projective
		transfers of all directly realizable nonsingular ordered
		binary gadgets over a fixed signature set with entries in \(L\).
		If \(G\ne\varnothing\) and every element of \(G\) has finite
		projective order, then \(G\) is a finite subgroup of
		\(\PGL_2(L)\).
	\end{lemma}
	
	\begin{proof}
		This is \cite[Lemma~5.2]{Liu2026DEQ2Holant}.
	\end{proof}
	
	For a current signature set \(\cF\), suppose that every nonzero directly
	realizable ordered binary \(B\) is nonsingular and that \([BX]\) has finite
	order.  Series composition and \cref{lem:bounded-torsion} then give the finite
	group
	\begin{equation}
		G(\cF)
		=\{[BX]:B\ne0\text{ is an ordered binary directly realizable over }\cF\}.
		\label{eq:complete-safe-group}
	\end{equation}
We call this the \emph{complete transfer group}, since it includes the
transfer classes of all directly realizable nonsingular binaries over \(\cF\).

We next classify this finite group. Changes of basis and comparisons
between successive groups are handled in \cref{sec:odd-residue}.
	
	\begin{proposition}[Finite projective binary groups]
		\label{prop:finite-pgl2-classification}
		Every finite subgroup of \(\PGL_2(\mathbb C)\) is conjugate to exactly one
		of the following types:
		\[
		C_m,\qquad D_{2m},\qquad A_4,\qquad S_4,\qquad A_5,
		\]
		where \(m\ge1\) for \(C_m\), \(m\ge2\) for \(D_{2m}\),
		\(|C_m|=m\), \(|D_{2m}|=2m\), and
		\(|A_4|=12\), \(|S_4|=24\), \(|A_5|=60\).
	\end{proposition}
	
	\begin{proof}
		This is the classical conjugacy classification
		\cite[Section~5.1]{Beardon1983}, also stated in
		\cite[Proposition~2.1]{Liu2026DEQ2Holant}.
	\end{proof}
	
	\begin{lemma}[Finite projective-group facts]
		\label{lem:finite-pgl2-facts}
		The following properties will be used below.
		\begin{enumerate}
			\item In
			\[
			D_{2m}=\langle r,s\mid r^m=s^2=1,\ srs=r^{-1}\rangle,
			\]
			the rotations \(\angles r\cong C_m\) form a normal index-two
			subgroup, and every element outside it is a reflection and an
			involution.  If \(m\) is odd, the rotation group is the unique
			index-two subgroup.  If \(m\) is even, the only nonidentity
			rotational involution is the central element \(r^{m/2}\).
			\item The group \(A_4\) has no index-two subgroup, and the centralizer
			of an involution is \(V_4\).
			\item The unique index-two subgroup of \(S_4\) is \(A_4\).  Every
			involution is a transposition or a double transposition. Precisely
			the double transpositions belong to \(A_4\).
			\item The group \(A_5\) is simple and has no index-two subgroup.  The
			centralizer of each involution is \(V_4\).
			\item The projective centralizer of the distinguished involution is
			\[
			C_{\PGL_2(\mathbb C)}([Z])
			=
			\left\{
			\left[\begin{pmatrix}a&0\\0&b\end{pmatrix}\right]:a,b\in\mathbb C,\ ab\ne0
			\right\}
			\sqcup
			\left\{
			\left[\begin{pmatrix}0&a\\b&0\end{pmatrix}\right]:a,b\in\mathbb C,\ ab\ne0
			\right\}.
			\]
			Thus \([Z]\) is the only nonidentity diagonal involution, while
			every nonsingular anti-diagonal class is an involution.
		\end{enumerate}
	\end{lemma}
	
	\begin{proof}
		The first four items follow from the cyclic and dihedral presentations
		and the standard permutation models \cite{DixonMortimer1996}.
		The projective centralizer calculation in the last item is given in
		\cite[Lemma~2.2]{Liu2026DEQ2Holant}.
	\end{proof}
	
	The classification identifies the abstract group type. To use it in
	K-Holant, we must also preserve the edge through a change of basis in
	\(GO(X)\). The proof below keeps the distinguished classes \([X]\) and
	\([Z]\) fixed and constructs the required normalization explicitly.
	
	Let \(G\le\PGL_2(\AlgNums)\) be finite, and suppose each class has
	a directly realized ordered-binary representative.  Choose one
	nonsingular ordered binary \(B_g\) for each class \(g\), normalized by
	\(R_g=B_gX\), \([R_g]=g\), and \(R_1=I\), where \(1=[I]\).  The matrix of
	every representative of class \(g\) is \(\rho B_g\) for a unique nonzero
	scalar \(\rho\).  Fix one direct gadget for each of these finitely many
	representatives and substitute it whenever that representative is used.
	Displayed
	projective identities are read up to their known nonzero scalars. A scalar is
	shown explicitly whenever its zero/nonzero status or a coefficient ratio is
	used.

	\begin{theorem}[Weighted equality or an odd cyclic or dihedral group]
		\label{thm:finite-group-split}
		Let \(\cF\) be a finite algebraic signature set containing the height-one
		mountain
		\[
		B=\begin{pmatrix}a&b\\c&d\end{pmatrix},
		\qquad ad\ne0,
		\]
		and suppose every nonzero directly realizable binary over \(\cF\) is
		nonsingular and its transfer has finite projective order.  Then one of the following
		holds.
		\begin{enumerate}
			\item A nonsingular weighted equality
			\(E=\diag(p,q)\), with \(pq\ne0\), is directly realizable over
			\(\cF\), up to a known nonzero scalar.
			\item There exist an odd \(N\ge3\) and \(D\in GO(X)\) such that the global normalization \(\cF\mapsto D\cF\) puts
			\(G(D\cF)\) in exactly one of the forms
			\[
			\{[R(s)]:s\in\mu_N\},
			\qquad
			\{[R(s)],[ZR(s)]:s\in\mu_N\},
			\qquad
			R(s)=\frac12
			\begin{pmatrix}1+s&1-s\\1-s&1+s\end{pmatrix}.
			\]
			Every displayed projective class has a fixed directly realizable
			ordered-binary representative over \(D\cF\).
		\end{enumerate}
	\end{theorem}

    For \(s\in\AlgNums\), the displayed matrix family satisfies
    \(P^{-1}R(s)P=\diag(1,s)\).
	
	\begin{proof}
		Put \(G=G(\cF)\) and
		\[
		H=\angles{G,[Z]}.
		\]
		All gadget entries lie in the number field generated by the entries of
		\(\cF\), so \cref{lem:bounded-torsion} makes \(G\) finite.  If
		\([T]\in G\), swapping the two variables and \cref{eq:transfer-reversal} give
		\([ZT^{-1}Z]\in G\).  Since inversion permutes \(G\), it follows that
		\([Z]G[Z]=G\).  Consequently
		\[
		G\triangleleft H,\qquad [H:G]\le2,
		\]
		and \(H=G\) or \(H=G\sqcup[Z]G\), so \(H\) is finite.
		
		For the distinguished mountain,
		\[
		T=BX=\begin{pmatrix}b&a\\d&c\end{pmatrix}
		\]
		has class \([T]\in G\).  Since \(a,d\ne0\), both off-diagonal entries of
		\(T\) are nonzero.  An anti-diagonal transfer will give the weighted equality we seek. Suppose \([S]\in G\)
		has an anti-diagonal representative
		\[
		S=\begin{pmatrix}0&p\\q&0\end{pmatrix},
		\qquad pq\ne0.
		\]
		By the definition of \(G\), there is a directly realizable ordered binary
		\(B_S\) and a scalar \(\rho\ne0\) such that \(B_SX=\rho S\).  Hence
		\[
		B_S=\rho SX=\rho\diag(p,q),
		\]
		which is the nonsingular weighted equality in outcome~1.
		
		We now use
		\cref{prop:finite-pgl2-classification,lem:finite-pgl2-facts}.  Their use
		here is only group-theoretic: the distinguished class is the fixed
		\([Z]\), and no conjugating matrix from the abstract classification is
		applied to the signature set.
		
		If \(H\) is cyclic, it is abelian, so \([T]\) commutes with \([Z]\).  The
		projective centralizer of \([Z]\) consists of the diagonal and
		anti-diagonal classes.  The two nonzero off-diagonal entries of \(T\)
		exclude the diagonal case. Thus \([T]\in G\) is anti-diagonal and gives the
		weighted equality.
		
		Suppose \(H\cong A_4\).  Since \(A_4\) has no index-two subgroup, the
		index bound forces \(G=H\).  The centralizer of \([Z]\) in \(H\) is a
		\(V_4\), so it contains an involution \(u\ne1,[Z]\).  The class \(u\)
		belongs to \(G\) and commutes with \([Z]\).  Because \([Z]\) is the only
		nonidentity diagonal involution, \(u\) is anti-diagonal.  The same argument
		applies when \(H\cong A_5\): simplicity gives \(G=H\), and the
		centralizer of an involution is again \(V_4\).
		
		Suppose \(H\cong S_4\).  If \(G=H\) and \([Z]\) corresponds to a
		transposition, choose the disjoint transposition. If \(G=H\) and \([Z]\) corresponds
		to a double transposition, choose either of the other two double
		transpositions.  In both cases the chosen involution lies in \(G\), commutes
		with \([Z]\), and differs from it, so it is anti-diagonal.  If \(G\ne H\),
		then the unique index-two subgroup gives \(G=A_4\).  Here \([Z]\notin G\)
		is an involution outside \(A_4\), hence a transposition.  If \(z'\) is the
		disjoint transposition, then \(u=[Z]z'\) is a double transposition in
		\(A_4=G\). It commutes with \([Z]\), differs from it, and is therefore
		anti-diagonal.
		
		It remains to consider
		\[
		H=D_{2m}=\angles{r,z:r^m=z^2=1,\ zrz=r^{-1}},
		\]
		where the distinguished reflection is denoted by \(z=[Z]\) when \([Z]\)
		is a reflection.  If \(m\) is even and \([Z]=r^{m/2}\) is the central
		half-turn, then \([T]\in G\) commutes with \([Z]\). Its two nonzero
		off-diagonal entries again make it anti-diagonal.  Suppose instead that
		\([Z]\) is a reflection and \(m\) is even.  The two involutions
		\[
		r^{m/2},\qquad [Z]r^{m/2}
		\]
		commute with \([Z]\) and differ from it.  If \(G=H\), both belong to
		\(G\).  If \([H:G]=2\), then \([Z]\notin G\), and the two displayed
		elements lie in opposite cosets of \(G\). Hence exactly one belongs to
		\(G\).  In either event \(G\) contains an involution different from
		\([Z]\) that commutes with \([Z]\), and this involution is anti-diagonal.
		
		The sole unresolved case is therefore
		\[
		H\cong D_{2N},\qquad N\ge3\text{ odd},\qquad [Z]\text{ a reflection},
		\]
		with \(G=C_N\) or \(G=D_{2N}\).  Indeed, for odd \(N\) the rotation
		subgroup is the unique index-two subgroup. Thus a proper \(G\) is the
		rotation subgroup, while otherwise \(G=H\).
		
		Choose a rotation class \([R]\) of exact order \(N\), and write
		\[
		R=\begin{pmatrix}u&v\\w&x\end{pmatrix}.
		\]
		The dihedral relation gives \([ZRZ]=[R^{-1}]\), while
		\[
		ZRZ=\begin{pmatrix}u&-v\\-w&x\end{pmatrix},
		\qquad
		R^{-1}\doteq\begin{pmatrix}x&-v\\-w&u\end{pmatrix}.
		\]
		The entries \(v,w\) cannot both vanish: otherwise \([R]=[R^{-1}]\), so
		\([R]^2=1\), contrary to the odd order \(N\ge3\).  Since at least one of
		\(v,w\) is nonzero, comparison in the preceding projective equality forces
		the projective scalar to be \(1\), and hence \(u=x\).  Exactly one of
		\(v,w\) cannot vanish.  In that case nonsingularity gives \(u=x\ne0\), and
		a projective rescaling makes \(R\) a nonidentity unipotent matrix. Its
		powers are pairwise distinct in characteristic zero.  Therefore
		\[
		R=\begin{pmatrix}\alpha&\beta\\ \gamma&\alpha\end{pmatrix},
		\qquad \beta\gamma\ne0.
		\]
		
		Choose \(q\in\AlgNums^\times\) with \(q^2=\gamma/\beta\), and put
		\[
		D=\diag(q,1).
		\]
		Then
		\[
		D^{\mathsf T}XD=qX,
		\]
		so \(D\in GO(X)\). Moreover \(DZD^{-1}=Z\).  With
		\(\delta=q\beta=q^{-1}\gamma\),
		\[
		DRD^{-1}
		=\begin{pmatrix}\alpha&\delta\\ \delta&\alpha\end{pmatrix}.
		\]
		In Walsh coordinates,
		\[
		P^{-1}(DRD^{-1})P
		=\diag(\alpha+\delta,\alpha-\delta).
		\]
		Both eigenvalues are nonzero because \(R\) is nonsingular.  Put
		\[
		s_0=\frac{\alpha-\delta}{\alpha+\delta}.
		\]
		The exact projective order of \([R]\) implies that \(s_0\) is a primitive
		\(N\)-th root of unity, and
		\[
		DRD^{-1}
		=(\alpha+\delta)R(s_0),
		\qquad
		R(s)=\frac12\begin{pmatrix}1+s&1-s\\1-s&1+s\end{pmatrix}.
		\]
		Since \(R(s)R(t)=R(st)\), the normalized rotation classes are precisely
		\[
		\{[R(s)]:s\in\mu_N\}.
		\]
		The diagonal matrix \(D\) fixes \([Z]\).  Thus, when \(G=C_N\), these are
		all normalized classes, and, when \(G=D_{2N}\), the remaining classes are
		precisely
		\[
		\{[ZR(s)]:s\in\mu_N\}.
		\]
		
		To identify the normalized group exactly, we compare gadget outputs in
		both directions. If \(B_0\) is a directly realizable ordered binary over
		\(\cF\), the global change of basis gives the boundary signature
		\(B_0^D\) over \(D\cF\) with
		\[
		B_0^D\doteq DB_0D^{\mathsf T},
		\qquad
		(DB_0D^{\mathsf T})X=qD(B_0X)D^{-1}.
		\]
		Hence \(T_{B_0^D}\doteq D T_{B_0}D^{-1}\).  Applying the same argument
		with \(D^{-1}\in GO(X)\) gives the converse, so
		\[
		G(D\cF)=[D]G(\cF)[D]^{-1}.
		\]
        Every displayed class therefore has an ordered binary representative
        realizable over \(D\cF\), with the normalization \(D\in GO(X)\).
	\end{proof}
	
	\subsection{Invariants and Finite Updates}
	\label{sec:odd-residue}

	A later contraction or factor extraction can enlarge the transfer group,
so the group classification may need to be applied again. To prove that
this process terminates, we compare all groups in the original coordinates
over one fixed number field. We also retain the original height-one
mountain under the accumulated change of basis. The following lemma
records these invariants and transfers the final classification back to
the initial signature set.
\begin{lemma}[Invariants under adding signatures and changing basis]
		\label{lem:normalization-ledger}
		Let \(\cF_0\) be a finite preprocessed algebraic signature set, and let
		\(L_0\) be a
		number field containing its entries, their conjugates, and every algebraic
		constant fixed before the construction.  Consider a sequence
		\((A_i,\cF_i)\), in which every \(\cF_i\) is a finite preprocessed
		algebraic signature set, starting from the given \(\cF_0\) with
		\(A_0=I\), in which each step is one of the following.
		\begin{enumerate}
			\item Add signatures under the following conditions: set
			\(A_{i+1}=A_i\), choose a finite
			\(\cF_{i+1}\supseteq\cF_i\) with
			\[
			\KHolant(\cF_{i+1})\leT\KHolant(\cF_i),
			\]
			and require every newly added positive-arity signature \(h\) to
			satisfy
			\[
			h\doteq A_i^{\otimes\arity(h)}h_0
			\quad\text{for a tensor }h_0\text{ over }L_0.
			\]
			\item A global normalization: choose \(D_i\in GO(X)\), and set
			\[
			A_{i+1}=D_iA_i,
			\qquad
			\cF_{i+1}=D_i\cF_i,
			\]
			with the value rescaling from each transformed edge given by \cref{eq:GOX-global-scalar}.
		\end{enumerate}
		Then, at every index \(i\),
		\begin{equation}
			A_i\cF_0\subseteq\cF_i,
			\qquad
			A_i^{\mathsf T}XA_i=\mu_iX
			\quad(\mu_i\in\AlgNums^\times),
			\qquad
			\KHolant(\cF_i)\equivT\KHolant(\cF_0),
			\label{eq:normalization-ledger}
		\end{equation}
		and every positive-arity \(h\in\cF_i\) satisfies
		\[
		h\doteq A_i^{\otimes\arity(h)}h_0
		\quad\text{for a tensor }h_0\text{ over }L_0.
		\]
		
		At an index where every nonzero directly realizable ordered binary is
		nonsingular and has finite projective transfer order, let \(G(\cF_i)\)
		be the group of its projective transfer classes \([BX]\), and put
		\[
		\overline G_i=[A_i]^{-1}G(\cF_i)[A_i].
		\]
		Then \(\overline G_i\le\PGL_2(L_0)\).  On any consecutive sequence
		of such indices, a global normalization leaves \(\overline G_i\) unchanged,
		whereas adding signatures can only enlarge it.  Thus all groups compared in
		the finite-group argument lie in the single fixed field \(L_0\). Any generic
		interpolation that leaves this field is used only when it completes the
		complexity classification of the current case.
		
		If \(\Tract(K\cF_i)\) holds, then \(\Tract(K\cF_0)\) holds.
		Separately, if one matrix \(V\) satisfies the \(K\)-coordinate
		conditions of \cref{def:common-transformability} for
		\(\cF_i\) and \(\cC\), where \(\cC\) is closed under nonzero scalar multiplication,
		then
		\[
		W=A_i^{-1}V
		\]
		satisfies the same \(\cC\)-conditions for \(\cF_0\).  In particular, this
		applies to \(\cA,\cP,\cL\).
	\end{lemma}
	
	\begin{proof}
		This is \cite[Lemma~2.26]{Liu2026DEQ2Holant}, applied inductively to the two
		displayed kinds of step.  When signatures are added, restriction supplies the reverse reduction.  A normalization conjugates the complete transfer group by
		\([D_i]\), which cancels when the group is expressed in the original basis using
		\([A_{i+1}]=[D_iA_i]\).
	\end{proof}
	
	A new factor can advance the proof even when the transfer group stays
	unchanged. To measure this progress, for a factor-closed signature set
	\(\cF\) containing an irreducible signature of arity at least three, put
	\begin{equation*}
		\nu(\cF)
		=\min\{\arity(h):h\in\cF
		\text{ is irreducible and }\arity(h)\ge3\}.
	\end{equation*}
	This parameter is used only after the branch in which every irreducible factor has
	arity at most two has been removed.

	\begin{definition}[Equivalent updates with strict progress]
		\label{def:height-one-outer-successor}
		Let \((A,\cF)\) satisfy \cref{lem:normalization-ledger} for the initial
		set \(\cF_0\) and its fixed number field \(L_0\).  An update to
		\((A^+,\cF^+)\) must satisfy the same invariant and use only the
		additions of signatures and global changes of basis allowed there.  In particular,
		\(\cF^+\) is finite, preprocessed, algebraic, and factor-closed, and
		\[
		\begin{gathered}
			A^+\cF_0\subseteq\cF^+,
			\qquad (A^+)^{\mathsf T}XA^+=\mu_{A^+}X\quad(\mu_{A^+}\ne0),\\
			\KHolant(\cF^+)\equivT\KHolant(\cF).
		\end{gathered}
		\]
		The transformed height-one mountain and every old irreducible signature
		remain members up to the same invertible change of basis and their recorded
		nonzero scalars.  Every nonzero directly realizable ordered binary over
		both sets is nonsingular with finite projective transfer order.  Put
		\[
		H=[A]^{-1}G(\cF)[A],
		\qquad H^+=[A^+]^{-1}G(\cF^+)[A^+].
		\]
		Both groups are subgroups of \(\PGL_2(L_0)\).  When both values of \(\nu\)
		are defined, the update is required to satisfy
		\[
		H\subseteq H^+,
		\qquad \nu(\cF^+)\le\nu(\cF),
		\qquad
		H\ne H^+\ \text{or}\ \nu(\cF^+)\ne\nu(\cF).
		\]
		We call such an update a \emph{strict update}. Thus at least one comparison is strict.  If no nonbinary irreducible
		signature remains, the arity-at-most-two alternative gives tractability
		and no further update is made.
	\end{definition}
	
	Finiteness of each individual group does not by itself bound the number
	of updates. The fixed number field gives a uniform bound on group order,
	which we combine with the decreasing arity parameter.

	\begin{lemma}[Termination of the strict updates]
		\label{lem:height-one-progress-well-founded}
		Fix the number field \(L_0\) in
		\cref{lem:normalization-ledger}.  There is no infinite sequence
		\[
		(H_0,\nu_0),(H_1,\nu_1),\ldots
		\]
		such that every \(H_i\) is a finite subgroup of \(\PGL_2(L_0)\), every
		\(\nu_i\) is an integer at least three, and every update satisfies
		\[
		H_i\subseteq H_{i+1},
		\qquad
		\nu_{i+1}\le\nu_i,
		\]
		with at least one of these two inequalities strict.  In the height-one
		construction these quantities are
		\[
		H_i=[A_i]^{-1}G(\cF_i)[A_i],
		\qquad
		\nu_i=\min\{\arity(f):f\in\cF_i
		\text{ is irreducible and }\arity(f)\ge3\}.
		\]
        If the second set is empty, the arity-at-most-two case is tractable and no
        further strict update is made.
	\end{lemma}
	
	\begin{proof}
		Write \(\delta=[L_0:\mathbb Q]\).  If
		\([T]\in\PGL_2(L_0)\) has projective order \(m\), the ratio of the two
		eigenvalues of \(T\) is a primitive \(m\)-th root of unity \(\zeta_m\).
		The eigenvalues lie in an extension of \(L_0\) of degree at most two, whence
		\[
		\varphi(m)=[\mathbb Q(\zeta_m):\mathbb Q]
		\leq [L_0(\zeta_m):\mathbb Q]\leq2\delta.
		\]
		There are therefore only finitely many possible projective element orders.
		Let \(M_0\) be their maximum.  By
		\cref{prop:finite-pgl2-classification}, every finite subgroup of
		\(\PGL_2(L_0)\) has order at most
		\[
		C_0=\max\{2M_0,60\}.
		\]
		
		For a sequence in the statement, define
		\[
		\Phi_i=(C_0-|H_i|)+(\nu_i-3)\in\mathbb Z_{\ge0}.
		\]
		The inclusions and inequalities make \(\Phi_i\) nonincreasing, and the strict
		condition decreases it by at least one at every update.  Hence the sequence
		has at most \(\Phi_0\) updates.  In the application,
		\cref{lem:normalization-ledger} places every group in the original basis in
		\(\PGL_2(L_0)\), adding signatures makes the groups nested, and
		\cref{def:height-one-outer-successor} makes \(\nu\) nonincreasing.  Factor
		closure between updates is finite because every nontrivial tensor split replaces
		a signature by positive-arity factors of strictly smaller arity. It is not
		part of the displayed measure.
	\end{proof}
	
	When neither a classification nor a strict update occurs, the remaining
	signatures will share a tractable basis. For either group with odd \(N\),
	we will show that the single matrix
	\(P\) fixed in \cref{eq:Walsh-coordinates} satisfies
	\begin{equation}
		(P^{-1})^{\otimes\arity(f)}f\in\cP
		\quad(f\in\cF),
		\qquad
		P^{\mathsf T}XP=2Z\in\cP.
		\label{eq:common-Walsh-product-presentation}
	\end{equation}
	These conditions show that \(K\cF\) is \(\cP\)-transformable
	by \cref{def:common-transformability}.
	
To analyze the remaining branch, we first process contractions and their
factors. The condition below describes the case in which these factors are
binary and their transfer classes lie in the current group. Fix a signature set
	\(\cF\) for which \(G(\cF)\) is defined, a
	subgroup \(H\le G(\cF)\), and one directly realizable ordered binary
	\(B_h\) with \([B_hX]=h\) for each \(h\in H\).  We will refer to the
	following contraction condition for \(f\in\cF\):
	\begin{enumerate}[label=\textnormal{(C)},leftmargin=*]
		\item\label{cond:height-one-contractions}
		For every ordered pair of distinct variables of \(f\), every chosen \(B_h\),
		and both kernel orientations \(XB_hX\) and \(XB_h^{\mathsf T}X\), the
		contraction is zero, or all its irreducible factors are
		nonsingular binaries \(B\) with \([BX]\in H\).
	\end{enumerate}
	Any nullary factors are absorbed into the overall scalar.
	The conditions are in the current \(K\)-coordinates. After a Walsh change
	of coordinates the group membership is conjugated by \(P\) as well.
	Applications below take \(H=G(\cF)\), with the complete transfer group from
	\cref{eq:complete-safe-group}.  Before using
	condition~\ref{cond:height-one-contractions}, any factor that gives
	hardness, tractability, or an update in
	\cref{def:height-one-outer-successor} is handled separately.  The next lemma specifies the finite family of outputs and factors retained
	during this processing. We denote its finite closure under these
contractions and factor extractions by \(\cF^{\mathrm{cl}}\).
	
	\begin{lemma}[Finite recursive contraction and factor extraction]
		\label{lem:height-one-card-saturation}
		Let \((A,\cF)\) be a stage of a sequence satisfying
		\cref{lem:normalization-ledger}, over its fixed number field \(L_0\).
		Suppose that \(\cF\) is finite,
		preprocessed, algebraic, and factor-closed, contains the transformed height-one mountain,
		and every nonzero directly realizable ordered binary over \(\cF\) is
		nonsingular and has finite projective transfer order.  Let
		\(H=G(\cF)\), and fix one directly realizable ordered-binary
		representative of each class in this finite group.
		
		Start with every nonbinary irreducible factor currently in \(\cF\).  For
		each factor \(f\), process every choice of an ordered pair of
		distinct variables, a fixed representative of a class in \(H\), and either
		kernel orientation.  Add the resulting nonzero contraction of two variables and, through
		\cref{lem:factor-saturation}, all of its irreducible factors.
		Apply the same procedure to each new nonbinary factor.  This
		finite recursive procedure reaches one of the following outcomes:
		\begin{enumerate}
			\item \(\KHolant(\cF)\) is \(\SharpP\)-hard or
			\(\Tract(K\cF)\) holds.
            \item there is a strict update \((A,\cF)\mapsto(A^+,\cF^+)\)
            satisfying \cref{def:height-one-outer-successor}.
			\item there is a finite preprocessed algebraic factor-closed set
			\(\cF^{\mathrm{cl}}\supseteq\cF\), with
			\[
			\KHolant(\cF^{\mathrm{cl}})\equivT\KHolant(\cF),
			\qquad G(\cF^{\mathrm{cl}})=H,
			\]
			for which the coordinate matrix remains \(A\).  Every indicated
			contraction of two variables has been processed, every
			nonbinary irreducible factor supplied by factor closure is retained in
			\(\cF^{\mathrm{cl}}\), and every binary factor it produces is
			nonsingular with transfer class in \(H\).  A nonbinary factor produced
			from a parent \(f\) has arity at most \(\arity(f)-2\).
		\end{enumerate}
        Outcome~3 is closure under all the indicated contractions and factor
        extractions, with the same group and coordinate matrix. It may preserve
        \(\nu\), so it need not be a strict update.
	\end{lemma}
	
	\begin{proof}
		Begin with the finitely many nonbinary irreducible factors in
		\(\cF\).  For each factor, process its finite set of
		ordered-variable, representative, and orientation choices.  For each choice,
		add the resulting contraction and all irreducible factors
		under \cref{eq:finite-equivalent-enlargement,eq:factor-saturation}.
		The transformed mountain rules out the common-one-sided factor exception.
		
		An odd-arity factor gives outcome~1 by \cref{thm:external-odd,eq:K-equivalence}.
		A singular binary, or a nonsingular binary whose transfer has infinite
		projective order, is handled by \cref{lem:binary-interpolation,lem:common-unary-exit}.  After each addition, consider all directly realizable
		nonzero ordered binaries.  If any is singular or its transfer has infinite
		projective order, the preceding argument gives outcome~1.
		Otherwise the group is defined, and an element outside \(H\) strictly
		enlarges the group in the original basis,
		and a new nonbinary irreducible factor of arity below \(\nu(\cF)\)
		strictly decreases \(\nu\).  These give outcome~2, with both comparisons
		preserved as in \cref{def:height-one-outer-successor}.
		If neither occurs, retain all the binary factors and repeat the procedure
		on each new nonbinary irreducible factor.
		
		If no classification or strict update occurs, the group and its chosen
		representatives stay fixed. There are finitely many initial factors and
		finitely many choices at each step. Every newly produced nonbinary factor
		has arity at most the arity of
		its parent minus two.  The recursive depth is therefore bounded, so only
		finitely many factors and contraction choices occur.  If neither of the
		first two outcomes occurs, processing all these choices gives outcome~3.
		Each addition keeps the transformed mountain in the set and satisfies the fixed-field and scalar conditions of \cref{lem:normalization-ledger}. \end{proof}

Thus the recursion terminates by arity descent, while
\cref{lem:height-one-progress-well-founded} bounds the number of strict
updates between such recursions.
	
	\subsection{Contractions for the odd cyclic and dihedral groups}
	\label{subsec:residual-phase-toolkit}

	For tensors \(A,B\) of the same arity, a pencil is a family \(q_s=A-sB\).
	Its parameter range is specified at each use.

	We now turn the contraction condition into restrictions on the signatures
	themselves. The next two lemmas constrain four coefficients at a time and
	then the parity of the whole Walsh support. Both apply in the cyclic and
	dihedral cases.
We use the Walsh coordinates \(\widetilde f\) defined in
\cref{eq:Walsh-coordinates}.
	Each contraction below is realized by an \(X\)-edge gadget.
The matrix \(P\) expresses the same output in Walsh coordinates, where
the dependence on the phase is linear. Representative scalars are retained
in the gadget reduction.
	
	\subsubsection{Contraction kernels and the four-coefficient relation}
	\label{subsec:odd-dih-cards}
	
	Fix odd \(N\ge3\).  With \(R(s)\) as in
	\cref{thm:finite-group-split}, the rotation and reflection classes are
	\begin{equation}
		[R(s)],\qquad[ZR(s)],
		\qquad s\in\mu_N.
		\label{eq:odd-dih-native-transfers}
	\end{equation}
	
	The Walsh forms are
	\[
	P^{-1}R(s)P=D(s),
	\qquad P^{-1}ZP=X.
	\]
	
	For binary chains \(L,R\), with free inputs indexing rows and joined
	inputs indexing columns, joining their second inputs through one \(X\)-edge
	gives \(LXR^{\mathsf T}\).  Fix a direct binary-chain representative of
	each projective class in \cref{eq:odd-dih-native-transfers}.  Contracting
	two variables then gives the following Walsh kernels,
	up to a nonzero scalar:
	
	\begin{equation}
		\mathsf K^+(s)=D(-s)=\begin{pmatrix}1&0\\0&-s\end{pmatrix},
		\qquad
		\mathsf K^-(s)=D(-s)X=\begin{pmatrix}0&1\\-s&0\end{pmatrix},
		\qquad s\in\mu_N.
		\label{eq:odd-dih-kernels}
	\end{equation}
	
	We call contraction with \(\mathsf K^+(s)\) a \emph{phase contraction}
and contraction with \(\mathsf K^-(s)\) a \emph{reflection contraction}.
Their nonzero entries select equal and unequal bits, respectively.
They correspond to the rotation and reflection classes above.
Phase contractions occur in both groups, while reflection contractions occur
only in the dihedral group. In Walsh coordinates, swapping the two inputs acts on a
transfer matrix \(T\) by
\[
P^{-1}(PTP^{-1})^\#P=ZT^{\mathsf T}Z.
\]
This map fixes \(D(s)\) and sends \(XD(s)\) to \(-sXD(s^{-1})\).
It therefore preserves the projective rotation and reflection classes.
Attaching binary representatives to the variables adds no kernels, so
\cref{eq:odd-dih-kernels} lists every oriented kernel needed below.
	
	The next coefficient relation and lemma hold for every integer
	\(N\ge3\). For such an \(N\) and
	\(a,b,c,d\in\AlgNums\), let \(\mathcal R_N(a,b,c,d)\) mean that, for some
	\(\alpha\in\mu_N\), one of the following three systems holds:
	
	\begin{align}
		c&=-\alpha a,&d&=-\alpha b,
		&&\text{(proportional)},
		\label{eq:odd-dih-R-proportional}\\
		b&=c=0,&d&=\alpha a,
		&&\text{(diagonal monomial)},
		\label{eq:odd-dih-R-diagonal}\\
		a&=d=0,&c&=\alpha b,
		&&\text{(anti-diagonal monomial)}.
		\label{eq:odd-dih-R-antidiagonal}
	\end{align}

	\begin{lemma}[Phase-pencil relation]
		\label{lem:odd-dih-RN}
		For every integer \(N\ge3\), the relation
		\(\mathcal R_N(a,b,c,d)\) is equivalent to the following
		condition: for every \(s\in\mu_N\),
		\[
		\diag(a-sb,c-sd)Z
		\quad\text{is zero or is proportional to }D(\eta)
		\text{ for some }\eta\in\mu_N.
		\]
	\end{lemma}
	
	\begin{proof}
		This is \cite[Lemma~7.1]{Liu2026DEQ2Holant}.
	\end{proof}
	
	\subsubsection{Support parity}
	\label{subsec:odd-dih-parity}
	
	\begin{lemma}[Walsh parity from contractions]
		\label{lem:odd-dih-all-arity-parity}
		Let \(N\ge3\) be odd, and let \(H\) be the cyclic or dihedral group
		in Walsh coordinates, respectively
		\[
		H=\{[D(s)]:s\in\mu_N\}
		\quad\text{or}\quad
		H=\{[D(s)],[XD(s)]:s\in\mu_N\}.
		\]
		Let \(f\ne0\) have positive even arity. At arity two, suppose the
		Walsh transfer class \([P^{-1}(fX)P]\) belongs to \(H\).  At arity at least four, suppose that every irreducible factor of every nonzero contraction of \(f\) using an ordered pair of distinct
		variables, a chosen binary representative with Walsh transfer class in \(H\), and either kernel orientation is supported on
		one Walsh parity class. In the cyclic case suppose that class is even.  Then
		\(\widetilde f\) is supported on one parity class, which is even in the
		cyclic case.
	\end{lemma}
	
	\begin{proof}
		At arity two, a Walsh binary is, up to a nonzero scalar, a diagonal or
		anti-diagonal transfer times \(Z\), and hence has one parity.  Only the
		diagonal alternative occurs in the cyclic group.  At larger arity, the
		tensor product of the irreducible factors of a nonzero contraction has one parity. In the
		cyclic case it is even.
		
		At arity four, each \(\mathsf K^+(s)\)-contraction is zero, diagonal, or
		anti-diagonal, and every entry is affine in \(s\).  Among three phases,
		either two contractions vanish, annihilating the pencil, or two nonzero contractions have
		the same support parity, which forces the opposite entries to vanish
		identically.  In the only remaining pattern---one zero, one diagonal, and
		one anti-diagonal contraction---each opposite pair acquires two roots, with the
		same conclusion.  Hence every equal-bit slice has one parity.
		For two supported strings \({\mathbf x},{\mathbf y}\) of opposite parity,
        there are positions \(i\ne j\) with \((x_i,y_i)=(x_j,y_j)\),
        unless the four pairs \((x_k,y_k)\), \(k\in[4]\), exhaust \(\Ftwo^2\).
        In the latter case both strings have weight two, contradicting their
        opposite parities. In the former case, deleting positions \(i,j\)
        contradicts the parity of the equal-bit slices.  In the cyclic case, odd Walsh parity would
		realize a reflection and is therefore excluded.
		
		For arity at least six, fix distinct variables \(i,j\).  The phase contraction is
		\[
		\widetilde h_s({\mathbf y})=\lambda_s\bigl(\widetilde f(00,{\mathbf y})
		-s\widetilde f(11,{\mathbf y})\bigr),
		\qquad \lambda_s\ne0.
		\]
		A zero contraction makes all others proportional.  Otherwise two phases have the
		same contraction parity.  Every coordinate in the opposite parity then
		vanishes at two distinct phases, so its affine coefficient vanishes
		identically. The two nonzero phases also recover the slices \(A,B\) with
		the common parity.
		
		In the cyclic case, delete two equal coordinates of any supported string and
		choose a phase other than the at most one value cancelling its residual
		coefficient. The resulting Walsh contraction has even-parity support. In the dihedral case,
        two supported strings \({\mathbf x},{\mathbf y}\) of opposite parity at
        arity at least six have positions \(i\ne j\) with
        \((x_i,y_i)=(x_j,y_j)\).  At most one phase cancels each residual
		coefficient, so \(|\mu_N|\ge3\) leaves a phase preserving both strings.
		Deleting the chosen pair preserves their opposite parities, contradicting
		the common slice parity.
	\end{proof}
	
	For every \(f\) covered by \cref{lem:odd-dih-all-arity-parity}, Walsh
	duality gives the equivalent complement relation in the \(K\)-coordinates,
	\[
	f(\bar {\mathbf x})=f({\mathbf x})
	\quad\text{or}\quad
	f(\bar {\mathbf x})=-f({\mathbf x}),
	\]
	
	with only the plus sign in the cyclic subcase.

\subsubsection{Quaternary classification under phase contractions}
	Fix \(N\ge3\) and a quaternary signature \(q\) in the current set, and put
	\(\widetilde q=(P^{-1})^{\otimes4}q\).  When \(\widetilde q\) is supported
	on the even Hamming-weight positions in Walsh coordinates, write
	
	\begin{equation}
		p=\widetilde q_{0000},
		\qquad t=\widetilde q_{1111},
		\qquad x_{ab}=\widetilde q_{\one_{\{a,b\}}}.
		\label{eq:odd-dih-q4-coordinates}
	\end{equation}

For every ordering \((i,j,k,l)\) of \([4]\), the phase-contraction matrix is
\[
H_{ij}^+(s)=\begin{pmatrix}p-sx_{ij}&0\\0&x_{kl}-st\end{pmatrix}.
\]
When its nonzero Walsh transfer classes belong to \(C_N\),
\cref{lem:odd-dih-RN} gives \(\mathcal R_N(p,x_{ij},x_{kl},t)\).

	\begin{lemma}[Quaternary classification under phase contractions]
		\label{lem:odd-cyc-q4-family}
		Let \(N\ge3\), let \(\cF\) be a finite algebraic signature set in
		\(K\)-coordinates, and, for every \(s\in\mu_N\), fix a gadget-realizable ordered binary representative \(C_s\) over \(\cF\), in a specified orientation,
		whose Walsh matrix is
		\[
		\widetilde C_s=\kappa_sD(s)Z,
		\qquad \kappa_s\in\AlgNums^\times.
		\]
		Let \(q\in\cF\) be quaternary, and
		suppose \(\widetilde q=(P^{-1})^{\otimes4}q\) is supported on the even Hamming-weight
		positions in Walsh coordinates.  Define \(p,t,x_{ab}\) by
		\cref{eq:odd-dih-q4-coordinates}.  Assume that, for every contracted variable pair \(\{i,j\}\),
		both orders of the residual pair \(\{k,l\}=[4]\setminus\{i,j\}\),
		and every \(s\in\mu_N\), contraction
		with \(C_s\) gives either the zero contraction or a nonsingular binary contraction whose
		Walsh transfer class belongs to
		\[
		C_N=\{[D(\eta)]:\eta\in\mu_N\}.
		\]
		Organize the eight
		coordinates into four complementary pairs.  Put \(B_0=(p,t)\) and
		\(\mathfrak I=\{\{1,2\},\{1,3\},\{1,4\}\}\).
		For \(I=\{a,b\}\in\mathfrak I\), set \(x_I=x_{ab}\),
		put \(I^c=[4]\setminus I\), and define
		the other three complementary pairs \(B_I=(x_I,x_{I^c})\),
        \(I\in\mathfrak I\). Then exactly one of the following holds:
		\begin{enumerate}
			\item If \(p,t\ne0\), then \(\lambda=t/p\in\mu_N\), and each pair \(B_I\)
			is zero or
			\begin{equation}
				B_I=(-p\lambda\alpha_I^{-1},-p\alpha_I),
				\qquad \alpha_I\in\mu_N.
				\label{eq:odd-cyc-q4-nonzero-endpoints}
			\end{equation}
			\item If \(p=t=0\), each pair \(B_I\) is zero or
			\[
			B_I=\rho_I(1,\alpha_I),
			\qquad \rho_I\ne0,
			\quad \alpha_I\in\mu_N.
			\]
		\end{enumerate}
		In particular, exactly one of \(p,t\) cannot be nonzero.
		Conversely, every even Walsh coefficient array in these alternatives satisfies the
		stated contraction condition under every phase contraction.
	\end{lemma}
	
	\begin{proof}
		This is \cite[Lemma~7.2]{Liu2026DEQ2Holant}.
	\end{proof}
	
	\subsection{The odd-dihedral case}
	\label{proofsec:odd-dihedral}

	\begin{theorem}[The odd-dihedral case]
		\label{thm:odd-dihedral}
		Let \((A,\cF)\) be a stage of the sequence in
		\cref{lem:normalization-ledger} at which \(\cF\) is finite, preprocessed,
		algebraic, and factor-closed and contains the transformed height-one
		mountain.  Suppose that, after the normalization of
		\cref{thm:finite-group-split}, every nonzero directly realizable binary is
		nonsingular with finite projective transfer order, and
		\(G(\cF)\) is
		\[
		D_{2N}=\{[R(s)],[ZR(s)]:s\in\mu_N\}
		\]
		for an odd \(N\ge3\).
		Then at least one of the following holds.
		\begin{enumerate}
			\item \(\KHolant(\cF)\) is \(\SharpP\)-hard or \(\Tract(K\cF)\) holds.
            \item There is a strict update \((A,\cF)\mapsto(A^+,\cF^+)\)
            satisfying \cref{def:height-one-outer-successor}.
			\item There is an equivalent finite preprocessed algebraic factor-closed
			set \(\cF^{\mathrm{cl}}\supseteq\cF\), in the same
			coordinates and with the same matrix \(A\), for which the fixed
			Walsh matrix \(P\) satisfies
			\[
			(P^{-1})^{\otimes\arity(f)}f\in\cP
			\quad(f\in\cF^{\mathrm{cl}}),
			\qquad
			P^{\mathsf T}XP=2Z\in\cP.
			\]
			Hence \(K\cF\) is \(\cP\)-transformable.
		\end{enumerate}
	\end{theorem}
	
	In the remaining dihedral branch, choose a nonbinary irreducible
signature of minimum arity after the finite contraction procedure. Its
nonzero contractions factor into binary signatures whose transfer classes
lie in the group. We show that this forces a known classification at arity four or contradicts
irreducibility at larger arities. The four-variable analysis controls
support and coefficient ratios, the six-variable argument uses the finite
classification in \cref{thm:cert-odd-dih-q6-interface}, and matching pencils
supply a common binary factor at higher arities.
	
	\subsubsection{Support and coefficients on even blocks}
	\label{subsec:odd-dih-phase-product-class}
	
For a partition \(\pi\) of a finite variable set \(V\) into even blocks,
using \(\one_C\) for the incidence vector of \(C\subseteq V\), define
	\[
	L_\pi=\Span_{\Ftwo}\{\one_C:C\in\pi\}.
	\]
	Its cosets are precisely the sets obtained by toggling each block as one
	unit.
	
	The class \(\mathcal P_N(V)\) consists of the zero tensor and
	all tensors \(g\) for which there exist such a partition \(\pi\), an offset
	\({\mathbf a}\in\Ftwo^V\), a scalar \(c\in\AlgNums\setminus\{0\}\), and block ratios
	\(r_C\in(-1)^{|C|/2}\mu_N\) such that
	\begin{align*}
		\supp(g)&={\mathbf a}+L_\pi,
		\\
		g\!\left({\mathbf a}+\sum_{C\in\pi}t_C\one_C\right)
		&=c\prod_{C\in\pi}r_C^{t_C}
		\qquad\bigl((t_C)_{C\in\pi}\in\Ftwo^\pi\bigr).
	\end{align*}
	For \(C\subseteq V\), write \({\mathbf a}|_C\) for the restriction of \({\mathbf a}\) to \(C\),
	and \(\overline{{\mathbf a}|_C}\) for its bitwise complement.
	Equivalently, with the tensor factors placed on their indicated variable sets,
	\[
	g=c\bigotimes_{C\in\pi}
	\bigl(\delta_{{\mathbf a}|_C}+r_C\delta_{\overline{{\mathbf a}|_C}}\bigr).
	\]
	The block variables \((t_C)\) therefore have logical coefficient tensor
	\(c\prod_C r_C^{t_C}\). The block ratios will let us track coefficients
	as well as support through contractions.  For nonzero \(g\), its support
	determines \(\pi\): the inclusion-minimal nonzero vectors of \(L_\pi\) are
	precisely the block indicators.  We omit \(V\) when clear.  The class
	\(\mathcal P_N\) differs from the standard product class \(\cP\) of
	\cref{subsec:standard-classes}. Its even-block and coefficient-ratio
	restrictions remain part of its definition.
	
	\begin{lemma}[Closure of \(\mathcal P_N\) under contractions]
		\label{lem:phase-product-closure}
		The class \(\mathcal P_N\) is closed under nonzero scalar multiplication,
		variable permutation, tensor product on disjoint variable sets, and
		contraction of two variables against any kernel
		\(\mathsf K^+(s)\), \(\mathsf K^-(s)\), or its transpose from
		\cref{eq:odd-dih-kernels}, where \(s\in\mu_N\).  Consequently, if
		every irreducible factor of a contraction has Walsh tensor in
		\(\mathcal P_N\), then the contraction itself has Walsh tensor in
		\(\mathcal P_N\).
	\end{lemma}
	
	\begin{proof}
		For a tensor product, take the disjoint union of the two even-block
		partitions and concatenate their offsets and block ratios.  Scalars and variable
		permutations merely change the leading coefficient or relabel that data.
		
		Consider a contraction of two variables of a nonzero tensor in \(\mathcal P_N\).  If the
		variables lie in the same block \(C\), the kernel with the wrong xor gives zero.
		For the supported xor, delete the two variables.  When variables remain in \(C\), its
		new block ratio is
		\[
		-s r_C\quad\hbox{or}\quad-r_C/s,
		\]
		and therefore belongs to
		\((-1)^{(|C|-2)/2}\mu_N\).  If \(|C|=2\), the block instead contributes a
		scalar, possibly zero, which is absorbed into the leading coefficient.
		
		If the variables lie in distinct blocks \(C,D\), the kernel affinely identifies
		their two block variables.  After complementing one residual offset when
		needed, replace those blocks by
		\[
		(C\setminus\{i\})\mathbin{\dot\cup}(D\setminus\{j\}).
		\]
		Its ratio is one of
		\[
		-s^{\pm1}r_Cr_D^{\pm1}
		\in(-1)^{(|C|+|D|-2)/2}\mu_N.
		\]
		All other blocks are unchanged.  Transposing \(\mathsf K^-(s)\) only supplies
		a nonzero scalar and replaces \(s\) by \(s^{-1}\), so it is covered as well.
		Thus every nonzero contraction is again in \(\mathcal P_N\).
		The last assertion follows from factorization into irreducibles. Nullary factors are
		absorbed into the scalar.
	\end{proof}
	
	\subsubsection{Quaternary support and coefficients}
	\label{subsec:odd-dih-q4}

Use the even-parity coordinates in \cref{eq:odd-dih-q4-coordinates}.
For every ordering \((i,j,k,l)\) of \([4]\), the reflection-contraction matrix is
\[
H_{ij}^-(s)=\begin{pmatrix}0&x_{jl}-sx_{il}\\x_{jk}-sx_{ik}&0\end{pmatrix}.
\]
Condition~\ref{cond:height-one-contractions} for the dihedral group gives
\(\mathcal R_N(p,x_{ij},x_{kl},t)\) from the phase contraction and
\(\mathcal R_N(x_{jk},x_{ik},x_{jl},x_{il})\) from the reflection contraction.

	When \(\widetilde q\) is supported on the odd Hamming-weight positions in Walsh coordinates, write
	\[
	u_i=\widetilde q({\mathbf e}_i),
	\qquad v_i=\widetilde q(\one^4\oplus {\mathbf e}_i).
	\]
	
	For every such variable ordering, condition~\ref{cond:height-one-contractions}
	for the dihedral group gives
	\[
	\mathcal R_N(u_k,v_l,u_l,v_k),
	\qquad
	\mathcal R_N(u_j,u_i,v_i,v_j).
	\]
	
	\begin{proposition}[Dihedral quaternary classification]
		\label{prop:odd-dih-q4}
		Let \(N\ge3\) be odd, and let \(q\in\cF\) be a nonzero quaternary
		satisfying condition~\ref{cond:height-one-contractions} with
		\(H=G(\cF)=\{[R(s)],[ZR(s)]:s\in\mu_N\}\), with \(R(s)\) as in \cref{thm:finite-group-split}.  Put
		\(\widetilde q=(P^{-1})^{\otimes4}q\).  Then
		\(\widetilde q\in\mathcal P_N([4])\).  More precisely,
		\(\widetilde q\) has one Walsh parity and exactly one or two complementary
		pairs whose coefficient pairs are nonzero.  After choosing an orientation, there is a scalar
		\(\rho\in\AlgNums\setminus\{0\}\) such that, for one nonzero complementary
		pair \(\{{\mathbf x},\bar {\mathbf x}\}\), its nonzero coefficients in the order
		\((\widetilde q_{\mathbf x},\widetilde q_{\bar {\mathbf x}})\) are
		
		\begin{equation}
			\rho(1,\delta),
			\qquad \delta\in\mu_N,
			\label{eq:odd-dih-q4-one-block}
		\end{equation}
		
		or, for two nonzero complementary pairs \(\{{\mathbf x},\bar {\mathbf x}\}\) and
		\(\{{\mathbf y},\bar {\mathbf y}\}\), its nonzero coefficients in the order
		\((\widetilde q_{\mathbf x},\widetilde q_{\bar {\mathbf x}},
		\widetilde q_{\mathbf y},\widetilde q_{\bar {\mathbf y}})\) are
		
		\begin{equation}
			\rho(1,\delta,-\gamma,-\delta\gamma^{-1}),
			\qquad \delta,\gamma\in\mu_N.
			\label{eq:odd-dih-q4-two-block}
		\end{equation}
		
		Conversely, every Walsh coefficient array in
		\cref{eq:odd-dih-q4-one-block,eq:odd-dih-q4-two-block} satisfies the
		stated contraction condition under all phase and reflection contractions.
	\end{proposition}
	
	\begin{proof}
		Phase contractions give a \(\mu_N\)-ratio on every nonzero complementary
		pair.  If \(p,t\ne0\), \cref{eq:odd-cyc-q4-nonzero-endpoints} puts all nonzero central
		coefficients in \(-p\mu_N\).  For two nonzero central coefficient pairs, the
		\(2\)-by-\(2\) coefficient matrix of the reflection contraction has all
		four entries nonzero, and its only possible proportional alternative
		in \cref{lem:odd-dih-RN} requires a ratio in \(-\mu_N\), not \(\mu_N\).
		The cosets are disjoint for odd \(N\).
		
		With zero endpoints, three nonzero central coefficient pairs would require three
		pairwise ratios in \(-\mu_N\), although the product of two lies in
		\(\mu_N\).  Thus at most two complementary pairs have nonzero coefficients.  Applying the same argument
		to both orientations of the reflection kernels
		\(\mathsf K^-(s)\) gives the identical conclusion on the odd Hamming-weight positions.
		
		Substitution gives \cref{eq:odd-dih-q4-one-block,eq:odd-dih-q4-two-block}
		and proves their converse, including zeros.  The one-pair case uses \(\pi=\{[4]\}\). The two-pair case forms a \(2+2\) partition coset with ratios
		\(-\gamma,-\delta\gamma^{-1}\in-\mu_N\), whose product is \(\delta\).
		Thus both parities lie in \(\mathcal P_N([4])\), with odd parity absorbed
		by the coset offset.
		
		The one-pair case is a CP signature.  Up to flips and permutation, a
		support consisting of two complementary pairs is \(\{0000,1111,1100,0011\}\), and
		\cref{eq:odd-dih-q4-two-block} gives
		\[
		\widetilde q_{\mathbf x}\widetilde q_{\bar {\mathbf x}}
		=\widetilde q_{\mathbf y}\widetilde q_{\bar {\mathbf y}},
		\]
		
		for the two complementary pairs \(\{{\mathbf x},\bar {\mathbf x}\}\) and
		\(\{{\mathbf y},\bar {\mathbf y}\}\) in the support.  Hence it is a product of two
		nonsingular weighted equality/disequality
		binaries.  Thus \(\widetilde q\in\cP\), as required by
		\cref{eq:common-Walsh-product-presentation}.
	\end{proof}
	
	\subsubsection{Products of binaries and the arity-six case}
	\label{subsec:odd-dih-q6}
	
Given a finite variable set \(V\), a perfect matching \(M\) of \(V\), and
	\({\boldsymbol{\epsilon}}\in\Ftwo^M\), the corresponding \emph{matching coset} is
	\[
	C(M,{\boldsymbol{\epsilon}})=
	\{{\mathbf x}\in\Ftwo^V:x_u\oplus x_v=\epsilon_{\{u,v\}}
	\text{ for every }\{u,v\}\in M\}.
	\]
Equivalently, \(C(M,{\boldsymbol{\epsilon}})={\mathbf a}+L_M\), where \({\mathbf a}\in\Ftwo^V\) satisfies
\(a_u\oplus a_v=\epsilon_{\{u,v\}}\) for every matched pair.
Choose an order \((u,v)\) for each pair and use \(x_u\) as its independent
bit, so \(x_v=x_u\oplus\epsilon_{\{u,v\}}\). The coefficients indexed by
these bits form the \emph{logical coefficient tensor}.
The all-zero and all-one choices of \({\boldsymbol{\epsilon}}\) impose
equal and unequal bits on every matched pair, respectively.
The matching support alone does not impose a rank-one condition on this tensor.

	For the minimum-arity argument, contractions must be products of binaries.
	Write \(\mathcal B_N(V)\subseteq\mathcal P_N(V)\) for zero together with
	the tensors whose partition consists entirely of pairs.  Every nonzero
	member is a product of nonsingular Walsh equality/disequality binaries,
	with ratio in \(-\mu_N\) on each pair.  For the resulting perfect matching
	\(M\), its support is \({\mathbf a}+L_M\), is complement-closed, and has
	cardinality \(2^{|V|/2}\).
	The contraction formulas in \cref{lem:phase-product-closure} show that
	\(\mathcal B_N\) is itself closed under all phase and reflection contractions:
	joining different pairs makes one pair, and closing a pair gives a scalar.
	
	\begin{theorem}[Arity-six factorization from contractions]
		\label{thm:odd-dih-q6-base}
		Let \(N\ge3\) be odd.  If every nonzero oriented phase or reflection contraction
		of an arity-six signature \(f\) has Walsh tensor in \(\mathcal B_N([4])\),
		then \(f=0\) or \(f\) is a tensor product of three nonsingular binaries.
	\end{theorem}
	
	\begin{proof}
		Suppose \(f\ne0\), and write \(F=\widetilde f\).  By
		\cref{lem:odd-dih-all-arity-parity}, \(F\) has one parity.  Fix a variable pair
		and fix \(x_i\oplus x_j=\xi\), with the two corresponding slices \(A,B\), so the normalized contractions are \(q_s=A-sB\).  Put \(U=\supp(A)\cup\supp(B)\).
		Each nonzero coordinate remains nonzero at at least \(N-1\) phases, while every
		nonzero contraction has four support points.  Thus
		\[
		(N-1)|U|\le4N.
		\]
		Moreover \(U=\bigcup_s\supp(q_s)\) is complement-closed.  For \(N\ge5\)
		it is therefore empty or consists of at most two complementary pairs.  For
		\(N=3\), the same conclusion is exactly
		\cref{thm:cert-odd-dih-q6-interface}\textnormal{(i)}.  Part~\textnormal{(ii)}
		of that theorem now applies for every odd \(N\): the
		support of \(F\) is a partition coset of type \((6)\), \((4+2)\), or
		\((2+2+2)\), or belongs to the orbit of
		\begin{equation}
			S_\star=\{000000,000011,001100,110011,111100,111111\}.
			\label{eq:odd-dih-q6-nonproduct-support}
		\end{equation}
		Here and below, bit relabelings describe support coordinates only.  They
		permute the available xor kernels projectively, with possible phase
		inversion, so all contraction arguments cover the full indicated orbits.
		
		A \((6)\) support has a supported-xor contraction with exactly two
		nonzero entries, each with a unique lift, meaning a unique extension of the
		residual assignment to the original support consistent with the kernel.  This contradicts the four-point
		support required by \(\mathcal B_N([4])\).  On a \((4+2)\) support, use
		block variables and denote the four nonzero coefficients by \(c_{uv}\).
		Closing the binary block gives, after normalizing its orientation,
		\[
		(c_{00}-s c_{01})\delta_{0000}
		+(c_{10}-s c_{11})\delta_{1111}.
		\]
		At most two phases are forbidden by cancellation. Since \(N\ge3\), some
        contraction again has exactly two nonzero entries. On \(S_\star\), contract variables
		\(3,5\) with \(\mathsf K^+(s)\).  Only the two endpoints have equal bits
		there, giving
		\[
		F_{000000}\delta_{0000}-sF_{111111}\delta_{1111},
		\]
		again with two nonzero entries and unique lifts.
		
		Only the \((2+2+2)\) support remains.  Its eight coefficients
		\(c_{uvw}\) are nonzero.  Closing any one of its binary blocks leaves a
		\(2\)-by-\(2\) table whose determinant vanishes at every \(s\in\mu_N\):
		a nonzero contraction has the same four-point matching support, and a zero contraction
		also has zero determinant.  For example,
		\[
		\det\bigl(c_{uv0}-s c_{uv1}\bigr)_{u,v\in\Ftwo}=0.
		\]
		This polynomial has degree at most two and at least three roots.  Its
		constant and quadratic coefficients show that the two faces obtained by fixing that coordinate have
		rank one.  Doing this for all three axes makes every axis face minor zero.
		Since all eight entries are nonzero, the ratios along each axis are
		independent of the other two bits, so \(c_{uvw}=c r_1^u r_2^v r_3^w\).
		This is the required three-binary factorization, which pulls back through
		the Walsh change of coordinates.
	\end{proof}
	
	\subsubsection{Matching pencils}
	\label{subsec:odd-dih-support-circuits}

	At higher arities, binary factorizations of separate contractions must
	be compared as the phase varies. We call a pencil \(A-sB\) \emph{projectively nonconstant} when
	\(A,B\) are linearly independent.  Otherwise it is zero or has constant
	projective class wherever it is nonzero.
A projective line consists of the classes in a two-dimensional linear span.
Here rank one is taken over the matching blocks: it means that the logical
coefficient tensor is a product of one vector for each block.
A \emph{one-factor Segre ruling} is a linear family of such tensors in which
all but one block factor stay fixed.
A matrix flattening uses one block bit as the row index and the remaining
block bits as the column index.  The next two
results use this fixed-factor structure to pass from individual contraction
outputs to a common binary factor of the original tensor.
	
	\begin{lemma}[Three matching products on a line]
		\label{lem:odd-dih-three-coset-circuit}
		Let three pairwise nonproportional tensor products of nonsingular diagonal
		or anti-diagonal binaries on the same \(m\) variables satisfy a linear relation
		with all coefficients nonzero.  Either their supports coincide and
		the three tensors lie in a single one-factor Segre ruling, or their matchings differ only on four variables, where the three
		different pairings occur.  In the latter case all binary factors outside
		those four variables are common up to nonzero scalars.
	\end{lemma}
	
	\begin{proof}
		Let the supports be \(S_0,S_1,S_2\), matching cosets of dimension \(m/2\).
		No point can belong to just one support.  If two supports coincide, the
		third is contained in them and has the same cardinality, so all coincide.
		On their common block coordinates, every \(2\times2\) minor of a matrix flattening has degree at
		most two along the line.  Three rank-one points make these minors vanish
		identically.  A line of rank-one tensors varies in at most one factor: if
		two factors of two of its points were nonproportional, flattening across
		one of them would give rank two at a generic third point.  This is a single one-factor Segre ruling.
		
		Otherwise each \(S_i\) is the union of its intersections with the other
		two supports.  Two proper affine subspaces cover an affine space only when
		they are disjoint parallel hyperplanes: each has at most half its points,
		so both must attain that bound.  Thus the pairwise intersections are
		hyperplanes and the triple intersection is empty.  The union of two
		perfect matchings is a disjoint union of alternating even cycles, counting
		a common edge as a two-cycle.  If there are \(k=m/2\) edges and \(b\)
		components, the intersection of their translation spaces has dimension \(b\).  Here
		\(b=k-1\), so exactly one component is a four-cycle and all other edges
		coincide.  The offsets on those common edges agree because the supports
		intersect.  After removing their free binary coordinates, \(S_2\) is the
		four-point symmetric difference of the other two local supports. Its
		matching is the third pairing on the same four variables.  Finally, on each
		pairwise intersection only two tensors contribute to the dependence.
		Fixing one assignment to four variables with nonzero coefficient leaves the full remaining matching coset, on which their products are proportional.  Uniqueness of the factors
		on disjoint pairs gives the asserted common factors.
	\end{proof}
	
	\begin{proposition}[Phase-independent matching and common factors]
		\label{prop:odd-dih-support-circuit-rigidity}
		Let \(N\ge3\) be odd, and let \(f\) have even arity \(n\ge8\).  Suppose
		every nonzero result of one or more successive contractions of two variables using
		the kernels in
		\cref{eq:odd-dih-kernels} has Walsh tensor in \(\mathcal B_N\).  Then each
		contraction pencil can be written as \(q_s=H\otimes h_s\), where \(H\)
		is a fixed product of nonsingular binaries on disjoint pairs and \(h_s\)
		is a binary tensor on one fixed remaining variable pair.  Zero members and proportional
		pencils are included.
	\end{proposition}
	
	\begin{proof}
		Write a pencil as \(q_s=A-sB\).  A zero member makes all other members
		proportional, which settles that case.  If the pencil is not projectively
		constant, any three distinct phases give a linear dependence with all three
		coefficients nonzero,
		so \cref{lem:odd-dih-three-coset-circuit} applies to these members of
		\(\mathcal B_N\).  In its common-support
		case, the \(2\times2\) flattening minors vanish identically, giving the fixed
		one-factor Segre ruling with that matching for the whole pencil.
		
		In the second case of \cref{lem:odd-dih-three-coset-circuit}, the matchings
		differ on four variables and all other binary factors are common.  Close every such common pair, on the original gadget,
		with its supported-xor kernel at phase \(1\).  For a binary ratio
		\(r\in-\mu_N\), this gives a nonzero multiple of \(1-r\), nonzero because
		\(N\) is odd.  These operations multiply the three contractions by the same
		nonzero scalar.  At least one pair is closed since \(n\ge8\).  Before
		closing the original pencil variables, the resulting arity-six tensor is a
		nonzero result of contractions of \(f\), hence is itself a binary matching
		product.  Its contractions cannot have three different matchings: closing variables
		of a fixed matching deletes an edge or merges two edges, independently of
		the phase.  This contradiction proves that the matching is independent of phase.
	\end{proof}
	
	\subsubsection{A common binary factor suffices}
	\label{subsec:odd-dih-cross-sector}
	
	\begin{lemma}[A binary factor from contractions]
		\label{lem:odd-dih-common-binary-factor}
		Under the hypotheses of \cref{prop:odd-dih-support-circuit-rigidity}, every
		nonzero \(f\) is reducible.
	\end{lemma}
	
	\begin{proof}
		Put \(F=\widetilde f\), fix variables \(i,j\), and for each
		\(\xi\in\{0,1\}\) consider the restriction to \(x_i\oplus x_j=\xi\), viewed as an
		\(n\)-ary tensor that is zero outside that constraint.
		Its two slices are \(A=F_{ij}^{0,\xi}\), \(B=F_{ij}^{1,1-\xi}\).
		Exactly one of \(A,B\) cannot be nonzero.  Indeed, if \(A\ne0\) and \(B=0\), take a
		supported string of \(A\), close any two dangling edges with its supported
		xor, and choose a phase avoiding cancellation of the projected string.  The
		resulting nonzero contraction has complement-closed support.  Its
		complementary string would require an assignment in \(B\), a contradiction.
		The case \(A=0\), \(B\ne0\) is identical.
		
		For each nonzero restriction, \cref{prop:odd-dih-support-circuit-rigidity} writes
		the pencil as \(H\otimes h_s\), where \(H\) is a fixed product of external
		binaries and
		\[
		h_s=(a-sb)\delta_{\mathbf e}+(c-sd)\delta_{\bar {\mathbf e}}
		\]
		on one external pair, with \({\mathbf e}\in\Ftwo^2\).  Projectively constant pencils are included by
		choosing any pair.  Every nonzero \(h_s\) has ratio in \(-\mu_N\), so
		\cref{lem:odd-dih-RN} applies.  Restoring the variables \(i,j\), its
		proportional alternative gives a binary block \(\{i,j\}\). Either
		monomial alternative gives one quaternary CP block containing \(i,j\).
		All other blocks are nonsingular binaries.  Both slices \(A,B\) are
        nonzero, so the restored block has both coefficients nonzero.
        This factors the restriction in Walsh coordinates.
		
		If only one of the two restrictions is nonzero, this factorization of \(F\) already has an
		external binary factor, since \(n\ge8\).  Otherwise denote the restrictions for \(\xi=0,1\) by \(U,V\).  We use the following observation repeatedly.  Close
		an external pair with an xor supported by both restrictions, and choose a phase
		at which both resulting contractions are nonzero.  Such a phase exists: a contraction of
		either displayed block product with a supported xor has unique lifts, except when an entire
		binary block is closed, in which case only one phase can cancel its
		scalar.  The nonzero contraction \(G\) is a binary matching product with
		both \(ij\) xors present.  Its two \(ij\)-xor restrictions have:
		\begin{enumerate}[label=\textnormal{(O\arabic*)},leftmargin=*]
			\item the same block of four variables containing \(i,j\).
			\item the same external matching and the same offset on each external pair.
			\item proportional binary factors on each external pair, including
			equality of their ordered coefficient ratios.
		\end{enumerate}
		Indeed, \(i,j\) belong to different matching edges in \(G\).  Either xor
		restriction joins those edges into one block of four variables and leaves every
		other binary factor unchanged.  In particular, this observation identifies
		coefficients as well as supports.
		
		First, neither \(U\) nor \(V\) can have \(\{i,j\}\) as a binary block.
		If \(U\) did, close variables in two distinct external binary blocks of
		\(U\).  Both xors occur in \(U\), so one is also supported by \(V\).
		The observation would require a block of four variables through \(i,j\), whereas
		this contraction leaves their binary block unchanged.
		
		Write the blocks of four variables of \(U,V\) as
		\(K=\{i,j,a,b\}\), \(L=\{i,j,c,d\}\).  Close \(a,b\) with its fixed
		xor in \(U\).  That restriction now has \(\{i,j\}\) as a binary block, so the
		observation implies that \(V\) does not support this xor.  Therefore
		\(a,b\) lie in the same block of \(V\), with the opposite xor.  Either
		\(L=K\), or \(\{a,b\}\) is an external binary block of \(V\).  By
		symmetry, in the latter case \(\{c,d\}\) is an external binary block of
		\(U\), and the two pairs are disjoint.  Choose a variable \(e\) outside
		\(\{i,j,a,b,c,d\}\), possible since \(n\ge8\), and close \(a,e\).
		Both restrictions have unique lifts for either xor.  The block of four variables through
		\(i,j\) in the contracted \(U\) contains \(b\), whereas that in the
		contracted \(V\) remains \(\{i,j,c,d\}\).  This contradicts the
		observation.  Thus \(K=L\).
		
		Let \(M_U,M_V\) be the external matchings on \(E=[n]\setminus K\),
		and write \(M_U(e),M_V(e)\) for the respective mates of \(e\).
		For \(e\in E\), closing \(a,e\) makes the block of four variables through
		\(i,j\) equal to \(\{i,j,b,M_U(e)\}\) in one restriction and
		\(\{i,j,b,M_V(e)\}\) in the other.  The observation gives
		\(M_U(e)=M_V(e)\), so the external matchings coincide.  There are at least
		two external pairs.  For any one of them, \(R\), choose \(e\) on another
		pair and close \(a,e\).  The pair \(R\) remains external.  The observation
		now identifies its binary factor in the two restrictions, hence in
		\(U,V\), since the contraction avoids \(R\).  Therefore \(U+V=F\) has
		this common factor.  Invertible Walsh coordinates preserve tensor
		factorization, proving the lemma.
	\end{proof}
	
	\begin{proof}[Proof of \cref{thm:odd-dihedral}]
		Apply \cref{lem:height-one-card-saturation}.  Its first two outcomes are
		the first two conclusions of the theorem.  In outcome~3, work in the
		equivalent set \(\cF^{\mathrm{cl}}\), with the same matrix and
		group.  If a nonbinary irreducible signature remains, choose one, \(f\), of minimum arity
		\(n\).  Odd-arity factors have already been handled by \cref{thm:external-odd}.  Every irreducible factor of a contraction of two variables of \(f\)
		has smaller arity, and hence is binary by minimality.  Its transfer class is in
		\(D_{2N}\), so every nonzero contraction of two variables is in \(\mathcal B_N\) in Walsh
		coordinates.  Closure of \(\mathcal B_N\) gives the same assertion for
		all further contractions involving at least one variable pair.
		
		If \(n=4\), \cref{prop:odd-dih-q4} gives either a quaternary CP signature, to which \cref{lem:walsh-ge-terminal} applies, or a product of
		two binaries, contradicting irreducibility.  If \(n=6\),
		\cref{thm:odd-dih-q6-base} contradicts irreducibility.  If \(n\ge8\),
		\cref{lem:odd-dih-common-binary-factor} gives the same contradiction.
		Thus either hardness or tractability has been proved, or every irreducible
		in \(\cF^{\mathrm{cl}}\) is
		binary.  In the latter case one Walsh basis puts all signatures in \(\cP\),
		and \(P^{\mathsf T}XP=2Z\in\cP\), giving outcome~3 with the stated change of basis.
	\end{proof}
	
	\subsection{The odd-cyclic case}
	
	Only the phase-contraction kernels are available in the cyclic case. We
	first show that a minimum-arity nonbinary irreducible signature must have
	arity four. Binary attachments then turn that quaternary into a signature
	covered by an existing dichotomy or force a factorization. This gives the
	same classification-or-update alternatives as in the dihedral case.

	\begin{theorem}[The odd-cyclic case]
		\label{thm:odd-cyclic}
		Let \((A,\cF)\) be a stage of the sequence in
		\cref{lem:normalization-ledger} at which \(\cF\) is finite, preprocessed,
		algebraic, and factor-closed and contains the transformed height-one
		mountain.  Let \(N\ge3\) be odd, and suppose every nonzero directly
		realizable binary is nonsingular with finite projective transfer order, and
		\[
		[P]^{-1}G(\cF)[P]
		=C_N=\{[\diag(1,s)]:s\in\mu_N\}.
		\]
		Then at least one of the following holds.
		\begin{enumerate}
			\item \(\KHolant(\cF)\) is \(\SharpP\)-hard or \(\Tract(K\cF)\) holds.
            \item There is a strict update \((A,\cF)\mapsto(A^+,\cF^+)\)
            satisfying \cref{def:height-one-outer-successor}.
			\item There is an equivalent finite preprocessed algebraic factor-closed
			set \(\cF^{\mathrm{cl}}\supseteq\cF\), in the same
			coordinates and with the same matrix \(A\), for which the fixed
			Walsh matrix \(P\) satisfies
			\[
			(P^{-1})^{\otimes\arity(f)}f\in\cP
			\quad(f\in\cF^{\mathrm{cl}}),
			\qquad
			P^{\mathsf T}XP=2Z\in\cP.
			\]
			Hence \(K\cF\) is \(\cP\)-transformable.
		\end{enumerate}
	\end{theorem}
	
	\subsubsection{Phase-independent matching}
	\label{subsec:odd-cyc-pencil}
	
	Throughout this odd-cyclic proof, \(\widetilde f\) denotes the fixed
	Walsh coordinates, and
	\([P]^{-1}G(\cF)[P]=C_N=\{[D(s)]:s\in\mu_N\}\), with one direct representative fixed for
	each class and \(N\ge3\) odd.
	
	For the conditional analysis below, work in outcome~3 of
	\cref{lem:height-one-card-saturation}, retaining its finite, preprocessed,
	algebraic, factor-closed set, still denoted by \(\cF\), the matrix \(A\), and the
	transformed height-one mountain.  Odd-arity irreducible signatures have
	already been handled by \cref{thm:external-odd}.  If only binary irreducible factors remain, the common Walsh basis gives
	the third conclusion of the theorem. Otherwise choose a nonbinary
	irreducible \(f\in\cF\) of minimum arity \(n\), which is even.
	For distinct variables \(i,j\), use the phase kernel \(\mathsf K^+(s)\)
    from \cref{eq:odd-dih-kernels}, with \(s\in\mu_N\).
    After division by the nonzero representative scalar, the contraction has Walsh coefficients
	\begin{equation}
		(\partial_{ij}^{\mathsf K^+(s)}\widetilde f)({\mathbf y})=A({\mathbf y})-sB({\mathbf y}),
		\qquad
		A({\mathbf y})=\widetilde f(00,{\mathbf y}),
		\quad B({\mathbf y})=\widetilde f(11,{\mathbf y}).
		\label{eq:odd-cyc-card-pencil}
	\end{equation}
	By \cref{lem:height-one-card-saturation}\textnormal{(3)}, each indicated
	nonzero contraction and all its irreducible factors belong
	to \(\cF\).  A nonbinary factor would have arity at most \(n-2\),
	contradicting minimality.  Every binary factor is nonsingular and its Walsh
	transfer class belongs to \(C_N\).  Thus
	condition~\ref{cond:height-one-contractions} holds for \(f\). Each nonzero
	contraction is a product of such binaries on a perfect matching.  A binary with Walsh transfer class \([D(t)]\) has
	Walsh tensor proportional to
	\[
	D(t)Z=\diag(1,-t).
	\]
Thus the support is the all-zero matching coset on the residual variables:
	for a perfect matching \(M\),
	\begin{equation}
		C(M,\mathbf0)=\{{\mathbf y}\in\Ftwo^{[n]\setminus\{i,j\}}:
		y_a=y_b\text{ for every }\{a,b\}\in M\}.
		\label{eq:odd-cyc-matching-code}
	\end{equation}
	The logical coefficient tensor on these matching variables is nowhere zero
    and has tensor rank one.  By
	\cref{lem:odd-dih-all-arity-parity}, \(\widetilde f\) has even Walsh
	parity.
	
	\begin{lemma}[Phase-independent matching]
		\label{lem:odd-cyc-one-matching}
		Let \(N\ge3\) be odd, and let \(f\in\cF\) be a nonbinary
		irreducible signature of minimum arity under the assumptions of this
		subsection, including condition~\ref{cond:height-one-contractions} with
		\([P]^{-1}G(\cF)[P]=C_N\).  Use the contractions and matching cosets in \cref{eq:odd-cyc-card-pencil,eq:odd-cyc-matching-code}.  For each pair
		\(i,j\), either every contraction vanishes, or \(\partial_{ij}^{\mathsf K^+(s)}\widetilde f=0\) for at most one
		phase \(s\in\mu_N\) and there is a fixed perfect matching \(M_{ij}\)
		such that
		\begin{equation*}
			\partial_{ij}^{\mathsf K^+(s)}\widetilde f\ne0
			\quad\Longrightarrow\quad
			\supp(\partial_{ij}^{\mathsf K^+(s)}\widetilde f)=C(M_{ij},\mathbf0)
			\qquad(s\in\mu_N).
		\end{equation*}
	\end{lemma}
	
	\begin{proof}
		If \(A-sB\) and \(A-tB\) vanish at distinct phases, then \(A=B=0\),
		so the first alternative holds.  Otherwise at most one phase gives a zero
		contraction.  Any two distinct nonzero contractions satisfy
		\begin{equation}
			\supp(A)\cup\supp(B)
			=\supp(A-sB)\cup\supp(A-tB),
			\label{eq:odd-cyc-two-card-union}
		\end{equation}
		because a scalar affine polynomial cannot vanish at two phases unless both
		coefficients vanish.
		
		If no contraction vanishes, choose three phases.  By
		\cref{eq:odd-cyc-two-card-union}, the third matching coset lies in the union
		of the first two.  A vector space contained in \(U\cup V\) lies in one of
		them (otherwise sum vectors from its two exclusive parts), and all matching
		cosets have the same dimension.  Thus the third coset equals one of the first
		two.  Two phases with one coset and a third with another would make a string in
		the latter difference vanish at two phases, contradicting the affine-pencil
		identity.  Hence all cosets coincide.
		
		If \(A-s_0B=0\), then every other contraction is \((s_0-s)B\), and the conclusion
		is immediate.
	\end{proof}
	
	\subsubsection{Pairs with every phase contraction zero}
	\label{subsec:odd-cyc-zero-graph}
	
	The preceding lemma leaves open the possibility that every phase
	contraction on a pair is zero. We encode these pairs by a graph and show
	that none occur above arity four. Define
	\begin{equation*}
		Z_f^{\mathrm{cyc}}
		=
		\{\{i,j\}:\partial_{ij}^{\mathsf K^+(s)}\widetilde f=0\text{ for every }s\in\mu_N\}.
	\end{equation*}
	Two phases in \cref{eq:odd-cyc-card-pencil} show that an edge is exactly
	the support equation
	\begin{equation*}
		\widetilde f({\mathbf x})\ne0
		\quad\Longrightarrow\quad
		x_i+x_j=1.
	\end{equation*}
	It follows that every connected component of
	\(Z_f^{\mathrm{cyc}}\) is complete bipartite.  An odd path forces unequal
	endpoint bits and hence supplies the corresponding edge. An even path
	forces equal bits and cannot also be a zero edge.
	
	\begin{lemma}[Every variable pair has a nonzero phase contraction above arity four]
		\label{lem:odd-cyc-zero-graph-empty}
		Under the assumptions of \cref{lem:odd-cyc-one-matching}, if
		\(\arity(f)=n\ge6\), then
		\[
		Z_f^{\mathrm{cyc}}=\varnothing.
		\]
	\end{lemma}
	
	\begin{proof}
        Each color class has size at most two: if one contained \(p,q,r\)
        and the other \(b\), some phase contraction on \(p,q\) would be
        nonzero. Its matching support contains the all-zero residual string,
        contradicting the forced disequality \(x_r\ne x_b\).
		
		An isolated \(\{1\}\sqcup\{2\}\) is impossible.  For every \(k\ge3\),
		the all-zero string of a nonzero \((1,k)\)-contraction has the unique lift
		\(x_1=x_k=1\), whence
		\begin{equation*}
			\widetilde f({\mathbf e}_1+{\mathbf e}_k)\ne0
			\qquad(k=3,\ldots,n).
		\end{equation*}
		A nonzero \((2,3)\)-contraction then has each \({\mathbf e}_1+{\mathbf e}_k\), \(k\ge4\), by the
		unique lift \(x_2=x_3=0\).  Its support matching would contain every
		\(\{1,k\}\), impossible.
		
		Thus only \(2+1\) and \(2+2\) remain.  Label the size-two color
		\(\{1,3\}\) and the other \(\{2\}\) or \(\{2,4\}\).  In a \((1,3)\)
		contraction the forced equations give the unique lift
		\[
		x_1=x_3=1-x_2.
		\]
		Put
		\[
		h(x_2,\mathbf y)=\widetilde f(1-x_2,x_2,1-x_2,\mathbf y),
		\]
		where \(\mathbf y\) denotes the remaining coordinates. In the \(2+2\)
		case, also impose the support equation \(x_4=x_2\).
		Then
		\begin{equation*}
			(\partial_{13}^{\mathsf K^+(s)}\widetilde f)(x_2,\mathbf y)=u_s(x_2)h(x_2,\mathbf y),
			\qquad
			u_s(0)=-s,\quad u_s(1)=1.
		\end{equation*}
		Every contraction is nonzero and equals a nowhere-zero unary \(u_s\) times \(h\).
		Dividing this contraction by the nowhere-zero unary \(u_s\) only modifies the matching factor incident
		with variable \(2\), so \(h\) factors over that matching. In the \(2+1\) case, let \(k\)
        be the mate of variable \(2\) in this matching. Reversing the unique
		lift gives a quaternary CP factor on \(\{1,2,3,4\}\) in the \(2+2\) case,
		or on \(\{1,2,3,k\}\) in the \(2+1\) case, tensor the remaining binaries.
		For \(n\ge6\) this contradicts irreducibility.
	\end{proof}
	
	\subsubsection{Weight-two coefficients and the arity bound}
	\label{subsec:odd-cyc-weight-two}
	
	We next use the weight-two coefficients to reconstruct a matching of
	the original variables. Once that matching is fixed, contractions impose
	rank-one restrictions on its logical coefficient tensor. The following lemma
	recovers the tensor from those restrictions.

	\begin{lemma}[Rank-one diagonal restrictions]
		\label{lem:diagonal-slice-integration}
		Let \(r\ge3\), and let \(h:\Ftwo^r\to\mathbb C^\times\) satisfy
		\(h(\zero^r)=1\).  Suppose that, for every \(i\ne j\), the restriction
		\(h|_{t_i=t_j}\), viewed as a function of the common bit and the other
		\(r-2\) bits, has tensor rank one.  If \(r\ge4\), then
		\[
		h({\mathbf t})=\prod_{i=1}^r\rho_i^{t_i}
		\qquad(\rho_i\ne0).
		\]
		If \(r=3\), there are \(\rho_1,\rho_2,\rho_3,\kappa\ne0\) such that
		\begin{equation}
			h({\mathbf e}_i)=\rho_i,\qquad
			h({\mathbf e}_i+{\mathbf e}_j)=\kappa\rho_i\rho_j,\qquad
			h(111)=\kappa\rho_1\rho_2\rho_3.
			\label{eq:diagonal-slice-q3-table}
		\end{equation}
	\end{lemma}
	
	\begin{proof}
		This is \cite[Lemma~7.8]{Liu2026DEQ2Holant}.
	\end{proof}
	
	\begin{theorem}[The minimum nonbinary arity is four in the cyclic case]
		\label{thm:odd-cyc-higher-collapse}
		Let \(\cF\) be a finite preprocessed algebraic factor-closed signature
		set containing the transformed height-one mountain.  Suppose every nonzero
		directly realizable binary is nonsingular and has finite projective transfer
		order, and, in the fixed
		Walsh coordinates,
		\[
		[P]^{-1}G(\cF)[P]
		=C_N=\{[\diag(1,s)]:s\in\mu_N\},
		\qquad N\ge3\text{ odd}.
		\]
		Let \(f\in\cF\) be an even-arity nonbinary irreducible signature of
		minimum arity among the nonbinary irreducible signatures in \(\cF\).
		For every ordered pair of distinct variables of \(f\), every fixed direct
		representative of a class in \(G(\cF)\), and either kernel orientation,
		assume that the resulting contraction of two variables is zero or all of its irreducible factors are nonsingular binaries \(B\) satisfying
		\[
		[P^{-1}(BX)P]\in C_N.
		\]
		Then \(\arity(f)\le4\), and hence \(\arity(f)=4\).
	\end{theorem}
	
	\begin{proof}
		The contraction hypothesis is condition~\ref{cond:height-one-contractions}
		for the stated group, so
		\cref{lem:odd-cyc-one-matching,lem:odd-cyc-zero-graph-empty} apply.
		Assume that \(f\) has arity \(n\ge6\).  By
		\cref{lem:odd-cyc-zero-graph-empty}, every pair has a nonzero phase contraction.
		For \(S\subseteq[n]\), write \(a_S=\widetilde f(\one_S)\).
		We abbreviate \(a_{i_1\cdots i_k}=a_{\{i_1,\ldots,i_k\}}\).
		
		\emph{Step 1: the constant coefficient is nonzero.}
		If \(a_\varnothing=0\), the all-zero coefficient
		\((\partial_{ij}^{\mathsf K^+(s)}\widetilde f)(0)=-sa_{ij}\) of a nonzero matching contraction forces
		\(a_{ij}\ne0\) for every pair.  Hence all contractions are nonzero.  For disjoint
		pairs \(ij,pq\),
		\[
		(\partial_{ij}^{\mathsf K^+(s)}\widetilde f)({\mathbf e}_p+{\mathbf e}_q)=a_{pq}-sa_{ijpq}
		\]
		is nonzero at some phase.  The common matching from
		\cref{lem:odd-cyc-one-matching} would therefore contain every edge on its
		at least four residual variables, impossible.  Normalize henceforth
		\(a_\varnothing=1\), and put
		\[
		W_f=\{\{p,q\}:a_{pq}\ne0\}.
		\]
		
		\emph{Step 2: a nonempty \(W_f\) factors \(f\).}
		For deleted \(i,j\) and a disjoint edge \(pq\in W_f\), the coefficients
		\(1-sa_{ij}\) and \(a_{pq}-sa_{ijpq}\) exclude at most one phase each.
		Thus some contraction contains \({\mathbf e}_p+{\mathbf e}_q\), and its common support matching
		contains \(pq\).  Deleting two variables outside any two adjacent
		\(W_f\)-edges would put both edges in one matching, so \(W_f\) is a
		matching.
		
		Moreover, every \(pq\in W_f\) is a global equality constraint.  Indeed, if
		some supported \({\mathbf z}\) had \(z_p\ne z_q\), choose equal coordinates \(i,j\)
		outside \(p,q\), possible by even parity and \(n\ge6\).  The residual string
		violates the \(pq\)-edge in every nonzero \((i,j)\)-contraction and therefore
		vanishes at all phases.  Two phases of
		\cref{eq:odd-cyc-card-pencil} then make both lifts vanish, including \({\mathbf z}\).
		Consequently
		\begin{equation}
			pq\in W_f\quad\Longrightarrow\quad
			\supp(\widetilde f)\subseteq\{{\mathbf z}:z_p=z_q\}.
			\label{eq:odd-cyc-global-equality}
		\end{equation}
		
		Fix \(E_1=\{u,v\}\in W_f\) and \(R=[n]\setminus E_1\).  In an
		\((u,k)\)-contraction, the all-zero string has the unique lift
		\(x_u=x_k=x_v=0\). Hence, for \(p,q\in R\setminus\{k\}\), its coefficient
		at \({\mathbf e}_p+{\mathbf e}_q\) is exactly \(a_{pq}\).  Its matching therefore consists of
		the \(W_f\)-edges on \(R\setminus\{k\}\) and one edge incident with \(v\).
		Since that last edge cannot cover all \(n-2\ge4\) variables, the induced
		matching on \(R\) is nonempty.  Choose \(k\) on one of its edges, with
		partner \(\ell\).  In the \((u,k)\)-contraction, \(\ell\) has no incident
		\(W_f\)-edge left, so it must be the partner of \(v\).  Any vertex
		unmatched in \(W_f|_R\) would also have to be matched to \(v\), which is
		impossible.  Thus \(W_f=\{E_1,\ldots,E_r\}\) is perfect, where \(r=n/2\).
		
		Equation \eqref{eq:odd-cyc-global-equality} identifies \(\widetilde f\)
		with a logical coefficient tensor \(h(t_1,\ldots,t_r)\) on the common bits of the
		\(E_k\)'s.  Contracting one variable from \(E_i\) and one from \(E_j\) imposes
		\(t_i=t_j\), and the remaining partners determine the deleted bits uniquely.
		The contraction is nonzero since its all-zero coefficient is one.  Its support is
		contained in the matching coset obtained by merging \(E_i,E_j\).  Both this
		coset and the contraction's matching support have \(2^{r-1}\) points, so they
		coincide.  Removing the nowhere-zero unary phase on the merged variable
		therefore gives a full-support rank-one diagonal slice of \(h\).  Every assignment in \(\Ftwo^r\) has two equal coordinates, so these
		restrictions cover \(\Ftwo^r\) and \(h\) has full support.  With \(h(\zero^r)=1\),
		\cref{lem:diagonal-slice-integration} factors \(h\) when \(r\ge4\).  If
		\(r=3\), use \cref{eq:diagonal-slice-q3-table}.  Let \(H_i(s)\) be the
		two-bit logical coefficient matrix obtained by closing \(E_i\) at phase \(s\).  For
		\(\{i,j,k\}=[3]\),
		\begin{equation*}
			\det H_i(s)=\rho_j\rho_k(\kappa-1)
			(1-\kappa s^2\rho_i^2).
		\end{equation*}
		If \(\kappa\ne1\), the contraction is nonzero and hence rank one for every phase,
		whereas \(s\mapsto s^2\) permutes the odd-order group \(\mu_N\). The last
		factor cannot vanish at two phases.  Thus \(\kappa=1\), and the arity-six
		tensor also factors.  Both cases contradict irreducibility.
		
		\emph{Step 3: an empty \(W_f\) is impossible.}
		Fix \(i,j\).  Every contraction now has constant coefficient one and is nonzero.
		Choose two edges \(e,e'\) of its common residual matching.  Full matching
		support gives
		\(a_{\{i,j\}\cup e}a_{\{i,j\}\cup e'}\ne0\), while its two-bit rank-one
		identity is
		\begin{equation*}
			a_{e\cup e'}-s a_{\{i,j\}\cup e\cup e'}
			=s^2a_{\{i,j\}\cup e}a_{\{i,j\}\cup e'}
			\qquad(s\in\mu_N).
		\end{equation*}
		This nonzero quadratic cannot vanish at the \(N\ge3\) distinct phases.
		The contradiction completes the proof.
	\end{proof}
	
	\subsubsection{The quaternary cyclic case}
	\label{subsec:odd-cyc-quaternary}
	
    The minimum nonbinary arity is now four, but the resulting quaternary
    need not be a mountain. We first record the eight-vertex classification
    used to finish this branch, then adjust the Walsh coefficients by binary
    attachments.
	
	An eight-vertex signature has arity four and is zero on inputs of odd Hamming weight.
	
	\begin{theorem}[Irreducible quaternary \(K\)-Holant dichotomy]
		\label{thm:quaternary-kholant-dichotomy}
		Let \(\cF\) be a finite algebraic signature set containing an
		irreducible eight-vertex signature \(M\) with \(M(0000)M(1111)\ne0\).
		If \(\Tract(K\cF)\) holds, then \(\KHolant(\cF)\in\FP\). Otherwise
		\(\KHolant(\cF)\) is \(\SharpP\)-hard.  The predicate \(\Tract(K\cF)\) is decidable.
	\end{theorem}
	
	\begin{proof}
		This is \cite[Theorem~1.3]{Liu2026DEQ2Holant}.
	\end{proof}
	
	\begin{corollary}[Eight-vertex \(K\)-Holant dichotomy]
		\label{cor:eight-vertex-complete-closure}
		Let \(\cF\) be a finite algebraic signature set containing an
		eight-vertex signature \(M\) with \(M(0000)M(1111)\ne0\).
		If \(\Tract(K\cF)\) holds, then \(\KHolant(\cF)\in\FP\). Otherwise
		\(\KHolant(\cF)\) is \(\SharpP\)-hard.
	\end{corollary}
	
	\begin{proof}
		If \(\Tract(K\cF)\) holds,
		\cref{thm:predicate-algorithms,eq:K-equivalence} give the polynomial-time
		conclusion.  Assume \(\neg\Tract(K\cF)\).  If \(M\) is irreducible,
		\cref{thm:quaternary-kholant-dichotomy} gives hardness.
		
		Suppose instead that \(M\) is reducible.  Its nonzero endpoint values
		exclude a \(1+3\) factorization: setting the unary block to \(1\) and the
		other block to \(000\) would give a nonzero odd-weight entry.  Thus every
		nontrivial factorization has type \(2+2\).  In such a factorization, the
		nonzero endpoints and the odd-weight zeros force each binary factor to be
		diagonal with both diagonal entries nonzero.  Up to a variable permutation and a
		nonzero scalar, one factor is therefore
		\[
		E=\diag(p,q),\qquad p,q\in\AlgNums^\times.
		\]
		Because \(M\) contains both charge signs, the common-one-sided exception in
		\cref{lem:factor-saturation} is impossible.  Factor extraction gives
		\(\KHolant(\cF,E)\leT\KHolant(\cF)\).
		Choose algebraic \(a,b\ne0\) with \(a^2p=b^2q=1\), and put
		\(D=\diag(a,b)\).  Then
		\[
		D^{\mathsf T}XD=abX,
		\qquad DED^{\mathsf T}=I,
		\]
		so \(D\in GO(X)\) and \(D(\cF\cup\{E\})\) contains \(I\).
		By \cref{lem:downward-inheritance,lem:GOX-tract-invariance}, this transformed
		set still fails \(\Tract\).  Hence \cref{thm:eq2-kholant-dichotomy}
		gives hardness, which the change of basis and the displayed reduction
        transfer to \(\KHolant(\cF)\).
	\end{proof}
	
	The next lemma specifies the output and phase choices of these attachments.

\begin{lemma}[Attaching cyclic binary representatives to four variables]
		\label{lem:odd-cyc-four-port-dressing}
		Let \(N\ge3\), let \(\cF\) be a finite algebraic signature set in
		\(K\)-coordinates, and, for every \(s\in\mu_N\), fix a gadget-realizable ordered binary representative \(C_s\) over \(\cF\), in a specified orientation,
		whose Walsh matrix is
		\[
		\widetilde C_s=\kappa_sD(s)Z,
		\qquad \kappa_s\in\AlgNums^\times.
		\]
		Let \(q\in\cF\) be quaternary.
		For \(r_1,r_2,r_3,r_4\in\mu_N\), attach \(C_{r_i}\) to variable \(i\) of
		\(q\) through one \(X\)-edge and leave the other variable of each
		representative exposed.  After division by
		\(2^4\prod_{i=1}^4\kappa_{r_i}\), the resulting output \(h\) satisfies
		\begin{equation*}
			\widetilde h_{\mathbf y}
			=\widetilde q_{\mathbf y}\prod_{i=1}^4r_i^{y_i}.
		\end{equation*}
		Consider the map
		\begin{equation*}
			(r_1,r_2,r_3,r_4)
			\longmapsto
			(R,r_{12},r_{13},r_{14}),
			\qquad
			R=\prod_{i=1}^4r_i,
			\quad r_I=\prod_{i\in I}r_i,
		\end{equation*}
		Write \((\mu_N)^2:=\{z^2:z\in\mu_N\}\).  A target
		\((R,u,v,w)\in\mu_N^4\) occurs exactly when
		\begin{equation}
			\frac{uvw}{R}\in(\mu_N)^2.
			\label{eq:cyclic-dressing-image}
		\end{equation}
		For odd \(N\), this map is bijective.
	\end{lemma}
	
    \begin{proof}
        In Walsh coordinates an attachment acts by
        \(\widetilde C_s(2Z)=2\kappa_sD(s)\). Thus four attachments give
        the normalization \(2^4\prod_{i=1}^4\kappa_{r_i}\) in the statement.
        The phase-map calculation is \cite[Lemma~7.3]{Liu2026DEQ2Holant}.\looseness=-1
	\end{proof}
	
	\begin{theorem}[The quaternary odd-cyclic case]
		\label{thm:odd-cyc-quaternary-closure}
		Let \(\cF\) be a finite preprocessed algebraic factor-closed signature
		set containing the transformed height-one mountain.  Suppose every nonzero
		directly realizable ordered binary is nonsingular and has finite projective
		transfer order, and, in the fixed Walsh coordinates,
		\[
		[P]^{-1}G(\cF)[P]
		=C_N=\{[\diag(1,s)]:s\in\mu_N\},
		\qquad N\ge3\text{ odd}.
		\]
		Let \(0\ne q\in\cF\) be quaternary.  For every ordered pair of distinct
		variables of \(q\), every fixed direct representative of a class in
		\(G(\cF)\), and either kernel orientation, assume that the resulting
		contraction of two variables is zero or all of its irreducible factors are
		nonsingular binaries \(B\) with
		\[
		[P^{-1}(BX)P]\in C_N.
		\]
		Then at least one of the following holds.
		\begin{enumerate}
			\item \(\KHolant(\cF)\) is \(\SharpP\)-hard or
			\(\Tract(K\cF)\) holds.
			\item Up to a permutation of its variables,
			\[
			q\doteq a\otimes b,
			\]
			where \(a,b\) are nonzero positive-arity signatures of arity strictly
			below four, and
			\[
			\KHolant(\cF,a,b)\leT\KHolant(\cF).
			\]
		\end{enumerate}
	\end{theorem}
	
	\begin{proof}
		The stated contractions satisfy
		condition~\ref{cond:height-one-contractions} for the cyclic group.
		In the Walsh tensor \(\widetilde q\), put
		\[
		p=\widetilde q_{0000},\qquad
		t=\widetilde q_{1111},\qquad
		x_I=\widetilde q_{\one_I}\quad(|I|=2).
		\]
		For each such \(I\), put \(B_I=(x_I,x_{I^c})\), so
		\(B_{I^c}\) has the same two coefficients in reverse order.
		By \cref{lem:odd-dih-all-arity-parity}, \(\widetilde q\) has even-parity support. \Cref{lem:odd-cyc-q4-family} excludes exactly one nonzero
		Walsh endpoint.  Suppose first that \(p,t\ne0\), and put \(\lambda=t/p\).
		For each \(B_I\ne(0,0)\), it gives
		\((x_I,x_{I^c})=(-p\lambda\alpha_I^{-1},-p\alpha_I)\).
		For each zero pair choose \(\alpha_I\in\mu_N\) arbitrarily, using one
		index from each complementary pair as in \cref{lem:odd-cyc-q4-family}.
		By \cref{lem:odd-cyc-four-port-dressing}, prescribe
		\(R=\lambda^{-1}\) and \(r_I=\alpha_I\lambda^{-1}\).  Let \(h\) be the
		signature obtained with the normalization in that lemma.  Its Walsh coefficient array
		\(\widetilde h\) has endpoints
		\((p,p)\), value \((-p,-p)\) on each complementary pair with \(B_I\ne(0,0)\),
		and zero on every pair with \(B_I=(0,0)\).
		
		Here the eight even-parity Walsh positions form the endpoint pair and
		three central pairs \(\{\one_I,\one_{I^c}\}\), where \(|I|=2\).
		In the inverse Walsh formula of \cref{eq:Walsh-coordinates}, a pair
		\(\{{\mathbf y},\bar {\mathbf y}\}\) with common coefficient \(b\) contributes
		\[
		b\bigl((-1)^{{\mathbf x}\cdot {\mathbf y}}+(-1)^{{\mathbf x}\cdot\bar {\mathbf y}}\bigr).
		\]
		This vanishes when \(\wt({\mathbf x})\) is odd.  At either endpoint all signs are
		positive.  At \({\mathbf x}=\one_J\), \(|J|=2\), the endpoint pair and the
		pair \(\{\one_J,\one_{J^c}\}\) have positive signs, and the other two
		central pairs have negative signs.  Thus, if exactly \(k\) of the three pairs \(B_I\) are
		nonzero,
		\[
		\begin{aligned}
			h(\zero^4)=h(\one^4)&=2p(1-k),\\
			h(\one_J)&=
			\begin{cases}
				2p(k-1),&\text{if \(B_J\ne(0,0)\)},\\
				2p(k+1),&\text{if \(B_J=(0,0)\)}.
			\end{cases}
		\end{aligned}
		\]
		Walsh inversion gives the K-coordinate entries of the normalized gadget
		output. Thus \(k=1\) gives a
		nonzero multiple of \(X\otimes X\), up to a variable permutation, whose
		factorization pulls back through the invertible binary attachments just
		constructed.  For
		\(k=0\), the original Walsh tensor is already a higher-arity CP signature, so
		\cref{lem:walsh-ge-terminal} gives outcome~1.  For \(k=2,3\), the output in
		K-coordinates has even-parity support and two nonzero endpoints.  Add this quaternary and apply \cref{cor:eight-vertex-complete-closure}.  Its
		conclusion transfers to \(\cF\) by the gadget
		reduction and preservation under taking subsets, again giving outcome~1.
		
		It remains that \(p=t=0\).  If just one complementary pair has nonzero coefficients,
		\cref{lem:walsh-ge-terminal} gives outcome~1 directly.  Otherwise label the three complementary pairs
		\(\{\one_{I_\ell},\one_{I_\ell^c}\}\), \(\ell=1,2,3\).
		Write each nonzero pair \(B_I\) as
		\(\rho_I(1,\alpha_I)\), choose the unique \(\beta_I\in\mu_N\) with
		\(\beta_I^2=\alpha_I\). For a zero pair choose \(\beta_I\in\mu_N\)
		arbitrarily.  Apply \cref{lem:odd-cyc-four-port-dressing} with
		\(R=1,r_I=\beta_I\).  Write \(h\) for the resulting signature, rescaled according to
		\cref{lem:odd-cyc-four-port-dressing}.  Its nonzero Walsh coefficient pairs
		become \((\sigma_I,\sigma_I)\), where
		\(\sigma_I=\rho_I\beta_I\). Set \(\sigma_I=0\) when \(B_I=(0,0)\), and write
		\(\sigma_\ell=\sigma_{I_\ell}\).  These attachments again give an invertible operation on each variable.  Put \(c=\sigma_1+\sigma_2+\sigma_3\).  The same sign
		calculation, now with zero Walsh endpoints, gives
		\[
		\begin{aligned}
			h(\zero^4)=h(\one^4)&=2c,\\
			h(\one_{I_\ell})=h(\one_{I_\ell^c})
			&=2(2\sigma_\ell-c)\qquad(\ell=1,2,3),
		\end{aligned}
		\]
		and all odd-weight entries vanish.  If \(c\ne0\), this is a
		complement-invariant mountain. After division by \(2c\), its endpoints
		are one and its three central coefficient-pair values are
		\[
		r_\ell=2\sigma_\ell/c-1,
		\qquad 1+r_1+r_2+r_3=1+2c/c-3=0.
		\]
		Up to the nonzero scalar, this mountain is the output of invertible binary
		attachments to \(q\), so it may be
		added by \cref{eq:finite-equivalent-enlargement}.
		\Cref{cor:eight-vertex-complete-closure} makes the enlarged problem either
		tractable or \(\SharpP\)-hard.  Preservation under taking subsets in the first case and
		the gadget reduction in the second transfer that conclusion to
		\(\cF\), giving outcome~1.
		
		Assume \(c=0\).  Exactly one nonzero pair is impossible.  Two nonzero Walsh pairs have
		values \(\sigma,-\sigma\). After permuting variables, \(\widetilde h\) factors as
		\[
		\sigma
		\begin{pmatrix}0&1\\-1&0\end{pmatrix}_{14}\otimes
		\begin{pmatrix}0&1\\-1&0\end{pmatrix}_{23},
		\]
		so the factor-reduction equivalence applies.  Finally suppose that all three \(\sigma_i\)
		are nonzero.  Choose distinct \(i,j,k\) with
		\(\{i,j,k\}=\{1,2,3\}\) and \(\sigma_j^2+\sigma_k^2\ne0\), which exists since otherwise subtraction
		would give \(2\sigma_i^2=0\).  Put pair \(i\) at \(12\mid34\), and join
		variables \(3,4\) of one copy to variables \(1,2\) of another by edges.  The
		output is
		\begin{equation*}
			\frac14\widetilde G_i=
			\begin{pmatrix}
				\sigma_i^2&0&0&0\\
				0&-(\sigma_j^2+\sigma_k^2)&-2\sigma_j\sigma_k&0\\
				0&-2\sigma_j\sigma_k&-(\sigma_j^2+\sigma_k^2)&0\\
				0&0&0&\sigma_i^2
			\end{pmatrix}.
		\end{equation*}
		If every phase contraction of \(G_i\) were zero or had transfer class in
		\(C_N\) in Walsh coordinates, its equal nonzero endpoints and
		\cref{lem:odd-cyc-q4-family} would give \(\lambda=1\). Complement symmetry
		and odd \(N\) then force every nonzero central coefficient to be
		\(-\sigma_i^2\).  In particular,
		\[
		\sigma_j^2+\sigma_k^2=\sigma_i^2=(\sigma_j+\sigma_k)^2,
		\]
		because \(\sigma_i=-(\sigma_j+\sigma_k)\), contradicting
		\(\sigma_j\sigma_k\ne0\).  Hence a phase contraction of \(G_i\) is a
		nonzero binary which is singular or has Walsh transfer class outside
		\(C_N\).  A singular binary gives unary factors and the odd-arity theorem.
		A nonsingular binary of infinite projective transfer order is handled by
		\cref{lem:binary-interpolation}.  Otherwise its transfer has finite
		projective order, contradicting the definition of \(G(\cF)\).
		Every hardness or tractability conclusion above gives outcome~1 through the
		cited theorem and the equivalent-enlargement reduction.  In a
		factor case, the attached binaries act invertibly, so the displayed
		factorization pulls back to \(q\). The transformed mountain excludes the
		common-one-sided exception, and the factor-reduction equivalence gives the
		reduction in outcome~2.
	\end{proof}
	
	\begin{proof}[Completion of \cref{thm:odd-cyclic}]
		Apply \cref{lem:height-one-card-saturation}.  Its first two outcomes are the
		first two conclusions of the theorem.  In outcome~3, the listed two-variable
		contractions and all their irreducible factors are retained
		in the equivalent set supplied by outcome~3.
		Other fixed gadget outputs and any required factors are added separately
		by \cref{eq:finite-equivalent-enlargement,thm:external-factor}, as in the
		quaternary closure proof.  If a nonbinary irreducible signature remains, choose one of minimum arity.  The argument preceding
		\cref{lem:odd-cyc-one-matching} proves
		condition~\ref{cond:height-one-contractions} for this signature.
		\Cref{thm:odd-cyc-higher-collapse} bounds its arity by four, and
		\cref{thm:odd-cyc-quaternary-closure} then gives hardness, tractability, or
		a tensor factorization.  The last possibility contradicts the choice of an
		irreducible signature, so outcome~1 holds.
		
		In the remaining case every irreducible signature is binary.
		By \cref{lem:height-one-card-saturation}\textnormal{(3)}, each has transfer
		class in \(G(\cF)\) and Walsh tensor proportional to \(D(r)Z\).
		Hence \cref{eq:common-Walsh-product-presentation} holds,
		which proves the third conclusion of \cref{thm:odd-cyclic}.
	\end{proof}

\subsection{Completion of the height-one case}

	We now combine the binary reduction and the two group cases into one
	step of the iteration. Its alternatives distinguish a completed
	classification from a strict update to which the same argument applies.

	\begin{theorem}[Hardness, tractability, or a strict update]
		\label{thm:height-one-transfer-interface}
		Let \(\cF_0\) be the initial signature set, and let
		\((A,\cF)\) be a stage of the sequence in
		\cref{lem:normalization-ledger}. In particular,
		\[
		A\cF_0\subseteq\cF,
		\qquad A^{\mathsf T}XA=\mu_AX\quad(\mu_A\ne0),
		\qquad \KHolant(\cF)\equivT\KHolant(\cF_0).
		\]
		Suppose \(\cF\) is finite, preprocessed, algebraic, and factor-closed,
		and contains the height-one mountain
		\[
		B=\begin{pmatrix}c_0&u\\v&c_1\end{pmatrix},
		\qquad c_0c_1\ne0.
		\]
		Then one of the following occurs:
		\begin{enumerate}
			\item \(\KHolant(\cF)\) is \(\SharpP\)-hard, or
			\(\Tract(K\cF)\) holds.
			\item there are \(A^+\in GO(X)\) and a finite preprocessed algebraic
			factor-closed set \(\cF^+\), continuing the same
			sequence of additions and changes of basis, such that
			\[
			A^+\cF_0\subseteq\cF^+,
			\qquad
			(A^+)^{\mathsf T}XA^+=\mu_{A^+}X\quad(\mu_{A^+}\ne0),
			\qquad
			\KHolant(\cF^+)\equivT\KHolant(\cF),
			\]
			and, for \(D=A^+A^{-1}\), the transformed binary
			\(DBD^{\mathsf T}\) remains in \(\cF^+\) up to its known nonzero
			scalar.  In this alternative every nonzero directly realizable ordered
			binary over both \(\cF\) and \(\cF^+\) is nonsingular with
			finite projective transfer order, so
			\(G(\cF)\) and \(G(\cF^+)\) are defined.  Moreover, either
			\[
			[A]^{-1}G(\cF)[A]
			\subsetneq
			[A^+]^{-1}G(\cF^+)[A^+],
			\]
			or both values of \(\nu\) are defined and
			\[
			\nu(\cF^+)<\nu(\cF).
			\]
			Both quantities satisfy the non-strict comparisons and fixed-field
			conditions of \cref{def:height-one-outer-successor}.
		\end{enumerate}
	\end{theorem}
	
	\begin{proof}
		Maintain factor closure using the factor-reduction equivalence. The
		retained mountain excludes the common-one-sided exception. Each
		nontrivial factorization produces positive-arity factors of strictly smaller
		arity, so this process terminates.  An odd-arity member is handled by
        \cref{thm:external-odd,eq:K-equivalence}. A rank-one binary supplies unaries and hence the same
		conclusion.  If no nonbinary irreducible signature remains, the
		arity-at-most-two alternative gives tractability.
		
		All additions of signatures and global changes of basis satisfy
		\cref{lem:normalization-ledger}.  Whenever the groups are defined, conjugating them
		back to the original basis gives nested groups.  The old irreducible signatures, up to one common
		invertible transformation, remain in the set, so \(\nu\) cannot increase.  A strict
		group inclusion or decrease of \(\nu\) gives outcome~2.  If adding signatures preserves both quantities, the matrix and group
		remain unchanged.
		\Cref{lem:height-one-card-saturation} proves separately that the required
		recursive contraction and factor extraction is finite.
		
		If \(B\) is singular, \cref{lem:binary-interpolation} supplies its unary
		factors. If \(BX\) is nonsingular of infinite projective
		order, it first supplies a rank-one binary and then unaries.  The same argument applies to every other directly realizable binary that
		is singular or whose nonsingular transfer has infinite projective order.  We may therefore assume
		that every
		nonzero directly realizable binary is nonsingular and of finite projective transfer order.
		Then \cref{lem:bounded-torsion,thm:finite-group-split} give either
		a directly realizable nonsingular weighted equality or, after one
		\(GO(X)\)-normalization, an odd cyclic or odd dihedral residual form.
		In the equality case, write \(E=\diag(p,q)\), where \(pq\ne0\).
		The gadget gives \(\KHolant(\cF,E)\leT\KHolant(\cF)\).
		Choose a diagonal \(D\in GO(X)\) with \(DED^{\mathsf T}\doteq I\).
		The set \(D(\cF\cup\{E\})\) satisfies the hypotheses of
		\cref{thm:eq2-kholant-dichotomy}.  Its hard conclusion composes back to
		\(\cF\), while its tractable conclusion pulls back through the problem
		equivalence and \cref{lem:normalization-ledger}. Hence this is outcome~1.
		
		Apply \cref{thm:odd-cyclic,thm:odd-dihedral} in the two residual forms.
		Their first conclusions give outcome~1, and their updates satisfy
		outcome~2 by \cref{def:height-one-outer-successor}.  Their remaining
		conclusions give \cref{eq:common-Walsh-product-presentation}, which
		\cref{lem:normalization-ledger} implies \(\Tract(K\cF)\).
		This again gives outcome~1 and completes the case analysis.
	\end{proof}
	
	\begin{proof}[Proof of \cref{cor:height-one-complete}]
		Maintain \((A,\cF)\) and the invariant of
		\cref{lem:normalization-ledger} over the fixed initial number field, and
		apply \cref{thm:height-one-transfer-interface}.  Its first outcome is the
		required hardness or tractability conclusion.  The second outcome gives
		an update satisfying \cref{def:height-one-outer-successor}. By
		\cref{lem:height-one-progress-well-founded}, only finitely many such
		updates are possible.  Between these updates, factor closure terminates by
		strict arity descent, and the recursive contractions and factor extraction
		terminate by \cref{lem:height-one-card-saturation}.  Repeated application
		therefore reaches the first outcome.  Finally,
		\cref{lem:normalization-ledger} gives the conclusion for the
		initial signature set.
	\end{proof}
	
	\section{Mountain Height Two}
	\label{sec:height-two}

	\begin{theorem}[Dichotomy with a height-two mountain]
		\label{cor:height-two-complete}
		Let \(\cF\) be a finite algebraic signature set containing a height-two
		mountain.  Then
		\[
		\Tract(K\cF)
		\qquad\text{or}\qquad
		\KHolant(\cF)\text{ is }\SharpP\text{-hard}.
		\]
	\end{theorem}
	
	\begin{proof}
		A height-two mountain has arity four, is supported on Hamming weights
		\(0,2,4\), and has two nonzero endpoint values. It is therefore an
		eight-vertex signature satisfying the hypotheses of
		\cref{cor:eight-vertex-complete-closure}, which gives the conclusion.
	\end{proof}
	
	\Needspace{10\baselineskip}
	\section{Mountain Height Three}
	\label{sec:height-at-least-three}
	
	\begin{theorem}[Height-three mountain dichotomy]
		\label{thm:height-three-complete}
		Let \(\cF\) be a finite algebraic set of even-arity signatures such that
		\(\chi_r({\mathbf x})\in3\mathbb Z\) for every arity-\(r\) member \(f\in\cF\)
		and every \({\mathbf x}\in\supp(f)\).  Suppose that \(\cF\) contains a mountain
		\[
		M=a\delta_{\zero^6}+b\delta_{\one^6}+Q,
		\qquad ab\ne0,\qquad Q\text{ EO}.
		\]
		Then \(\Tract(K\cF)\) holds or \(\KHolant(\cF)\) is
		\(\SharpP\)-hard.
	\end{theorem}
	
	Write \(Q=M|_{\EO}\) for the central restriction. We construct a nonzero
	quaternary EO signature \(h\notin\mathcal B_\times\), where
	\(\mathcal B_\times\) is the class of products of nonsingular
	anti-diagonal binary signatures defined below, and obtain
	\(\KHolant(\cF,h)\leT\KHolant(\cF)\).
	Since \(\mathcal B_{\rm ARS}\subseteq\mathcal B_\times\),
	\cref{thm:hc-relative-eo-terminal} then gives hardness in the remaining
	case \(\neg\Tract(K\cF)\).
	If \(Q\in\mathcal B_\times\), a two-copy construction supplies
	\(h\). Otherwise, we use self-loops and two-copy contractions, together
	with the computer-assisted exhaustive classification in Appendix~\ref{proofsubsec:cert-high-charge-q6}.
	
	\subsection{Matching Products and Two-Copy Constructions}
	
	Let \(V\) be an ordered set of even cardinality.  For distinct \(i,j\in V\),
	write \(\partial_{ij}f\) for the contraction with an \(X\)-edge.
	In the following formula, the first two displayed arguments assign values
	to variables \(i,j\), respectively, and all other variables retain their order:
	\[
	(\partial_{ij}f)({\mathbf x}_{V\setminus\{i,j\}})
	=f(0,1,{\mathbf x}_{V\setminus\{i,j\}})
	+f(1,0,{\mathbf x}_{V\setminus\{i,j\}}).
	\]
	Define the matching-product class
	\[
	\mathcal B_\times(V)=
	\left\{\lambda\bigotimes_{\{u,v\}\in P}
	\begin{pmatrix}0&p_{uv}\\q_{uv}&0\end{pmatrix}_{u,v}:
	\begin{array}{l}
		P\text{ is a perfect matching of }V,\\
		\lambda\prod_{\{u,v\}\in P}p_{uv}q_{uv}\ne0
	\end{array}\right\},
	\]
	and put \(\mathcal B_\times^0=\mathcal B_\times\cup\{0\}\).
	Each factor has a specified orientation. Swapping its two variables exchanges its
	two nonzero entries.  In particular,
	\(\mathcal B_{\rm ARS}\subseteq\mathcal B_\times\).

	\begin{lemma}[A matching-product central restriction yields a quaternary signature]
		\label{lem:hc-transversal-cube-terminal}
		Let \(d\ge3\), and let \(\cF\) be a finite algebraic signature set
		containing
		\[
		M=a\delta_{\zero^{2d}}+b\delta_{\one^{2d}}+Q,
		\qquad ab\ne0,
		\]
		where, for a perfect matching
		\(P=\{\{u_1,v_1\},\ldots,\{u_d,v_d\}\}\) of the ordered variables,
		\[
		Q=\lambda\bigotimes_{k=1}^{d}
		\begin{pmatrix}0&p_k\\q_k&0\end{pmatrix}_{u_k,v_k},
		\qquad \lambda\prod_{k=1}^{d}p_kq_k\ne0.
		\]
		There is a fixed two-copy gadget over \(\{M\}\), using only \(X\)-edges, whose ordered boundary signature \(h\) is EO and has the six
		weight-two entries
		\begin{equation}
		\bigl(h_{0011},h_{0101},h_{0110},h_{1001},h_{1010},h_{1100}\bigr)
		=(ab,c_1,c_2,c_3,c_4,ab),
		\qquad c_1c_2c_3c_4\ne0.
		\label{eq:auto-wt-18a}
		\end{equation}
		In particular,
		\[
		\KHolant(\cF,h)\leT\KHolant(\cF).
		\]
		The signature \(h\) is not a tensor product of two nonsingular
		anti-diagonal binaries.  Moreover, its supported triples include
		\[
		1100\oplus1010\oplus0110=0000,
		\qquad
		0011\oplus0101\oplus1001=1111.
		\]
	\end{lemma}
	
	\begin{proof}
		We connect two copies so that their central matchings and the connecting
		edges form two odd-length alternating components. These supply four
		boundary patterns with one bit equal to one on each external pair,
		while the two endpoint choices supply the other two patterns.
		Write the matching of \(Q\) as
		\(P=\{\{w_0,w_1\},\ldots,\{w_{2d-2},w_{2d-1}\}\}\).  On two copies of
		the mountain, define a variable bijection \(p\) by fixing
		\(w_0,w_{2d-1}\) and transposing
		\[
		(w_1,w_2),(w_3,w_4),\ldots,(w_{2d-3},w_{2d-2}).
		\]
		Leave \((w_0,w_1)\) external on the first copy and \((w_0,w_2)\)
		external on the second, in the external order
		\((w_0^{(1)},w_1^{(1)},w_0^{(2)},w_2^{(2)})\).  For every
		\(r\notin\{0,1\}\), join the
		first-copy variable \(w_r\) to the second-copy variable \(p(w_r)\) by one \(X\)-edge.  With both copies central, the matching edges and the connecting X-edges form
		two alternating components: the edge between the first external pair and
		a path of length $4d-3$ between the second.  Both lengths are odd, so each
		weight-one/weight-one boundary pattern has a unique extension and a
		nonzero coefficient $c_i$.
		
		If both copies use endpoints, the edges force the two constant strings
		to be complementary.  This gives exactly the other two EO patterns, with
		unique extensions of value \(ab\).  A mixed endpoint--central extension is
		impossible: for $d\ge3$, some central matching edge has both variables connected to
		the constant copy and is forced equal.  Thus each nonzero coefficient in \eqref{eq:auto-wt-18a} has exactly one nonzero extension.
		
		A quaternary member of \(\mathcal B_\times\) has exactly four supported
		strings, so the six-string support excludes membership in this class.
		It also contains the two triples displayed in the statement, whose xor
		identities follow directly.
	\end{proof}

The first identity rules out the positive sign and the second rules out
		the negative sign in the support clause of \(\EOTract\).  Hence
		\(\neg\EOTract(\{h\})\), and \cref{thm:external-eo} makes
		\(\KHolant(\{h\})\) \(\SharpP\)-hard.  Since
		\[
		\KHolant(\{h\})
		\leT\KHolant(\cF,h)
		\leT\KHolant(\cF),
		\]
		the problem \(\KHolant(\cF)\) is \(\SharpP\)-hard.
	
	\subsection{Arity-Six Loop and Two-Copy Contractions}
	\label{subsec:hc-q6-base}
	
	The fifteen self-loop contractions alone do not force the required
	restriction on \(Q\), as the counterexample in
	\cref{rem:cert-high-charge-q6-one-loop-insufficient} shows. We therefore
also use two-copy contractions. If every output belongs to
\(\mathcal B_\times^0\), the finite classification in
\cref{thm:cert-high-charge-q6-interface} forces the support of \(Q\) into
one complementary pair. The calculation and its solver dependence are
given in \cref{prop:cert-high-charge-q6-joint-deck}.
	
	\begin{definition}[Loop and two-copy contractions of an arity-six mountain]
		\label{def:hc-q6-joint-deck}
		Let \(U=[6]\), and let
		\[
		M=a\delta_{\zero^6}+b\delta_{\one^6}+Q,
		\qquad
		Q=\sum_{S\in\binom U3}q_S\delta_{\one_S},
		\qquad ab\ne0.
		\]
		For each unordered pair \(\{i,j\}\in\binom U2\), let
		\(\partial_{ij}M\) be the quaternary EO signature obtained by one
		\(X\)-self-loop, with the residual variables in increasing order.  For
		\(R\in\binom{U\setminus\{i,j\}}2\),
		\begin{equation}
			(\partial_{ij}M)(\one_R)=q_{R\cup\{i\}}+q_{R\cup\{j\}}.
			\label{eq:hc-q6-loop-deck}
		\end{equation}
		These are the fifteen one-copy contractions.
		
		For a two-copy contraction, choose unordered two-element external sets
		\(A,B\in\binom U2\), list them in increasing order
		\[
		A=(a_0,a_1),
		\qquad
		B=(b_0,b_1),
		\]
		and retain the first/second labels of the two copies of \(M\).
		Put \(I_A=U\setminus A\) and \(I_B=U\setminus B\), and choose a
		bijection \(\pi:I_A\to I_B\).  Leave \(A\) external on the first copy
		and \(B\) external on the second copy, in the order
		\((a_0,a_1,b_0,b_1)\), and join every \(r\in I_A\) to
		\(\pi(r)\in I_B\) by an \(X\)-edge.  For
		\({\mathbf x},{\mathbf y}\in\{0,1\}^2\), put
		\[
		A_{\mathbf x}=\{a_r:x_r=1\},
		\qquad
		B_{\mathbf y}=\{b_r:y_r=1\}.
		\]
		The resulting quaternary EO signature \(C_{A,B,\pi}\) has, when
		\(\wt({\mathbf x})+\wt({\mathbf y})=2\),
		\begin{equation}
			C_{A,B,\pi}({\mathbf x},{\mathbf y})
			=ab\,\one[({\mathbf x},{\mathbf y})\in\{(00,11),(11,00)\}]
			+\sum_{\substack{T\subseteq I_A\\|T|=3-\wt({\mathbf x})}}
			q_{A_{\mathbf x}\cup T}\,
			q_{B_{\mathbf y}\cup(I_B\setminus\pi(T))}.
			\label{eq:hc-q6-two-copy-deck}
		\end{equation}
		The first term comes from complementary endpoint choices on the two
		copies, and the sum comes from two central assignments. A mixed choice
		would force four variables of the central copy to be equal, which is
		incompatible with its weight-three support. The endpoint choices leave
		two external ones, as do two central assignments after the four internal
		\(X\)-edges each account for one bit equal to one. Thus the output is EO.
		There are \(15^2\cdot4!=5400\) labelled two-copy contractions.
	\end{definition}
	
	\begin{lemma}[Support forced by arity-six contractions]
		\label{lem:hc-q6-joint-rigidity}
		Let \(M\) and its contractions \(\partial_{ij}M\), \(C_{A,B,\pi}\) be as in
		\cref{def:hc-q6-joint-deck}.  Suppose that every one of the fifteen
		signatures \(\partial_{ij}M\) and every one of the \(5400\) labelled
		signatures \(C_{A,B,\pi}\) belongs to
		\(\mathcal B_\times^0([4])\) on its four ordered dangling edges.  Then
		\begin{equation}
		\operatorname{supp}(Q)\subseteq\{\one_S,\one_{U\setminus S}\}
		\qquad\text{for some }S\in\binom U3.
		\label{eq:auto-q6-5}
		\end{equation}
	\end{lemma}
	
	\begin{proof}
		With \(t=ab\), the signatures \(\partial_{ij}M\) and \(C_{A,B,\pi}\) in
		\cref{eq:hc-q6-loop-deck,eq:hc-q6-two-copy-deck} are precisely the loop
		and labelled two-copy arrays covered by
		\cref{thm:cert-high-charge-q6-interface}.  It applies to arbitrary complex \(q_S\)
        and every \(ab\ne0\), and gives \eqref{eq:auto-q6-5}.
	\end{proof}
	
	\begin{theorem}[A quaternary EO output outside \(\mathcal B_\times\) from an arity-six mountain]
		\label{thm:hc-q6-complete}
		Let \(\cF\) be a finite algebraic signature set containing
		\[
		M=a\delta_{\zero^6}+b\delta_{\one^6}+Q,
		\qquad ab\ne0,
		\qquad Q\text{ EO}.
		\]
		There is a nonzero quaternary EO signature
		\(h\notin\mathcal B_\times\) such that
		\[
		\KHolant(\cF,h)\leT\KHolant(\cF).
		\]
		The signature \(h\) can be chosen as one of the fifteen one-copy contractions
		\(\partial_{ij}M\), one of the \(5400\) labelled two-copy contractions
		\(C_{A,B,\pi}\) from \cref{def:hc-q6-joint-deck}, or the explicit
		two-copy output of \cref{lem:hc-transversal-cube-terminal}.
	\end{theorem}
	
	\begin{proof}
		Suppose first that \(Q\in\mathcal B_\times\). Write
		\[
		Q=\lambda\bigotimes_{r=0}^2
		\begin{pmatrix}0&p_r\\q_r&0\end{pmatrix}_{w_{2r},w_{2r+1}},
		\qquad \lambda p_0q_0p_1q_1p_2q_2\ne0.
		\]
		The case \(d=3\) of \cref{lem:hc-transversal-cube-terminal} gives the
		six-string output \eqref{eq:auto-wt-18a}, with one nonzero extension per
		supported string. Choose this output as \(h\).
		
		Now suppose \(Q\notin\mathcal B_\times\). If \(Q=0\), a contraction
		\(C_{A,B,\pi}\) has exactly the two endpoint--endpoint coefficients
		\(ab\). Its support has size two, whereas a quaternary member of
		\(\mathcal B_\times\) has support size four. Choose this contraction as \(h\).
		It remains to consider \(Q\ne0\). If a loop or two-copy contraction
		lies outside \(\mathcal B_\times^0\), choose it as \(h\).
		Otherwise \cref{lem:hc-q6-joint-rigidity} gives
		\(\operatorname{supp}(Q)\subseteq\{\one_S,\one_{U\setminus S}\}\).
		Choose \(i\in S\) and \(j\notin S\). The two possible nonzero
        coefficients of \(\partial_{ij}M\) have the unique extensions
        \(\one_S,\one_{U\setminus S}\). Hence this contraction is nonzero with support at
		most two, contradicting its membership in \(\mathcal B_\times^0\).
		
		In each case, the selected \(h\) is the output of a direct gadget over
		\(\{M\}\), so \cref{eq:direct-gadget-reduction} gives the required reduction.
	\end{proof}

	\begin{proof}[Proof of \cref{thm:height-three-complete}]
		Apply zero/nullary preprocessing and assume \(\neg\Tract(K\cF)\).
		By \cref{thm:hc-q6-complete}, there is a nonzero quaternary EO signature
		\(h\notin\mathcal B_\times\) with
		\(\KHolant(\cF\cup\{h\})\leT\KHolant(\cF)\).
		The finite algebraic set \(\cF\cup\{h\}\) retains the mountain, whose
		endpoints witness both charge signs. The added signature has charge zero,
		so the charge divisibility condition is also preserved.
		By \cref{lem:downward-inheritance}, \(\Tract(K(\cF\cup\{h\}))\) fails.
		Since \(\mathcal B_{\rm ARS}\subseteq\mathcal B_\times\), we have
		\(h\notin\mathcal B_{\rm ARS}\). Thus
		\cref{thm:hc-relative-eo-terminal} applies to \(\cF\cup\{h\}\), and the
		preceding reduction transfers hardness to the original \(\cF\).
	\end{proof}
	
	\section{Mountain Height at Least Four}
	\label{sec:height-at-least-four}
	
	\begin{theorem}[Mountain dichotomy at every height at least four]
		\label{thm:height-at-least-four}
		Let \(d\ge4\), and let \(\cF\) be a finite algebraic set of even-arity
		signatures such that
		\[
		\chi_r({\mathbf x})\in d\mathbb Z
		\qquad(f\in\cF,\ \arity(f)=r,\ {\mathbf x}\in\supp(f)).
		\]
		Suppose that \(\cF\) contains a mountain
		\[
		M=a\delta_{\zero^{2d}}+b\delta_{\one^{2d}}+Q,
		\qquad ab\ne0,\qquad Q\text{ EO}.
		\]
		Then \(\Tract(K\cF)\) holds or \(\KHolant(\cF)\) is
		\(\SharpP\)-hard.
	\end{theorem}
	
	An \emph{\(X\)-matching product} is a nonzero scalar multiple of
	\(\bigotimes_{e\in P}X_e\), where \(P\) is a perfect matching of its
	variables. Such binary factors are irreducible, so
	uniqueness of irreducible factorization identifies both their variable pairs
	and their projective values.
	
	We first normalize the complete mountain so that all pair contractions
	of its central restriction are \(X\)-matching products. A common scalar
	then makes that restriction real-valued and complement-invariant, allowing
	the reconstruction theorem of Shao and Cai to apply. For \(d\ge5\), it
	gives a matching-product central restriction, and the two-copy construction
	of \cref{lem:hc-transversal-cube-terminal} finishes the case.
	At height four, the same conclusion holds when the contraction matchings
	satisfy the overlap condition below. Otherwise we recover the Reed--Muller
	central restriction, normalize its endpoint coefficients, and obtain
	equality by a Turing reduction.
	
	\begin{remark}[Relation to the real-valued Holant argument]
		\label{rem:sc20-high-height}
		The normalization by attaching binary signatures below follows the contraction and binary-attachment
		strategy of \cite[full version, Lemma~8.1]{ShaoCai2020}.  Their transformed
		signature set satisfies ARS globally.  Here the signature set is
		complex-valued: charge conservation makes the self-loop outputs EO, and
		\cref{thm:hc-relative-eo-terminal} forces their binary ARS factorization
		unless hardness has already been proved.  This supplies the factorization needed for the binary attachments.
		
        The constructions in \cref{lem:hc-port-normalization} use the complete
        mountain and its contractions. Together with the phase recovery in
        \cref{lem:hc-global-phase}, they establish the hypotheses of
        \cref{lem:sc-x-product-reconstruction} for the central restriction.
	\end{remark}
	
	\subsection{Common Normalization and Phase Recovery}
	\label{subsec:hc-normalization}
	
	For the proof of \cref{thm:height-at-least-four}, apply zero/nullary
	preprocessing and assume \(\neg\Tract(K\cF)\).
	If \(Q=0\), \cref{lem:pure-ge} already gives hardness.  Hence it suffices to
	consider \(Q\ne0\).
	Every self-loop on \(M\) kills its endpoints:
	\begin{equation}
	\partial_{ij}M=\partial_{ij}Q.
	\label{eq:hc-central-card}
	\end{equation}
	More generally, charge conservation keeps every supported half-charge
	in \(d\mathbb Z\). For an output of positive arity less than \(2d\), its
	absolute value is less than \(d\), so it must be zero. Such outputs are EO.  By
	\cref{thm:hc-relative-eo-terminal}, any such nonzero output outside
	\(\mathcal B_{\rm ARS}\) gives hardness.  Otherwise its factors are
	nonsingular binaries
	\[
	B(z)=\begin{pmatrix}0&z\\\bar z&0\end{pmatrix},
	\qquad z\in\AlgNums^\times,
	\]
	up to nonzero scalars.
	These factors are available by \cref{thm:external-factor}: the mountain
	is two-sided, excluding the common-one-sided exception.  The constructions
	below use finitely many such factors and compose their reductions.
	The extracted binaries have charge zero, and
	\cref{lem:downward-inheritance} excludes tractability for the corresponding
	oracle problems.
	
	To normalize a selected binary factor to \(X\), we attach a copy of that
	factor to one of its variables. Attaching the first variable of an ordered \(B(z)\) to a mountain variable through
	\(X\) acts on that variable by
	\begin{equation}
	\begin{aligned}
T_{B(z)^{\mathsf T}}&=B(z)^{\mathsf T}X=\diag(\bar z,z)=\bar zD(z/\bar z),\\
T_{B(z)^{\mathsf T}}B(z)&=|z|^2X.
\end{aligned}
	\label{eq:hc-port-mating}
	\end{equation}
	This invertible diagonal map preserves support and both nonzero endpoints.
	To preserve a previously normalized contraction, we use the following
	criterion for consistently oriented factors:
	\begin{equation}
	B(z)^{\mathsf T}XB(w)\doteq X
	\quad\Longleftrightarrow\quad
	w/z\in\mathbb R^\times
	\quad\Longleftrightarrow\quad B(w)\doteq B(z).
	\label{eq:hc-binary-splice}
	\end{equation}
	Indeed the two off-diagonal entries are \(\bar z w\) and \(z\bar w\).
	In this case \(T_{B(z)^{\mathsf T}}^{\mathsf T}XT_{B(w)^{\mathsf T}}\) is also a nonzero scalar multiple of
	\(X\).  All these operations use the extracted ordered binaries on the
	complete mountain.  
	
	\begin{lemma}[Vanishing pair contractions]
		\label{lem:wt-zero-loop}
		Let \(h\) have arity at least three.  If \(\partial_{ij}h=0\) for every
		pair \(i\ne j\), then \(h({\mathbf x})=0\) whenever
		\(0<\wt({\mathbf x})<\arity(h)\).
		In particular, an EO signature with all pair contractions zero is zero.
	\end{lemma}
	
	\begin{proof}
		This is \cite[Lemma~2.7]{CaiFuShao2020ARS}. Its three-equation
		argument is purely linear and uses neither ARS nor nonvanishing contractions.
	\end{proof}
	
	We also write \(\partial_e\) when \(e=\{i,j\}\).
	
	The next lemma propagates the normalized form from selected contractions
	to all pair contractions of the same EO tensor.
	\begin{lemma}[Propagation of \(X\)-matching products]
		\label{lem:hc-x-line-propagation}
		Let \(F\) be EO on an even variable set \(V\), with \(|V|\ge8\), and suppose
		every nonzero \(\partial_eF\) lies in \(\mathcal B_{\rm ARS}\).
		Either of the following conditions implies that every \(\partial_eF\)
		is an \(X\)-matching product.
		\begin{enumerate}
			\item Every \(\partial_eF\) with \(e\cap\{1,2\}=\varnothing\) is an
			\(X\)-matching product.
			\item Both \(\partial_{12}F\) and \(\partial_{ij}F\) are
			\(X\)-matching products,
			\(\{1,2\}\cap\{i,j\}=\varnothing\), and the matching of
			\(\partial_{ij}F\) pairs \(1,2\) with distinct mates \(u,v\).
		\end{enumerate}
		The remaining pair contractions are allowed to be zero in the hypotheses.
	\end{lemma}
	
	\begin{proof}
		Contracting any pair in an \(X\)-matching product gives another nonzero
		\(X\)-matching product: it either removes a matching edge with scalar
		\(2\), or splices two edges with scalar \(1\).  Thus, whenever a deleted
		pair \(e\) is disjoint from a pair \(r\) for which
		\(\partial_rF\) is an \(X\)-matching product,
		\begin{equation}
		0\ne\partial_e\partial_rF=\partial_r\partial_eF.
		\label{eq:hc-nonzero-double}
		\end{equation}
		In particular \(\partial_eF\ne0\).  A factor whose variables avoid \(r\)
		survives in this double contraction, so uniqueness of irreducible factorization makes it proportional to \(X\).
		
		For condition~1, only contractions with \(e\cap\{1,2\}\ne\varnothing\)
		remain.  A pair \(r\) disjoint from \(e\), with \(\partial_rF\) an
		\(X\)-matching product, first proves nonvanishing.
		For each binary factor on variables \(a,b\), choose a pair \(r\) avoiding
		\(\{1,2\}\cup e\cup\{a,b\}\).  This excludes at most five variables, so such
		an \(r\) exists and \(\partial_rF\) is an \(X\)-matching product.  The factor survives in
		\eqref{eq:hc-nonzero-double}, proving condition~1.
		
		For condition~2, put \(W=V\setminus\{1,2,i,j\}\), where
		\(u,v\in W\) and \(|W|\ge4\).  For
		\(e\in\binom W2\setminus\{\{u,v\}\}\), nonvanishing follows using the
		\(12\)-contraction.  The variables \(1,2\) cannot be paired in \(\partial_eF\):
		that pair would survive deletion of \(ij\), whereas deleting \(e\) from
		the \(X\)-matching product \(\partial_{ij}F\) pairs \(1,2\) only if \(e=\{u,v\}\).
		Consequently write
		\[
		\partial_eF=C(x_1,x_r)D(x_2,x_s)H.
		\]
		The \(12\)-double contraction makes \(H\) an \(X\)-matching product and the splice of
		\(C,D\) proportional to \(X\).  By \eqref{eq:hc-binary-splice},
		\(C\doteq D\) in these orientations.
		Both factors cannot be affected by deletion of \(ij\), since this would
		splice \(1,2\) together, contradicting the other contraction order.
		A surviving factor is proportional to \(X\) by comparison with the
		\(X\)-matching product \(\partial_{ij}F\).  Hence \(C,D\doteq X\), and the whole contraction is an
		\(X\)-matching product.
		
		For \(e=\{u,v\}\), choose two distinct variables \(r,s\in W\setminus\{u,v\}\).
		The contractions at the three disjoint pairs \(12,ij,rs\) are now
		\(X\)-matching products.
		The \(uv\)-contraction is nonzero by the \(12\)-comparison.  Each of its binary
		factors avoids at least one of these three pairs, and hence is
		proportional to \(X\) by \eqref{eq:hc-nonzero-double}.
		Thus every \(\partial_eF\) with \(e\subseteq W\) is an
		\(X\)-matching product.
		
		Next let \(e=\{i,k\}\) or \(\{j,k\}\), with \(k\in W\).
		It is nonzero by the \(12\)-comparison.  For a factor on \(a,b\), use
		\(r=12\) if \(\{a,b\}\) avoids \(12\). Otherwise choose
		\[
		r\in\binom{W\setminus(\{k\}\cup\{a,b\})}{2}.
		\]
		At most two variables of \(W\) are excluded, so the choice exists.
		Here \(r\) is disjoint from both \(e\) and the factor's variables, and
		\(\partial_rF\) is an \(X\)-matching product.  The factor remains in
		\eqref{eq:hc-nonzero-double}, so it is proportional to \(X\).
		Along with the already established \(X\)-matching product
		\(\partial_{ij}F\), this covers all deleted pairs avoiding
		\(12\).  Condition~1 completes the proof.
	\end{proof}
	
	We now apply this propagation to \(M\). We first
	normalize a nonzero contraction, then adjust the remaining factors while
	preserving that normalized contraction. The two cases below correspond
	to the two conditions of \cref{lem:hc-x-line-propagation}.
	\begin{lemma}[Normalization by binary attachments at every arity at least eight]
		\label{lem:hc-port-normalization}
		Under the hypotheses of \cref{thm:height-at-least-four}, assume
		\(\neg\Tract(K\cF)\) and \(Q\ne0\).
		Then either \(\KHolant(\cF)\) is \(\SharpP\)-hard, or a finite equivalent
		enlargement of \(\cF\), still supported on half-charges in
		\(d\mathbb Z\), contains a mountain
		\[
		M^\circ=A\delta_{\zero^V}+B\delta_{\one^V}+F,
		\qquad |V|=2d,\quad AB\ne0,\quad 0\ne F\text{ EO},
		\]
		such that
		\begin{equation}
		\partial_eF=\lambda_e\bigotimes_{t\in P_e}X_t,
		\qquad \lambda_e\in\AlgNums^\times
		\qquad\left(e\in\binom V2\right).
		\label{eq:hc-normalized-deck}
		\end{equation}
		Here \(P_e\) is a perfect matching of \(V\setminus e\).
		
	\end{lemma}
	
	\begin{proof}
		At every stage retain the mountain, add the finitely many
		contractions or factors used, and apply \cref{thm:hc-relative-eo-terminal} again.
		The charge condition and \(\neg\Tract(K\cF)\) are preserved.
		Thus every current nonzero contraction lies in \(\mathcal B_{\rm ARS}\), or
		hardness has already been proved.
		
		By \cref{lem:wt-zero-loop}, some contraction of \(Q\) is nonzero.
		Relabel it as \(\partial_{12}Q\).
		Extract its \(d-1\) binary factors and attach a copy of each to one
		of its variables on \(M\), with the orientation in
		\eqref{eq:hc-port-mating}.
		The resulting mountain \(M_0\) has nonzero endpoints and central restriction
		\(F_0\), with \(\partial_{12}F_0\) an \(X\)-matching product.
		For every pair \(e\) avoiding \(12\),
		\(\partial_{12}\partial_eF_0=\partial_e\partial_{12}F_0\ne0\).
		Hence all these contractions are nonzero.
		
		Suppose first that \(1,2\) form a matching edge in every such contraction.
		Write \(\partial_eF_0=C_e(x_1,x_2)H_e\).
		The \(12\)-comparison shows that the contraction of \(C_e\) is nonzero
		and \(H_e\) is an \(X\)-matching product.  If \(e,f\) are disjoint pairs avoiding
		\(12\), their double contraction is nonzero: in the \(e\)-order, deletion of
		\(f\) acts only on the \(X\)-matching product \(H_e\).
		The factor on \(12\) survives in both orders, so \(C_e\doteq C_f\).
		The graph of disjoint pairs on \(V\setminus\{1,2\}\) is connected.
		Indeed, intersecting distinct pairs exclude only three of its at least
		six variables, leaving a pair disjoint from both.
		Thus all \(C_e\) are proportional to one ordered factor \(C\).
		Extract \(C\) from one contraction and attach a copy to variable \(1\) of \(M_0\)
		with the orientation in \eqref{eq:hc-port-mating}.  Every contraction whose deleted pair avoids \(12\) is now an
		\(X\)-matching product.  After reapplying
		\cref{thm:hc-relative-eo-terminal}, condition~1 of
		\cref{lem:hc-x-line-propagation} makes every pair contraction an
		\(X\)-matching product.
		
		Otherwise choose a pair \(ij\) avoiding \(12\) whose contraction has distinct
		mates \(u,v\) for \(1,2\), and write
		\[
		\partial_{ij}F_0=C(x_1,x_u)D(x_2,x_v)H.
		\]
		The \(12\)-comparison makes \(H\) an \(X\)-matching product and gives
		\(C\doteq D\) by \eqref{eq:hc-binary-splice}.
		Extract both factors from this contraction and attach their copies to variables \(1,2\) of
		\(M_0\) in the corresponding orientations.  The new \(ij\)-contraction is an \(X\)-matching product.
		The attachments act by \(T_{C^{\mathsf T}}=C^{\mathsf T}X\) and
		\(T_{D^{\mathsf T}}=D^{\mathsf T}X\), as in \eqref{eq:hc-port-mating}.  The effective
		\(12\)-kernel is \(T_{C^{\mathsf T}}^{\mathsf T}XT_{D^{\mathsf T}}\doteq X\), so the
		\(12\)-contraction remains an \(X\)-matching product as well.  The
		diagonal maps preserve the distinct mates of variables \(1,2\).  Reapply
		\cref{thm:hc-relative-eo-terminal} to the current contractions, then use condition~2 of \cref{lem:hc-x-line-propagation}.
		
		All extensions are invertible and preserve the central support and the
		nonzero endpoints.  Their number depends only on the fixed source
		signature.  Factor extraction and gadget replacement give the
		claimed finite equivalence.  Every contraction added here is a contraction of the mountain
		and agrees with the corresponding contraction of its central restriction.
	\end{proof}
	
	The normalized contractions may still carry different complex scalars.
	To apply the reconstruction theorem, we next show that one common scalar
	makes all these coefficients positive and the central tensor satisfy ARS.
	\begin{lemma}[A common phase for EO contractions]
		\label{lem:hc-global-phase}
		Let \(F\) be algebraic and EO on an even variable set \(V\), with \(|V|\ge6\), and suppose
		\eqref{eq:hc-normalized-deck} holds with every \(\lambda_e\ne0\).
		There is \(c\in\AlgNums^\times\) such that \(G=c^{-1}F\) is real-valued
		and complement-invariant, and every coefficient
		\(\lambda_e/c\) in \eqref{eq:hc-normalized-deck} is positive real.
		In particular \(G\) satisfies ARS.
	\end{lemma}
	
	\begin{proof}
		For disjoint pairs \(e,f\), the two contraction orders give
		\[
		\partial_f\partial_eF=m_{ef}\lambda_e X_{P'},
		\qquad
		\partial_e\partial_fF=m_{fe}\lambda_f X_{P''},
		\qquad m_{ef},m_{fe}\in\{1,2\}.
		\]
		Here \(X_P=\bigotimes_{t\in P}X_t\) has only \(0,1\) entries.
		Both tensors are nonzero, so commutation identifies their supports and
		gives
		\[
		\lambda_e/\lambda_f=m_{fe}/m_{ef}\in\{1,2,1/2\}.
		\]
		The graph of disjoint pairs on \(V\) is connected: two intersecting pairs
		have a common disjoint pair since \(|V|\ge6\).
		Fix one pair \(e_0\) and take \(c=\lambda_{e_0}\).
		Every contraction of \(G=c^{-1}F\) is therefore real and
		complement-invariant, with a positive coefficient.
		
		Let \(\overline G\) denote coefficientwise conjugation and
		\(G^{\rm rev}({\mathbf x})=G(\bar {\mathbf x})\).
		Contraction commutes with both operations.  Hence every pair contraction
		of \(G-\overline G\) and of \(G-G^{\rm rev}\) vanishes.
		Both differences are EO, so \cref{lem:wt-zero-loop} makes them zero.
		This is the algebraic normalization used in the reconstruction theorem.
	\end{proof}
	
	We use the following algebraic part of the real-valued
	classification.
	
	\begin{lemma}[Shao--Cai matching-product reconstruction
		{\cite[full version, Lemma~2.11]{ShaoCai2020}}]
		\label{lem:sc-x-product-reconstruction}
		Let \(G\) be an EO tensor of even arity \(n\ge8\) satisfying ARS.
		Suppose
		\[
		\partial_eG=\mu_e\bigotimes_{t\in P_e}X_t,
		\qquad \mu_e\in\mathbb R^\times
		\qquad\left(e\in\binom{[n]}2\right),
		\]
		where each \(P_e\) is a perfect matching of the residual variables.
		If \(n\ge10\), then \(G\) is an \(X\)-matching product.
		The same conclusion holds for \(n=8\) if \(P_A\cap P_B\ne\varnothing\)
		for some \(A\ne B\) with \(A\cap B\ne\varnothing\). In that case the
		\(X\)-factor on that shared edge is also a factor of \(G\).
	\end{lemma}
	
	\subsection{Height at Least Five}
	\label{subsec:height-at-least-five}
	
	\begin{proof}[Proof of \cref{thm:height-at-least-four} for \(d\ge5\)]
		As above, assume \(\neg\Tract(K\cF)\) and \(Q\ne0\).
		Apply \cref{lem:hc-port-normalization}.  Outside its hard outcome there
		is a mountain
		\(M^\circ=A\delta_{\zero^{2d}}+B\delta_{\one^{2d}}+F\)
		with \(AB\ne0\) and every pair contraction of \(F\) an
		\(X\)-matching product.
		By \cref{lem:hc-global-phase}, \(G=c^{-1}F\) is EO with ARS
		and nonzero real \(X\)-matching product contractions.  Since \(2d\ge10\),
		\cref{lem:sc-x-product-reconstruction} makes \(G\), and hence \(F\),
		an \(X\)-matching product.
		Apply \cref{lem:hc-transversal-cube-terminal} to
		\(M^\circ\).  Its two-copy output is hard, and the finite equivalence
		transfers hardness to the original \(\cF\).
	\end{proof}
	
	\subsection{Height Four}
	\label{subsec:height-four}

	At arity eight, the reconstruction theorem needs an overlap between
	contraction matchings. Without that overlap, the remaining matching family
	leads to the Reed--Muller structure described below. We recover its central
	coefficients first, then compare them with the two endpoint coefficients.

	\begin{lemma}[A shared edge forces a matching product]
		\label{lem:q8-scope-collision}
		Let \(F\) be an arity-eight EO tensor satisfying
		\eqref{eq:hc-normalized-deck}.
		If \(P_A\cap P_B\ne\varnothing\) for some \(A\ne B\) with
		\(A\cap B\ne\varnothing\), then \(F\) is an \(X\)-matching product.
	\end{lemma}
	
	\begin{proof}
		Apply \cref{lem:hc-global-phase} and then the arity-eight clause of
		\cref{lem:sc-x-product-reconstruction}.
	\end{proof}

Under the hypotheses of \cref{lem:q8-scope-collision}, any finite
		algebraic signature set containing a mountain
		\(M^\circ=A\delta_{\zero^8}+B\delta_{\one^8}+F\), with \(AB\ne0\),
		has a \(\SharpP\)-hard \(K\)-Holant problem, by
		\cref{lem:hc-transversal-cube-terminal}.

We now consider the remaining case: \(P_A\cap P_B=\varnothing\) for
distinct intersecting deleted pairs \(A,B\).

	\subsubsection{Translation Form of the Matching Family}
	
	\begin{lemma}[Translation form of the matching family]
		\label{lem:q8-collision-free-atlas}
		Let \(F\) be an arity-eight EO tensor on a variable set \(V\), and suppose
		that for every \(u\ne v\) one has
		\[
		\partial_{uv}F
		=\lambda_{uv}\bigotimes_{e\in P_{uv}}X_e,
		\qquad \lambda_{uv}\ne0,
		\]
		where \(P_{uv}\) is a perfect matching of \(V\setminus\{u,v\}\).
		Suppose \(P_A\cap P_B=\varnothing\) whenever \(A\ne B\) and
		\(A\cap B\ne\varnothing\).  Then there is a bijection
		\(\iota:V\to\Ftwo^3\) such that, after relabelling each variable \(v\) by
		\(\iota(v)\), for $p\ne q$ and $h=p+q$ one has
		\begin{equation}
		P_{pq}
		=\bigl\{\{t,t+h\}:t\in \Ftwo^3\setminus\{p,q\}\bigr\},
		\label{eq:auto-q8-14}
		\end{equation}
		where duplicate unordered pairs are suppressed.
	\end{lemma}
	
	\begin{proof}
		We apply uniqueness of irreducible factorization and the
		computer-assisted exhaustive matching classification in
		\cref{prop:cert-q8-collision-free-atlas}.  For disjoint deleted pairs
		$e,r$, let $\kappa_r(P_e)$ be the two-edge matching left after contracting
		$r$ in the $X$-matching product on $P_e$: delete $r$ when $r\in P_e$, and
		otherwise splice the mates of the two ends of $r$ and retain the untouched
		edge.  Both operations are nonzero for $X$ factors.  Commutativity of the
		two loops and uniqueness of irreducible factorization give the exact constraint
		\begin{equation}
		\kappa_r(P_e)=\kappa_e(P_r)
		\qquad(e\cap r=\varnothing).
		\label{eq:auto-q8-14a}
		\end{equation}
		
		The condition on \(P_A\cap P_B\) and \eqref{eq:auto-q8-14a} are exactly the
		hypotheses of \cref{thm:cert-q8-collision-free-interface}, which supplies
        the required bijection and \eqref{eq:auto-q8-14}.
	\end{proof}
	
	Index the eight variables by the points of \(\Ftwo^3\) in lexicographic order.
	The first-order Reed--Muller code is
	\[
	RM(1,3)
	=
	\{(\ell(t))_{t\in\Ftwo^3}\mid
	\ell\colon\Ftwo^3\to\Ftwo\text{ is affine}\}.
	\]
	Explicitly, its elements are the truth tables of
	\[
	\ell(t_1,t_2,t_3)
	=a+u_1t_1+u_2t_2+u_3t_3,
	\qquad
	a,u_1,u_2,u_3\in\Ftwo,
	\]
	with all arithmetic over \(\Ftwo\).  Evaluating such an \(\ell\) on the
	eight inputs \(t\in\Ftwo^3\), in the fixed order above, gives one binary
	string of length eight.  The \(2^4=16\) choices of
	\((a,u_1,u_2,u_3)\) give exactly the sixteen strings of \(RM(1,3)\).
	For \({\mathbf x}\in\Ftwo^8\), we write
	\begin{equation}
		\RMcore({\mathbf x})=\one[{\mathbf x}\in RM(1,3)].
		\label{eq:RM-core}
	\end{equation}
	
	The arity-eight exception is the central restriction of the Reed--Muller signature
	\begin{equation}
	R_8^{\rm cen}({\mathbf x})=
	\one\left[{\mathbf x}=(a+{\mathbf u}\cdot t)_{t\in\Ftwo^3}
	\text{ for some }a\in\Ftwo,\ 0\ne {\mathbf u}\in\Ftwo^3\right].
	\label{eq:auto-wt-20}
	\end{equation}
	This tensor has fourteen supported strings and is not a matching product, but each contraction
	is an \(X\)-matching product.  For deleted variables \(p,q\), the surviving pairs are
	\(\{t,t+p+q\}\): each choice of one endpoint from every remaining pair
	has a unique extension with
	\({\mathbf u}\cdot(p+q)=1\).  A matching edge can therefore belong to both \(P_A\) and \(P_B\),
	for distinct \(A,B\), only when \(A\cap B=\varnothing\), so the arity-eight overlap clause of
	\cref{lem:sc-x-product-reconstruction} does not apply.
	
	For a set of eight variables \(V\) and a bijection \(\iota:V\to\Ftwo^3\), write
	\[
	(\iota_*F)((x_t)_{t\in\Ftwo^3})
	=F((x_{\iota(v)})_{v\in V}).
	\]
	
	The matching family is now fixed, but the coefficients of the original
	tensor are not yet determined. The next lemma recovers them from the
    contraction equations, using the symbolic coefficient argument in
    \cref{prop:cert-q8-canonical-support}.
	\begin{lemma}[The matching family determines the Reed--Muller central restriction]
		\label{lem:q8-canonical-support}
		Let \(F\) be an arity-eight EO tensor on a variable set \(V\) such that
		\[
		\partial_{uv}F
		=\lambda_{uv}\bigotimes_{e\in P_{uv}}X_e,
		\qquad \lambda_{uv}\ne0,
		\]
		for every $u\ne v$, where $P_{uv}$ is a perfect matching of
		$V\setminus\{u,v\}$.  Suppose \(P_A\cap P_B=\varnothing\) whenever
		\(A\ne B\) and \(A\cap B\ne\varnothing\).  Choose a bijection
		\(\iota:V\to\Ftwo^3\) supplied by
		\cref{lem:q8-collision-free-atlas}, put \(F^\iota=\iota_*F\), and
		write
		\[
		c_S=F^\iota(\one_S),
		\qquad S\in\binom{\Ftwo^3}4.
		\]
		Then
		\begin{equation}
		\supp(F^\iota)
		=\left\{\one_S:
		S\in\binom{\Ftwo^3}4,\ 
		\bigoplus_{s\in S}s=0\right\},
		\label{eq:auto-q8-17}
		\end{equation}
		the indicator vectors of the fourteen affine planes.  Moreover all fourteen nonzero
		coefficients in \eqref{eq:auto-q8-17} are equal.
	\end{lemma}
	
	\begin{proof}
		Relabel the variables through \(\iota\), continue to denote the resulting tensor
		by \(F^\iota=\iota_*F\), and write \(V=\Ftwo^3\).  For
		$S\in\binom V4$, $p\in S$, and $q\notin S$, put
		\[
		S'=S\mathbin\triangle\{p,q\},\qquad
		T=S\setminus\{p\},\qquad
		\sigma(S)=\bigoplus_{s\in S}s,\qquad h=p+q.
		\]
        The contraction coefficient at $\one_T$ is
		\[
		(\partial_{pq}F^\iota)(\one_T)=c_S+c_{S'}.
		\]
		The string \(\one_T\) selects one endpoint from each edge of the
		matching \eqref{eq:auto-q8-14} if and only if
		\begin{equation}
		\sigma(S)\in\{0,h\}.
		\label{eq:auto-q8-19}
		\end{equation}
		Indeed, in $V/\Span_{\Ftwo}\{h\}$ the fibre $\{p,q\}$ contains one point of
		$S$, while the other occupancies are $1,1,1$ or a permutation of $2,1,0$.
		Their quotient xor is respectively zero or the sum of the two odd fibres,
		proving \eqref{eq:auto-q8-19}.
		Coefficients at these strings equal $\lambda_{pq}$ and all others vanish.
		Thus the \(c_S\) and nonzero \(\lambda_{pq}\) satisfy the swap rule in
		\cref{thm:cert-q8-canonical-support-interface}.  That theorem eliminates
		all coefficients on four-sets that are not affine planes and identifies
		the fourteen plane coefficients
		with one \(c\ne0\), proving \eqref{eq:auto-q8-17}.
	\end{proof}
	
	Equivalently, the central restriction of the Reed--Muller signature from \eqref{eq:auto-wt-20} is
	\[
	R_8^{\rm cen}({\mathbf x})
	=\one\left[\wt({\mathbf x})=4,
	\bigoplus_{t:x_t=1}t=0\right].
	\]
    Together with the two constant strings, this gives
    \[
    \RMcore=\delta_{\zero^8}+\delta_{\one^8}+R_8^{\rm cen}.
    \]

	\subsubsection{Normalizing the Endpoint Coefficients}
	
	An information set of a binary code is a set of coordinates on which
	restriction is a bijection onto the corresponding binary coordinate space.
	By \cref{lem:q8-canonical-support}, all central coefficients have the
same value \(c\). The remaining question is whether the endpoint product
\(AB\) equals \(c^2\).  Restriction to an
information set determines an affine string uniquely, so joining two
copies along such a set makes the endpoint product visible in an arity-six
output.  This is the construction in the next lemma.
	
	\begin{lemma}[Normalizing the endpoint coefficients]
		\label{lem:q8-canonical-endpoint-repair}
		Let
		\[
		M^\circ=A\delta_{\zero^8}+B\delta_{\one^8}+cR_8^{\rm cen},
		\qquad ABc\ne0.
		\]
		Then either the two-copy \(X\) construction below yields a nonzero
		arity-six EO signature outside
		\(\mathcal B_\times\), or
		\begin{equation}
		AB=c^2
		\label{eq:q8-canonical-endpoint}
		\end{equation}
		and, for any algebraic \(s\) with \(s^8=c/A\),
		\begin{equation}
		\begin{gathered}
		\bigl(sD(s^{-2})\bigr)^{\otimes8}M^\circ=c\RMcore,\\
		sD(s^{-2})=\diag(s,s^{-1}),
		\qquad
		\bigl(sD(s^{-2})\bigr)^{\mathsf T}X\bigl(sD(s^{-2})\bigr)=X.
		\end{gathered}
		\label{eq:q8-canonical-normalization}
		\end{equation}
		The first outcome uses two source vertices, four information-set edges, and
		one final \(X\)-closure.
	\end{lemma}
	
	\begin{proof}
		Index the variables by \(V=\Ftwo^3\) and take the complementary affine
		information sets
		\[
		I_0=\{000,100,010,001\},
		\qquad
		I_1=V\setminus I_0.
		\]
		Join corresponding \(I_0\)-variables of two copies of \(M^\circ\) by four
		\(X\)-edges.  An affine string is determined by its values on either
		information set.  Hence the internal equations force the two
		\(RM(1,3)\) strings to be complementary on all eight coordinates.  The
		remaining variables form four opposite pairs indexed by \(I_1\), and their
		logical coefficient tensor on the four independent pair-choice bits is
		\[
		H({\mathbf y})=
		\begin{cases}
			AB,&{\mathbf y}=0^4\text{ or }1^4,\\
			c^2,&\text{otherwise}.
		\end{cases}
		\]
		Close one of these four pairs by an \(X\)-edge.  On the three
		remaining matching bits the two coefficients at complementary strings are \(AB+c^2\),
		whereas the other six coefficients are \(2c^2\).  If \(AB=-c^2\), the
		output is nonzero but its support omits two of the eight choices of one
		endpoint from each remaining pair.  Otherwise that
		support uniquely determines its three opposite pairs, and a rank-one
		logical coefficient tensor would have all \(2\times2\) minors vanish after fixing
		any one of the three pair-choice bits.  Choosing a minor through one of these complementary strings instead gives determinant
		\[
		\det\begin{pmatrix}
			AB+c^2&2c^2\\
			2c^2&2c^2
		\end{pmatrix}
		=2c^2(AB-c^2).
		\]
		Thus \(AB\ne c^2\) gives the first outcome, and the remaining case is
		\eqref{eq:q8-canonical-endpoint}.
		
		Finally, \(\bigl(sD(s^{-2})\bigr)^{\otimes8}\) multiplies a weight-\(k\) coefficient by
		\(s^{8-2k}\).  It fixes the central restriction, sends \(A\) to \(As^8=c\),
		and sends \(B\) to \(Bs^{-8}=AB/c=c\).  This proves
		\eqref{eq:q8-canonical-normalization}. Direct multiplication gives the
		displayed \(GO(X)\) identity.
	\end{proof}

In the current height-four branch, the first outcome of
\cref{lem:q8-canonical-endpoint-repair} gives hardness by
\cref{thm:hc-relative-eo-terminal}. In the remaining case \(AB=c^2\),
the global transformation in \cref{eq:q8-canonical-normalization} supplies
\(c\RMcore\). We next obtain binary equality by reduction.

	\subsubsection{Obtaining Equality from the Reed--Muller Signature}
	
	\begin{remark}[The Reed--Muller signature in Shao--Cai]
		\label{rem:sc-f8-core}
		With the variable order in \cref{eq:RM-core}, \(\RMcore\) is
		the signature \(f_8=\widehat f_8\) of
		\cite[full version, Section~8]{ShaoCai2020}.  Its Bell contraction identities
		are established in \cite[conference version, Theorem~13 and the
			construction following it]{CaiFuShao2020Entanglement}.
	\end{remark}

    We use the Bell contraction identity and occurrence reduction of
    \cite[full version, Definition~8.7 and Lemma~8.9]{ShaoCai2020}.
    They reduce the number of added Bell vertices to at most two.
    The half-charge condition then lets us evaluate the remaining equality
    vertices by closures over \(\cF\), giving the reduction needed for the
    binary-equality dichotomy. We then add the remaining Bell signatures.
    Recall \(\mathcal B=\{I,X,Y,Z\}\), with the ordered matrices of
    \cref{eq:bell-matrix-dictionary}.
	
	\begin{lemma}[The Reed--Muller Bell reduction]
		\label{thm:q8-normalized-bell-lift}
		Let \(\cF\) be a finite algebraic set of even-arity signatures such that
		\[
		\chi_r({\mathbf x})\in4\mathbb Z
		\qquad
		\bigl(f\in\cF,\ \arity(f)=r,\ {\mathbf x}\in\supp(f)\bigr).
		\]
		Suppose \(\cF\) contains \(c\RMcore\) for a known \(c\ne0\).  Then
		\begin{equation}
		\KHolant(\cF,\mathcal B)\leT\KHolant(\cF).
		\label{eq:q8-relative-bell-reduction}
		\end{equation}
	\end{lemma}
	
	\begin{proof}
        Boundary signatures below are realized by direct \(X\)-gadgets over
        the indicated signature set, with ordered dangling edges. Disconnected
        gadgets are allowed, with closed components contributing scalars.
		
		\smallskip
		\noindent\emph{Occurrence reduction.}
		For each ordered \(b\in\mathcal B\), attaching \(b\) to two variables of
		\(c\RMcore\) leaves, up to a known nonzero scalar and a fixed ordering of
		the six remaining variables, \(b^{\otimes3}\).
		This is the identity used in
		\cite[full version, Lemma~8.9]{ShaoCai2020}.
		It remains valid in \(K\)-coordinates, since
		\(XbX=\pm b\).  Thus any three \(b\)-vertices can be replaced by one
		\(b\)-vertex and one \(c\RMcore\)-vertex, reducing their number by two.
		Repeating this operation leaves at most two occurrences, with linear
		growth in the instance size and a tracked nonzero scalar.
		This is a tensor identity, so arbitrary complex labels in the remaining
		instance do not affect the reduction.
		
		For \(b=Y\) or \(I\), a self-loop on a \(b\)-vertex has value zero,
		while two such vertices sharing an edge collapse to an \(X\)-edge,
		up to their recorded orientation signs.  If they share both edges, they
		form a component of known scalar value.  We may therefore assume that
		the remaining \(b\)-vertices are disjoint.  Removing \(k\le2\) of them
		leaves a direct gadget with \(2k\) ordered dangling edges.
		
		\smallskip
		\noindent\emph{Adding \(I\).}
        All supported half-charges in \(\cF\) lie in \(4\mathbb Z\).
        Charge conservation across edges therefore makes every binary
        boundary tensor over \(\cF\) anti-diagonal and every quaternary
        boundary tensor EO.
		
		Use occurrence reduction for \(I\).  One disjoint occurrence evaluates
		to zero, since \(XIX=I\) is orthogonal to every anti-diagonal binary for
		the bilinear pairing.  For two disjoint occurrences, let \(H\) be the
		quaternary residual tensor, with the removed vertices on \(12\) and \(34\).
		On its EO support,
		\begin{equation}
		\left.(I_{12}\otimes I_{34})\right|_{HW_4^=}
		=
		\left.
		\frac{X_{13}\otimes X_{24}
			+X_{14}\otimes X_{23}
			-X_{12}\otimes X_{34}}{2}
		\right|_{HW_4^=}.
		\label{eq:q8-equality-eo-closures}
		\end{equation}
		Thus three \(X\) closures, with respective coefficients
		\(1/2,1/2,-1/2\), evaluate the desired contraction.  These are
		\(\KHolant(\cF)\)-oracle queries, valid for arbitrary complex
		coefficients of \(H\).  We have proved
        \begin{equation}
        \KHolant(\cF,I)\leT\KHolant(\cF).
        \label{eq:q8-equality-reduction}
        \end{equation}

		\smallskip
		\noindent\emph{Adding \(Y\).}
		The following identity holds on all sixteen four-bit assignments:
		\begin{equation}
		Y_{12}\otimes Y_{34}
		=X_{14}\otimes X_{23}-X_{13}\otimes X_{24}.
		\label{eq:q8-y-square-closures}
		\end{equation}
		Hence two disjoint \(Y\)-occurrences can be evaluated by the difference
		of two \(X\) closures of the residual gadget.
		The effective kernel of a \(Y\)-vertex attached through two edges
		is \(XYX=-Y\), so the two signs cancel here.
		
		It remains to evaluate one occurrence. Use the bilinear coefficient pairing of
\cref{subsec:reduction-conventions}.  If every binary boundary tensor \(H\) over \(\cF\) satisfies
		\(\langle H,Y\rangle=0\), every such instance has value zero.
		Otherwise fix one direct binary gadget \(G\) over \(\cF\) with
		\[
		\kappa=\langle G,Y\rangle\ne0.
		\]
		Its size and the nonzero algebraic number \(\kappa\) depend only on the
		fixed signature set \(\cF\), not on the input instance. Its truth table is
		computed by finite exact summation.
		For the residual binary \(H\) of an input, apply
		\eqref{eq:q8-y-square-closures} to the disconnected gadget \(H\otimes G\).
		Two oracle closures compute
		\[
		\langle H\otimes G,Y\otimes Y\rangle
		=\kappa\langle H,Y\rangle.
		\]
		Dividing by the known \(\kappa\), and restoring the sign from \(XYX=-Y\),
		gives the required one-occurrence value.
		Together with occurrence reduction, this proves
		\begin{equation}
		\KHolant(\cF,Y)\leT\KHolant(\cF).
		\label{eq:auto-q8-33}
		\end{equation}

        Since \(Y\) has charge zero, the argument for \(I\) also applies
        over \(\cF\cup\{Y\}\), giving
		\begin{equation}
		\KHolant(\cF,Y,I)\leT\KHolant(\cF,Y).
		\label{eq:auto-q8-35}
		\end{equation}
		Finally \(YX I=Z\) supplies the remaining Bell signature by a direct
		gadget, and an \(X\)-vertex merely subdivides an edge.
		This proves \eqref{eq:q8-relative-bell-reduction}.
	\end{proof}

	Under the hypotheses of \cref{thm:q8-normalized-bell-lift}, suppose also
	that \(\neg\Tract(K\cF)\). By \cref{lem:downward-inheritance},
	\(\Tract(K(\cF\cup\{I\}))\) fails. The binary-equality dichotomy
    \cref{thm:eq2-kholant-dichotomy} makes this enlarged problem hard, and
    \eqref{eq:q8-equality-reduction} transfers hardness to \(\KHolant(\cF)\).

\subsubsection{Completion of the height-four case}

	\begin{proof}[Proof of \cref{thm:height-at-least-four} for \(d=4\)]
		Assume \(\neg\Tract(K\cF)\) and \(Q\ne0\), as in
		\cref{subsec:hc-normalization}.
		Apply \cref{lem:hc-port-normalization} and suppose hardness has not
		already been proved.  If \(P_A\cap P_B\ne\varnothing\) for some
		\(A\ne B\) with \(A\cap B\ne\varnothing\), \cref{lem:q8-scope-collision} gives a matching-product central
		restriction, and \cref{lem:hc-transversal-cube-terminal} gives hardness.  Otherwise
		\cref{lem:q8-collision-free-atlas,lem:q8-canonical-support} give a
		bijection \(\iota:V\to\Ftwo^3\) under which the mountain is
		\[
		M^\circ=A\delta_{\zero^8}+B\delta_{\one^8}+cR_8^{\rm cen},
		\qquad ABc\ne0.
		\]
		The first outcome of \cref{lem:q8-canonical-endpoint-repair} is a nonzero EO signature outside \(\mathcal B_\times\), so
		\cref{thm:hc-relative-eo-terminal} gives hardness.
		Otherwise \eqref{eq:q8-canonical-normalization} supplies
		\(sD(s^{-2})=\diag(s,s^{-1})\in GO(X)\) with
		\[
		\bigl(sD(s^{-2})\bigr)^{\otimes8}M^\circ=c\RMcore.
		\]
		Apply \(sD(s^{-2})\) to the entire current signature set and again call it
		\(\cF\). This set contains \(c\RMcore\), the image of the retained
		mountain. The whole-set equivalence preserves support and charge, and
		\cref{lem:GOX-tract-invariance} preserves \(\neg\Tract(K\cF)\).
		Apply \cref{thm:q8-normalized-bell-lift} and its hardness consequence to
		this set. Composing its reduction with the preceding gadget and factor
		reductions and the change-of-basis equivalence gives hardness for the
		original signature set.
	\end{proof}
	
	\section{Synthesis and Proof of the Dichotomy}
	\label{sec:synthesis}
	
	\subsection{Combining the Height Cases}
	\label{sec:mountain-synthesis}
	
	The four height cases apply to the same ambient signature set.  For
heights at least three, the whole-set divisibility hypothesis is supplied
by \cref{lem:charge-bridge}.  We first combine the cases under that
hypothesis and then apply the result to an arbitrary input signature set.
	
	\begin{theorem}[Dichotomy for sets containing a mountain]
		\label{thm:mountain-interface}
		Let \(\cF\) be a finite algebraic set of even-arity signatures containing a
		mountain of height \(d\).  If \(d\ge3\), assume that
		\(\chi_r({\mathbf x})\in d\mathbb Z\) for every \(f\in\cF\) of arity \(r\) and
        every \({\mathbf x}\in\supp(f)\). Then
		\[
		\Tract(K\cF)
		\qquad\text{or}\qquad
		\KHolant(\cF)\text{ is }\SharpP\text{-hard}.
		\]
	\end{theorem}
	
	\begin{proof}
		Apply \cref{cor:height-one-complete} for \(d=1\),
		\cref{cor:height-two-complete} for \(d=2\),
		\cref{thm:height-three-complete} for \(d=3\), and
		\cref{thm:height-at-least-four} for \(d\ge4\).
	\end{proof}
	
	\subsection{Proof of the Main Theorem}
	\label{sec:assembly}

	We complete the main proof in \(K\)-coordinates and then return to
	ordinary coordinates. Before constructing a mountain, we handle the case
	where every signature is supported on one Hamming layer, using the
	following classification.
	
	\begin{theorem}[Single-layer K-Holant
		{\cite[Corollary~68]{MengWangXia2025EO}}]
		\label{thm:external-single-weight}
		Let \(\cF\) be a finite preprocessed algebraic signature set in which every
		signature is supported on one Hamming layer.  Then \(\KHolant(\cF)\) is
		\(\SharpP\)-hard unless either \(\cF\) is common one-sided with
		\(\EOTract(\cF|_{\EO})\), or both positive and negative charges occur and
		\(\EOTract(\cF_{\to\EO})\).  Both exceptions are in \(\FP\) by
		\cref{thm:external-eo}.
	\end{theorem}

	\begin{theorem}[Structural classification in \(K\)-coordinates]
		\label{thm:K-structural}
		Let \(\cF\) be a finite algebraic signature set.  If
		\(\Tract(K\cF)\) holds, then \(\KHolant(\cF)\in\FP\). Otherwise
		\(\KHolant(\cF)\) is \(\SharpP\)-hard.  The predicate \(\Tract(K\cF)\) is decidable.
	\end{theorem}
	
	\begin{proof}
		The tractable implication and its effective recognition follow from
		\cref{thm:predicate-algorithms,eq:K-equivalence}. For the converse, assume that
		\(\KHolant(\cF)\) is not \(\SharpP\)-hard. We prove
		\(\Tract(K\cF)\).
		By the preprocessing in \cref{subsec:zero-nullary}, we may assume that
		\(\cF\) contains neither zero nor nullary signatures.  We use the equivalence
		convention throughout, retaining \(\cF\) for the current set.
		
		If there is an odd-arity signature, apply
		\cref{thm:external-odd,eq:K-equivalence}. The hard case is excluded,
		so \(\Tract(K\cF)\) holds.  We may therefore assume that \(\cF\) has only
		even-arity signatures.  In this remaining case,
		test common one-sided support before factor extraction. For common one-sided support, charge conservation identifies the
		closed-instance values with those of the central restrictions, as in
		\cref{subsec:eo-predicate}. By \cref{thm:external-eo}, failure of
		\(\EOTract(\cF|_{\EO})\) would give hardness. Hence
		alternative~\ref{case:one-sided} holds.  If every member is single-layer, \cref{thm:external-single-weight} excludes hardness and gives
		\(\EOTract(\cF_{\to\EO})\) in the remaining two-sided case.
		This is alternative~\ref{case:single-weight}.

		In the remaining case, use \cref{thm:external-factor,lem:factor-saturation} and the
		tensor-reconstruction equivalence in \cref{subsec:factor-saturation}
		to replace each
		signature by its irreducible factors, still writing \(\cF\).
		Both charge signs persist: a supported charge is the sum of the charges
		on its factor blocks, so a positive or negative sum has a factor of that
		sign.  The set also remains multilayer, since a tensor product of
		single-layer signatures is single-layer.
		An odd-arity factor gives \(\Tract(K\cF)\) by
		\cref{thm:external-odd,eq:K-equivalence}, since hardness is excluded.
		If every remaining factor has arity at most two,
		alternative~\ref{case:arity-two} holds.
		We may therefore start the remaining analysis with nonzero irreducible
		even-arity signatures, at least one of arity greater than two.

		A tractability conclusion for this factor set transfers to the original
		set by reconstruction.  Indeed, the first two alternatives cannot hold
		for this two-sided, multilayer set, and each of the other five
		alternatives is closed under tensor products.
		
		Let \(d\) be the gcd of the nonzero supported half-charges of this
		preprocessed set \(\cF\).
        Apply \cref{lem:charge-bridge} to \(\cF\), which has supported
        half-charges of both signs. In
		its pinning outcome, apply \cref{lem:common-unary-exit} to
		\(\cF\cup\{\DeltaZero,\DeltaOne\}\). Hardness would transfer by
		the stated reduction and is excluded. The tractability predicate therefore
		holds on that set, and \cref{lem:downward-inheritance} gives
		\(\Tract(K\cF)\).  In the remaining outcome,
		the lemma supplies a mountain \(M\) of height \(d\) with
		\(\KHolant(\cF,M)\leT\KHolant(\cF)\) and
		\[
		\chi_{\arity(f)}({\mathbf x})\in d\mathbb Z
		\qquad(f\in\cF, {\mathbf x}\in\supp(f)).
		\]
		The same divisibility holds for \(M\), so \cref{thm:mountain-interface}
		applies to \(\cF\cup\{M\}\).  Its hard outcome would transfer to the
		current problem and is therefore excluded.  By
		\cref{lem:downward-inheritance}, \(\Tract(K\cF)\) holds for the current
		factor set. Tensor reconstruction then gives the predicate for the original set.
	\end{proof}
	
		For the transformability cases, as in \cref{def:common-transformability}, if \(W\) satisfies the transformability
		conditions in \(K\)-coordinates, the ordinary-coordinate basis matrix is \(M=KW\).  Since
		\(K^{\mathsf T}K=X\),
		\begin{equation*}
			M^{-1}(K\cF)=W^{-1}\cF,
			\qquad M^{\mathsf T}M=W^{\mathsf T}XW,
		\end{equation*}
		so one basis handles both the edge and the entire signature set.

	\begin{proof}[Proof of \cref{thm:main}]
		Decidability and polynomial-time evaluation in every tractable case follow
		from \cref{thm:predicate-algorithms}.  For hardness, put
		\(\widehat{\cF}=K^{-1}\cF\).  By \cref{eq:K-equivalence},
		\begin{equation*}
			\Holant(\cF)\equivT\KHolant(\widehat{\cF}).
		\end{equation*}
		Since \(K\widehat{\cF}=\cF\),
\begin{equation*}
\Tract(K\widehat{\cF})=\Tract(\cF).
\end{equation*}
		If \(\neg\Tract(\cF)\), then \(\neg\Tract(K\widehat{\cF})\), so
		\cref{thm:K-structural} implies that \(\KHolant(\widehat{\cF})\) is
		\(\SharpP\)-hard.  The displayed problem equivalence transfers this
        hardness to \(\Holant(\cF)\).
	\end{proof}

	\section{Conclusion}
	\label{sec:conclusion}
	
	We have proved an \(\FP\)-versus-\(\SharpP\)-hard dichotomy for Boolean
	Holant on general graphs with arbitrary finite sets of algebraic
	complex-valued signatures, under polynomial-time Turing reductions.
	The same seven conditions as in the odd-arity dichotomy now apply without
	an odd-arity assumption. This criterion is decidable from exact algebraic input.
	
	For the all-even cases, charge conservation and tensor-factor extraction
	either lead to an existing classification or produce a mountain signature
	of controlled support. Classifying these mountains by height completes the
	remaining cases. The gadget and factor reductions transfer the conclusions
back to the original signature set.
	
	\section*{Code Availability}
	The verification code for the finite calculations in
	Appendix~\ref{sec:exact-finite-certificates} is available in the
	\path{certificates/} directory of the
	\href{https://github.com/liuchliuch/Full-Holant-Dichotomy}%
	{\textcolor{repositorylink}{\GitHubIcon\textit{GitHub repository:
	Full-Holant-Dichotomy}}}. The complete suite is run from the repository
	root with \texttt{python3 certificates/run\_all\_exact.py}.
	The \path{README.md} in that directory gives requirements and reproduction
	commands. \path{theorem-to-script-map.md} identifies the statements checked
	by each program and distinguishes supplementary checks from calculations
	used in the proof.

	\clearpage
\appendix
\renewcommand{\thesection}{F.\Alph{section}}
\renewcommand{\theHsection}{fullappendix.\Alph{section}}
	\section{Exact Finite Certificates}
	\label{sec:exact-finite-certificates}

	\subsection{Exact finite verification}
	\label{subsec:certificate-terminology}
	
The appendix supplies four results used in the main proof:
\begin{enumerate}
\item \Cref{thm:cert-odd-dih-q6-interface} classifies the arity-six
support candidates used in \cref{thm:odd-dih-q6-base}.
\item \Cref{thm:cert-high-charge-q6-interface} gives the central-support
restriction used in \cref{lem:hc-q6-joint-rigidity,thm:hc-q6-complete}.
\item \Cref{thm:cert-q8-collision-free-interface} gives the translation
form of the matching family in \cref{lem:q8-collision-free-atlas}.
\item \Cref{thm:cert-q8-canonical-support-interface} determines the central
coefficients used in \cref{lem:q8-canonical-support}.
\end{enumerate}
The first three use exhaustive finite calculations. The last has a
symbolic proof, with supplementary computational checks recorded at the
end of the appendix. The arity-six coefficient calculation also uses exact
real-polynomial satisfiability queries, whose solver dependence is stated
in \cref{prop:cert-high-charge-q6-joint-deck}. Code Availability gives
the programs, requirements, and reproduction commands.
	
All arithmetic is exact over the indicated fields. We normalize a
projective vector by making its first nonzero coordinate one, unless a
primitive integer representative is specified. For a rational vector,
the latter is obtained by clearing denominators, dividing by the gcd of
the coordinates, and making the first nonzero coordinate positive.

	\subsection{Odd-dihedral, height-three, and height-four certificates}
	\label{proofsec:finite-certificates}
	
	\subsubsection{The order-three odd-dihedral arity-six certificate}
	\label{proofsubsec:cert-odd-dih-q6}
	
	Write \(\mathbb F_3=\mathbb Z/3\mathbb Z\).
	Fix \(\zeta_3\in\AlgNums\) with \(\zeta_3^2+\zeta_3+1=0\), and identify a
	residue \(r\in\mathbb F_3\) with its representative \(0,1,2\) whenever it
	appears as an exponent \(\zeta_3^r\).  For
	\(\epsilon\in\Ftwo\), let
	\[
	\mathcal C_\epsilon
	=
	\bigl\{\{{\mathbf x},\bar {\mathbf x}\}:{\mathbf x}\in\Ftwo^4,\
	\wt({\mathbf x})\equiv\epsilon\pmod2\bigr\}
	\]
	be the four complementary pairs of parity \(\epsilon\).  Choose an orientation
	\(({\mathbf x}_p,\bar {\mathbf x}_p)\), \(0\le p\le3\), of these four pairs. Using the indicator signatures of
\cref{subsec:holant-definitions}, for \(r,t\in\mathbb F_3\) put
	\begin{equation}
		{\mathbf V}_{p,r}=\zeta_3^r \delta_{{\mathbf x}_p}+\zeta_3^{2r}\delta_{\bar {\mathbf x}_p},
		\qquad
		{\mathbf E}_{pq}(r,t)={\mathbf V}_{p,r}-{\mathbf V}_{q,t}.
		\label{eq:cert-odd-dih-q4-atlas-vectors}
	\end{equation}
	Define the finite projective set
	\begin{equation*}
		\mathcal Q_3^\epsilon
		=
		\{[{\mathbf V}_{p,r}]:0\le p\le3,\ r\in\mathbb F_3\}
		\,\cup\,
		\{[{\mathbf E}_{pq}(r,t)]:0\le p<q\le3,\ r,t\in\mathbb F_3\}.
	\end{equation*}
	
For an arity-six tensor \(F\), use the slices \(F_{ij}^{ab}\) from
\cref{subsec:reduction-conventions}, with \(i<j\) and the residual variables
in increasing order. The kernels \(D(-s)X^\xi\), for \(\xi=0,1\), are
respectively \(\mathsf K^+(s)\) and \(\mathsf K^-(s)\) from
\cref{eq:odd-dih-kernels}. For \(s\in\mu_3\),
	\begin{equation}
		\partial_{ij}^{D(-s)X^\xi}F
		=
		F_{ij}^{0,\xi}-sF_{ij}^{1,1-\xi}.
		\label{eq:cert-odd-dih-q6-card}
	\end{equation}
	For \(\eta\in\Ftwo\), if
	\(\supp(F)\subseteq\{{\mathbf x}\in\Ftwo^6:\wt({\mathbf x})\equiv\eta\pmod2\}\), then this
	contraction has residual parity \(\eta\oplus\xi\).
	
	For a support \(S\subseteq\Ftwo^6\), define
	\begin{equation*}
		U_{ij}^{\xi}(S)
		=
		\{{\mathbf x}_{\widehat{\{i,j\}}}:{\mathbf x}\in S,\ x_i\oplus x_j=\xi\}.
	\end{equation*}
	The \((i,j,\xi)\) \emph{xor sector} consists of the assignments satisfying
	\(x_i\oplus x_j=\xi\).
	We will use the condition that \(U_{ij}^{\xi}(S)\) is empty or is the
	union of one or two complete complementary pairs from a single parity class
	\(\mathcal C_\epsilon\).

A secant is a line through two distinct projective points.  In the table below, the middle column lists
	the numbers of complementary pairs supporting the three projective points.
	For example, \(1/1/2\) means that two points each use one complementary pair and the
	third uses two. For three distinct complementary pairs \(p,q,u\), the points
	\([{\mathbf E}_{pq}(r,t)]\), \([{\mathbf E}_{pu}(r,v)]\), and \([{\mathbf E}_{qu}(t,v)]\) are collinear
	because \({\mathbf E}_{pq}(r,t)-{\mathbf E}_{pu}(r,v)+{\mathbf E}_{qu}(t,v)=0\).
	
	We use the pencil terminology of \cref{subsec:odd-dih-support-circuits}.
	Consider labelled pencils \(q_s=A-sB\), \(s\in\mu_3\), modulo one common
	nonzero scalar, such that all three \(q_s\) are nonzero and their three
	distinct projective classes are points of \(\mathcal Q_3^\epsilon\)
	assigned bijectively to the phase labels.  Only pencils satisfying these
	conditions are counted below.
	
	\begin{proposition}[Order-three odd-dihedral arity-six classification]
		\label{prop:cert-odd-dih-q6-order-three}
		The following statements hold.
		
		\begin{enumerate}[label=\textnormal{(\roman*)},leftmargin=*]
			\item Each set \(\mathcal Q_3^\epsilon\) has \(66\) points.  Its projective
			lines containing at least three of its points are exactly
			\begin{center}
				\small
				\begin{tabular}{ccl}
					\toprule
					number & supported pairs & geometry\\
					\midrule
					\(4\)   & \(1/1/1\) & one fixed complementary pair\\
					\(36\)  & \(2/2/2\) & a one-factor Segre ruling\\
					\(54\)  & \(1/1/2\) & a secant between two complementary pairs\\
					\(108\) & \(2/2/2\) & the displayed relation on three complementary pairs\\
					\bottomrule
				\end{tabular}
			\end{center}
			Every such line contains exactly three points of \(\mathcal Q_3^\epsilon\).
			Hence there are \(202\) such lines and \(202\cdot6=1212\) labelled
			pencils satisfying the stated conditions in each parity.
			
			\item Let \(\eta\in\Ftwo\), let
			\(\supp(F)\subseteq\{{\mathbf x}\in\Ftwo^6:\wt({\mathbf x})\equiv\eta\pmod2\}\), and suppose that every contraction
			in \cref{eq:cert-odd-dih-q6-card} is zero or has projective class in the
			appropriate set \(\mathcal Q_3^{\eta\oplus\xi}\).  For no fixed
			\(i<j\) and \(\xi\) does the projectively nonconstant family
			\(s\mapsto \partial_{ij}^{D(-s)X^\xi}F\) lie on a line from the third or fourth row of
			part~\textnormal{(i)}.  Consequently, for every such \(F\), each projectively nonconstant
			family \(s\mapsto \partial_{ij}^{D(-s)X^\xi}F\) either is supported on one complementary pair or
			lies in a one-factor Segre ruling.
			
			\item Exactly \(873\) supports \(S\subseteq\Ftwo^6\) make every
			\(U_{ij}^{\xi}(S)\) empty or a union of one or two complete complementary pairs of a single parity.  They are the empty support, the \(392\)
			distinct cosets,
			\[
			{\mathbf a}+L_\pi,
			\qquad
			L_\pi=\Span_{\Ftwo}\{\one_C:C\in\pi\},
			\]
			where \({\mathbf a}\in\Ftwo^6\) and \(\pi\) is a partition of the six variables into even
			blocks, and the
			\(480\)-element bit-flip/variable-permutation orbit of
			\begin{equation*}
				S_\star=
				\{000000,000011,001100,110011,111100,111111\}.
			\end{equation*}
			The distinct affine supports occur in the respective numbers
			\[
			32,\qquad240,\qquad120
			\]
			for partition types \((6)\), \((4+2)\), and \((2+2+2)\), where a partition
			type is the multiset of its block sizes.  Every nonempty
			support in the classification has one parity.
		\end{enumerate}
	\end{proposition}
	
	\begin{proof}
		The vectors in \cref{eq:cert-odd-dih-q4-atlas-vectors} give twelve
		projective points supported on one complementary pair and fifty-four supported on two.  Normalize every
		nonzero vector at its first nonzero coordinate.  For every one of the
		\(\binom{66}{2}\) point pairs, take the exact reduced-row-echelon form of
		its two-dimensional span and group equal row spaces.  The groups containing
		at least three points of \(\mathcal Q_3^\epsilon\) all have size three and give the table in
		part~\textnormal{(i)}.
		
		The table also has a direct interpretation in terms of complementary pairs.  The three points on each of the four complementary pairs give one line, for four lines in total.  On two fixed complementary pairs, the nine
		points supported on both pairs form a \(3\)-by-\(3\) Segre grid in the transformed
		phase coordinates \(u=r+t\), \(v=r-t\) in \(\mathbb F_3\).
		Fixing \(u\) or \(v\), not the original \(r\) or \(t\) in
		\cref{eq:cert-odd-dih-q4-atlas-vectors}, gives its two rulings and hence
		\(6\cdot2\cdot3=36\) lines.  Every \({\mathbf E}_{pq}(r,t)\) lies on the secant
		through \({\mathbf V}_{p,r}\) and \({\mathbf V}_{q,t}\), giving \(6\cdot3^2=54\) secants.
		For three distinct complementary pairs,
		\[
		{\mathbf E}_{pq}(r,t)-{\mathbf E}_{pu}(r,v)+{\mathbf E}_{qu}(t,v)=0
		\]
		gives \(\binom43 3^3=108\) lines supported on three complementary pairs.  Conversely, in a dependent
		triple every complementary pair occurring in one vector must occur in another.  According
		as the triple uses one, two, or three complementary pairs, it is one of the displayed
		types. Four complementary pairs would leave one occurring only once.
		
		Assigning the phases \(1,\zeta_3,\zeta_3^2\) to the three points on one of these \(202\) lines gives a
		unique affine pencil up to one common nonzero scalar.  Indeed,
		after the unique dependence is scaled, solve
		\[
		q_1=A-B,\qquad q_{\zeta_3}=A-\zeta_3 B
		\]
		and verify \(q_{\zeta_3^2}=A-\zeta_3^2B\).  All six phase assignments are
		distinct.  This proves the last assertion of part~\textnormal{(i)}.
		
		For part~\textnormal{(ii)}, after a bit flip to even parity, fix
		\({\mathbf x}_0=0000\), \({\mathbf x}_1=0011\), \({\mathbf x}_2=0101\), and \({\mathbf x}_3=0110\).  For
		\({\mathbf t}=(t_1,\ldots,t_4)\in\mathbb F_3^4\), define
		\((R_{\mathbf t}q)({\mathbf x})=\zeta_3^{\sum_{j=1}^4t_jx_j}q({\mathbf x})\).  Under \(R_{\mathbf t}\),
		the four oriented complementary-pair labels change by \(H{\mathbf t}\),
		where
		\[
		H=
		\begin{pmatrix}
			1&1&1&1\\
			1&1&-1&-1\\
			1&-1&1&-1\\
			1&-1&-1&1
		\end{pmatrix},
		\qquad
		\det H=-16\ne0\pmod3.
		\]
		The phase permutations are generated by \(s\mapsto\alpha s\) and
		\(s\mapsto\alpha/s\), with
		\(\alpha\in\mu_3\). These maps form \(D_6\cong S_3\) and realize all six
		permutations of \(\mu_3\).  Together with variable permutations and bit flips,
		these actions reduce the secant cases supported on two complementary pairs and the cases supported on three to the following
		representative pencils for a fixed restriction \(x_i\oplus x_j=\xi\).
        Group coordinates by their four oriented complementary pairs,
        writing \(\mathbf0_2=(0,0)\):
        \begin{align*}
        A_{\rm sec}&=\bigl((1,1),(\zeta_3,\zeta_3),\mathbf0_2,\mathbf0_2\bigr),\\
        B_{\rm sec}&=\bigl((\zeta_3^2,\zeta_3^2),(\zeta_3,\zeta_3),\mathbf0_2,\mathbf0_2\bigr),\\
        A_{\rm tri}&=\bigl((1,1),(\zeta_3,\zeta_3),(\zeta_3^2,\zeta_3^2),\mathbf0_2\bigr),\\
        B_{\rm tri}&=\bigl((\zeta_3,\zeta_3),(1,1),(\zeta_3^2,\zeta_3^2),\mathbf0_2\bigr).
        \end{align*}
		Each displayed pair \((A_*,B_*)\) denotes the pencil
		\(q_s=A_*-sB_*\), \(s\in\mu_3\).
		Since \(H\) is invertible over \(\mathbb F_3\), the maps \(R_{\mathbf t}\) shift the
		complementary-pair labels independently.  Variable permutations and bit flips act transitively
		on the complementary-pair choices and the two total parities.  Together with the phase permutations, this proves that every orbit of
		the secant cases supported on two complementary pairs and the cases supported on three under these operations is
		represented.
		
		Fix the selected pair \(i<j\) and sector \(\xi\).  It remains to exclude
		arity-six tensors \(F\) satisfying the hypotheses of part~\textnormal{(ii)}
		and realizing either selected representative pencil.  Put
		\(\zeta=\xi\oplus1\).  Under the restriction \(x_i\oplus x_j=\zeta\),
		a projectively nonconstant pencil is one
		of the \(1212\) labelled pencils from
		part~\textnormal{(i)}.
		Its slices \(C=F_{ij}^{0,\zeta}\) and \(D=F_{ij}^{1,1-\zeta}\) are fixed
		up to a common scalar \(\lambda\in\AlgNums^\times\).  For each of the
		fourteen variable pairs other than the selected pair, and for each xor sector and
		phase, substitution into \cref{eq:cert-odd-dih-q6-card} gives a quaternary
		contraction \({\mathbf q}_0+\lambda {\mathbf q}_1\), where the computed vectors
		\({\mathbf q}_0,{\mathbf q}_1\in\AlgNums^{\Ftwo^4}\) are independent of \(\lambda\).  For a
		fixed representative \({\mathbf v}\in\AlgNums^{\Ftwo^4}\setminus\{0\}\) of a point
		of the appropriate set \(\mathcal Q_3^\epsilon\), the set of \(\lambda\)
		for which this vector belongs to
		\(\Span\{{\mathbf v}\}\) is all of the field, empty, or one
		point. Exact two-column row reduction determines which case occurs.
		Intersecting the resulting finite unions over all remaining pairs, both xor
		sectors, and all three phases leaves no nonzero \(\lambda\) for either
		representative.
		
		For an opposite pencil with constant projective class, write \(C=c{\mathbf v}\) and
		\(D=d{\mathbf v}\), where \({\mathbf v}\in\AlgNums^{\Ftwo^4}\setminus\{0\}\) represents one
		of the \(66\) points of the appropriate \(\mathcal Q_3^\epsilon\) and \((c,d)\in\AlgNums^2\setminus\{(0,0)\}\).
		For each remaining contraction, substitution gives
		\({\mathbf q}_0+c{\mathbf q}_1+d{\mathbf q}_2\), where the computed
		\({\mathbf q}_0,{\mathbf q}_1,{\mathbf q}_2\in\AlgNums^{\Ftwo^4}\) are independent of \(c,d\).  Its
		preimage over all target points of \(\mathcal Q_3^\epsilon\) is a finite union of affine linear
		systems in \((c,d)\).  Exact row reduction,
		followed by intersection over all contractions in the hypothesis of
		part~\textnormal{(ii)}, leaves no
		nonzero pair \((c,d)\).  The identically zero opposite sector is checked
		separately and also fails.  All matrices have entries in
		\(\mathbb Q(\zeta_3)\). Their ranks and consistency are unchanged after
		extension to \(\AlgNums\).  Swapping the two variables when \(\xi=0\) changes nothing, whereas for
		\(\xi=1\) reversal at \(s\) gives a nonzero scalar multiple of the contraction
		at \(s^{-1}\), so unordered variable pairs exhaust the ordered conditions.
		
		For part~\textnormal{(iii)}, fix one variable pair and one xor sector.  A union satisfying the stated condition is empty, one of the eight unions consisting of one complementary pair across
		the two parities, or one of the twelve unions consisting of two complementary pairs.  Each residual string
		in a nonempty union occurs only in \(F_{ij}^{0,\xi}\), only in
		\(F_{ij}^{1,1-\xi}\), or in both.  Hence the exact finite superset of local candidates has
		\[
		1+8\cdot3^2+12\cdot3^4=1045
		\]
		members.  Combine the two disjoint xor sectors on the first variable pair,
		giving \(1045^2\) candidates.  Retain a candidate precisely when, on
		each of the other fourteen pairs, each of the two unions is empty or is
		the union of one or two complete complementary pairs of its own parity.  Encode a support by a
		\(64\)-bit integer whose bit \(k\) corresponds to the \(k\)-th string in
		lexicographic order on \(\{0,1\}^6\), starting at \(k=0\).  Exact
		intersection of these encoded support lists leaves \(873\) distinct supports.
		
		Independently generate every partition into even blocks and every
		coset of its block-indicator space.  This gives
		\[
		32+240+120=392
		\]
		distinct supports.  Generating all bit flips and variable permutations of
		\(S_\star\) gives \(480\) further supports.  These two families and the
		empty support are disjoint, and their union is exactly the \(873\)-support
		intersection.  Since every block indicator has even weight and
		\(S_\star\) has even parity, every nonempty support has one parity.
		
        The programs regenerate the projective point sets, lines, labelled
        pencils, affine preimages, support candidates, and orbits from these
        definitions, using exact linear algebra over \(\mathbb Q(\zeta_3)\)
        and exhaustive enumeration.
	\end{proof}

	\Needspace{7\baselineskip}
Part~\textnormal{(i)} below obtains the local support condition for
\(N=3\). Part~\textnormal{(ii)} depends only on that support condition,
so \cref{thm:odd-dih-q6-base} also uses it for \(N\ge5\) after its counting
argument.
    \begin{theorem}[Arity-six support classification]
        \label{thm:cert-odd-dih-q6-interface}
		\begin{enumerate}[label=\textnormal{(\roman*)},leftmargin=*]
			\item Let \(F\) be an arity-six tensor of one parity, and suppose that every
			contraction \(\partial_{ij}^{D(-s)X^\xi}F\), for every variable pair, xor, and
			\(s\in\mu_3\), belongs to \(\mathcal B_3([4])\).
			Then, for every fixed pair and xor, with
			\(A=F_{ij}^{0,\xi}\), \(B=F_{ij}^{1,1-\xi}\),
			\(\supp(A)\cup\supp(B)\) is empty or is the union of at most two
			complementary pairs in its residual parity.
			\item Let \(F\) be an arity-six tensor of one parity.  If the support union
			in part~\textnormal{(i)} is empty or consists of at most two
			complementary pairs for every pair and xor, then \(\supp(F)\) is empty,
			a partition coset of type \((6)\), \((4+2)\), or \((2+2+2)\), or a
			bit-flip/variable-permutation image of \(S_\star\) in
			\cref{eq:odd-dih-q6-nonproduct-support}.
		\end{enumerate}
	\end{theorem}
	
	\begin{proof}
        For part~\textnormal{(i)}, a nonzero member of
        \(\mathcal B_3([4])\) has the two-pair coefficient form
        \eqref{eq:odd-dih-q4-two-block} with \(N=3\). Rescaling this form
        gives one of the vectors \({\mathbf E}_{pq}(r,t)\) in
        \eqref{eq:cert-odd-dih-q4-atlas-vectors}. Thus every nonzero
        contraction has projective class in the corresponding
        \(\mathcal Q_3^\epsilon\). By
		\cref{prop:cert-odd-dih-q6-order-three}\textnormal{(ii)}, a projectively nonconstant
		pencil lies on a line supported on one complementary pair or a one-factor Segre ruling.  Its support
		union consequently contains at most two complete complementary pairs.
		Coordinatewise,
		\[
		\supp(A)\cup\supp(B)=\bigcup_{s\in\mu_3}\supp(A-sB),
		\]
		since a nonzero affine coefficient has at most one root.  This also handles
		zero and projectively constant pencils.  For part~\textnormal{(ii)}, the
		hypothesis says precisely that every local union
		\(U_{ij}^{\xi}(\supp(F))\) is empty or consists of one or two complete
		complementary pairs of a single parity, as required in the enumeration of
		\cref{prop:cert-odd-dih-q6-order-three}\textnormal{(iii)}.  That enumeration
        gives the asserted support alternatives.
	\end{proof}
	
    \subsubsection{Height-three loop and two-copy arrays}
    \label{proofsubsec:cert-high-charge-q6}

The arrays below list the EO coefficients of the contractions in
\cref{def:hc-q6-joint-deck}. For a mountain with endpoint coefficients
\(a,b\), the parameter is \(t=ab\), and \(q_S\) are the coefficients of
its central restriction. We retain the external variable orders from
that definition. Let
	\[
	\mathcal W_4=(0011,0101,0110,1001,1010,1100).
	\]
	The arrays below use this order.  
	Let \(U=[6]\), let \(t\in\mathbb C^\times\), and assign arbitrary complex
	coefficients \(q_S\) to the twenty triples \(S\in\binom U3\), and put
\(Q=\sum_{S\in\binom U3}q_S\delta_{\one_S}\).
For each deleted pair \(i<j\), order the residual variables as
	\(v_0<v_1<v_2<v_3\). For \({\mathbf w}\in\mathcal W_4\), write
	\({\mathbf w}=(w_0,w_1,w_2,w_3)\) and put
	\(R_{\mathbf w}=\{v_r:w_r=1\}\), and define the six EO entries
	\begin{equation}
		(\partial_{ij}Q)({\mathbf w})=q_{R_{\mathbf w}\cup\{i\}}+q_{R_{\mathbf w}\cup\{j\}},
		\qquad {\mathbf w}\in\mathcal W_4.
		\label{eq:cert-q6-loop-array}
	\end{equation}
	For unordered two-element external sets \(A,B\in\binom U2\), use the
	canonical increasing orders \(A=(a_0,a_1)\), \(B=(b_0,b_1)\) and the
external order \((a_0,a_1,b_0,b_1)\). The two copies are distinguished.
Put \(I_A=U\setminus A\),
	\(I_B=U\setminus B\), and let
	\(\pi:I_A\to I_B\) be a bijection.  Define the six EO entries, indexed by
	the concatenated strings \(xy\in\mathcal W_4\), where
	\({\mathbf x}=(x_0,x_1)\) and \({\mathbf y}=(y_0,y_1)\), by
	\begin{equation}
		C_{A,B,\pi}({\mathbf x},{\mathbf y})
		=t\,\one[({\mathbf x},{\mathbf y})\in\{(00,11),(11,00)\}]
		+\sum_{\substack{T\subseteq I_A\\|T|=3-\wt({\mathbf x})}}
		q_{A_{\mathbf x}\cup T}\,
		q_{B_{\mathbf y}\cup(I_B\setminus\pi(T))},
		\label{eq:cert-q6-two-copy-array}
	\end{equation}
	where \(A_{\mathbf x}=\{a_r:x_r=1\}\), and similarly for \(B_{\mathbf y}\).
	
	For any quaternary EO array
	\({\mathbf z}=(z_0,z_1,z_2,z_3,z_4,z_5)\) indexed in the order \(\mathcal W_4\),
	put \(P_0=z_0z_5\), \(P_1=z_1z_4\), and \(P_2=z_2z_3\).
	
	\begin{proposition}[Height-three coefficient classification]
		\label{prop:cert-high-charge-q6-joint-deck}
		Suppose every one of the fifteen arrays in
		\cref{eq:cert-q6-loop-array} and all
		\(15^2\cdot4!=5400\) labelled arrays in
		\cref{eq:cert-q6-two-copy-array} satisfy
		\begin{align}
			z_0=0&\Longleftrightarrow z_5=0,&
			z_1=0&\Longleftrightarrow z_4=0,&
			z_2=0&\Longleftrightarrow z_3=0,
			\label{eq:cert-q6-zero-pairs}\\
			P_0P_1P_2&=0,&
			P_iP_j(P_i^2-P_j^2)&=0\quad(0\le i<j\le2).
			\label{eq:cert-q6-pair-products}
		\end{align}
		Then
		\[
		\{R\in\tbinom U3:q_R\ne0\}\subseteq\{S,U\setminus S\}
		\qquad\text{for some }S\in\binom U3.
		\]
	\end{proposition}
	
	\begin{proof}

For each of the fifteen loop arrays,
		\cref{eq:cert-q6-zero-pairs,eq:cert-q6-pair-products} force at least one
		of its three complementary EO-entry pairs to vanish.  Selecting such a
		pair gives, for a suitable \(S\in\binom U3\) containing exactly one of
		\(i,j\), two signed equations
		\begin{equation*}
			q_S=-q_{S\mathbin\triangle\{i,j\}},
			\qquad
			q_{U\setminus S}=-q_{U\setminus(S\mathbin\triangle\{i,j\})}.
		\end{equation*}
		Thus every coefficient vector lies in at least one of exactly
		\(3^{15}=14\,348\,907\) branches defined by signed linear equations.  Signed
		union--find deduplicates them to \(1868\) linear subspaces. Modulo the \(S_6\) variable action
		and simultaneous complementation, they form exactly \(24\) orbits.
		
        For each family, substitute independent complex parameters into
        \cref{eq:cert-q6-loop-array,eq:cert-q6-two-copy-array}. Impose
        \cref{eq:cert-q6-zero-pairs,eq:cert-q6-pair-products}, \(t\ne0\),
        and the counterexample condition that two distinct complementary
        coefficient pairs \(\{q_S,q_{U\setminus S}\}\) are nonzero.
        Split each complex parameter into real and imaginary parts.
        A selected nonzero parameter may be normalized to \(1\), since
        the equations are homogeneous under
        \((q,t)\mapsto(cq,c^2t)\) for \(c\in\mathbb C^\times\).
        Exact satisfiability queries find no solution to these systems.
		
		All generated expressions have integer coefficients, and all reductions
		performed before the solver use exact rational polynomial arithmetic.  The
		all-zero orbit is immediate.
		The exact real-polynomial queries for the remaining twenty-three nonzero
		orbit families were evaluated with Z3~4.14.0.  The answers are
		\texttt{UNSAT}, meaning that no satisfying assignment exists. The
		calculation uses 61 queries. For nineteen of the nonzero orbit families,
		51 queries cover the cases obtained by choosing a nonzero parameter and
		normalizing it to \(1\). The script checks that parameter sets omitted by
		these normalizations are supported on at most one complementary pair and
		hence cannot meet the counterexample condition. For the other four families,
		10 queries verify reductions to one of those nineteen families, the zero
		vector, or one complementary pair.
		
        The wrapper requires Z3~4.14.0 and records compiler and solver
        versions and source hashes. It regenerates the orbit representatives
        and queries on each run. The option \texttt{--q6-audit-dir} saves
        inputs, parameter encodings, branch labels, and solver responses.

		The runs neither request nor check solver proof objects.
		The \texttt{UNSAT} answers therefore rely on the correctness of Z3.
        With this solver dependence, the exhaustive calculation proves the
        proposition for arbitrary complex coefficients.
	\end{proof}

	\begin{theorem}[Consequence for loop and two-copy arrays]
		\label{thm:cert-high-charge-q6-interface}
		Let \(U=[6]\), let \(t\in\mathbb C^\times\), and let
		\((q_S)_{S\in\binom U3}\) be arbitrary complex coefficients.  Form the
		fifteen loop arrays in \cref{eq:cert-q6-loop-array} and all \(5400\)
		labelled two-copy arrays in \cref{eq:cert-q6-two-copy-array}.  If every one
		of these arrays belongs to \(\mathcal B_\times^0([4])\), then
		\[
		\{R\in\tbinom U3:q_R\ne0\}
		\subseteq\{S,U\setminus S\}
		\qquad\text{for some }S\in\binom U3.
		\]
	\end{theorem}
	
	\begin{proof}
		Every member of \(\mathcal B_\times^0([4])\) satisfies
		\cref{eq:cert-q6-zero-pairs,eq:cert-q6-pair-products}. Hence
		\cref{prop:cert-high-charge-q6-joint-deck} applies and yields the assertion.
	\end{proof}
	
	\begin{remark}[Why the two-copy arity-six arrays are necessary]
		\label{rem:cert-high-charge-q6-one-loop-insufficient}
		Order each triple as \(S=\{s_0<s_1<s_2\}\), and order the twenty triples by
        the lexicographic order of \((s_0,s_1,s_2)\). Consider the coefficient
        vector whose twenty weight-three entries are all nonzero:
		\[
		(1,-1,-1,1,i,-i,-1,1,-i,-i,-i,-i,-1,1,-i,i,-1,1,1,-1)
		\]
		This vector does not define a member of \(\mathcal B_\times([6])\), although all fifteen
		arrays \cref{eq:cert-q6-loop-array} lie in
        \(\mathcal B_\times^0([4])\). Take \(A=B=\{1,2\}\) and
        \(\pi(3)=3\), \(\pi(4)=4\), \(\pi(5)=6\), \(\pi(6)=5\).
        In the order \(\mathcal W_4\), \cref{eq:cert-q6-two-copy-array} is
		\[
		(t,2,-2,-2,2,t),
		\]
		where \(t\ne0\). Its six nonzero entries put it outside that class.  Thus
		the one-loop conditions alone cannot imply
		the conclusion of \cref{thm:cert-high-charge-q6-interface}.
	\end{remark}
	
	\subsubsection{Arity-eight matching families}
	\label{proofsubsec:cert-q8-matching-atlas}
	
	Let \(V=\{0,1,\ldots,7\}\), and for distinct labels \(a,b\) write
	\(ab\) for the unordered pair \(\{a,b\}\).  For every pair
	\(e\in\binom V2\), let
	\(P_e\) be a perfect matching of \(V\setminus e\).  For disjoint
	\(e,r\in\binom V2\), write \(r=\{u,v\}\) and define \(\kappa_r(P_e)\) as follows.  If
	\(r\in P_e\), delete that edge.  Otherwise let \(u',v'\) be the mates of
	\(u,v\), remove \(uu'\) and \(vv'\), add \(u'v'\), and retain the unique
	untouched edge.  In both cases \(\kappa_r(P_e)\) is a two-edge matching on
	\(V\setminus(e\cup r)\).
	
	\begin{proposition}[Arity-eight matching classification]
		\label{prop:cert-q8-collision-free-atlas}
		Suppose
		\begin{equation}
			\kappa_r(P_e)=\kappa_e(P_r)
			\qquad(e\cap r=\varnothing)
			\label{eq:cert-q8-commuting-cards}
		\end{equation}
		for every pair of disjoint deleted pairs.  After relabeling so that
		\(P_{01}=\{23,45,67\}\), exactly nine labelled families satisfy these
		constraints.  Seven are induced by a global perfect matching \(M\) on eight
		variables, meaning \(P_e=\kappa_e(M)\) under the same deletion-and-splicing
		rule.  Each of these seven families has \(P_A\cap P_B\ne\varnothing\)
		for some distinct pairs \(A,B\) with \(A\cap B\ne\varnothing\).
		With labels \(0,\ldots,7\) interpreted as three-bit vectors, one remaining
		family is given by the matchings below and the other is obtained by swapping
		\(6,7\):
		\begin{equation}
			\begin{array}{cc}
				\toprule
				h&\{\{t,t+h\}:t\in\mathbb F_2^3\}\\ \midrule
				001&01,23,45,67\\
				010&02,13,46,57\\
				011&03,12,47,56\\
				100&04,15,26,37\\
				101&05,14,27,36\\
				110&06,17,24,35\\
				111&07,16,25,34.\\
				\bottomrule
			\end{array}
			\label{eq:cert-q8-translation-atlas}
		\end{equation}
		For a deleted pair \(p,q\), removing the edge \(pq\) from the row
		\(h=p+q\) gives \(P_{pq}\).
	\end{proposition}
	
	\begin{proof}
		Once \(P_{01}\) is fixed, each of the other twenty-seven variable pairs has the
		fifteen perfect matchings on its six residual variables as candidates.
		An exhaustive backtracking search imposes
        \cref{eq:cert-q8-commuting-cards} whenever both sides become known.
        Exactly nine complete matching families remain. They give the stated
        \(7+2\) split and the table in \cref{eq:cert-q8-translation-atlas}.
		
        The programs reproducing these nine families are described in
        \cref{subsec:certificate-supplementary}.
	\end{proof}
	
	\begin{theorem}[Translation form of the matching family]
		\label{thm:cert-q8-collision-free-interface}
		Let \(V\) be an eight-element variable set and, for every
		\(e\in\binom V2\), let \(P_e\) be a perfect matching of \(V\setminus e\).
		Suppose
		\[
		\kappa_r(P_e)=\kappa_e(P_r)
		\qquad(e\cap r=\varnothing).
		\]
		If \(P_A\cap P_B=\varnothing\) whenever \(A\ne B\) and
		\(A\cap B\ne\varnothing\), then there is a bijection
		\(\iota:V\to\mathbb F_2^3\) such that, after relabelling, for \(p\ne q\)
		and \(h=p+q\),
		\[
		P_{pq}
		=
		\bigl\{\{t,t+h\}:t\in\mathbb F_2^3\setminus\{p,q\}\bigr\},
		\]
		where duplicate unordered pairs are suppressed.
	\end{theorem}
	
	\begin{proof}
		Relabel \(V\) so that \(P_{01}=\{23,45,67\}\).  By
		\cref{prop:cert-q8-collision-free-atlas}, seven of the nine labelled
		families violate the condition on \(P_A\cap P_B\).  The remaining two
		differ only by a relabelling and give the displayed matching family.
	\end{proof}
	
	\begin{proposition}[Coefficients on affine planes]
		\label{prop:cert-q8-canonical-support}
		Identify \(V=\mathbb F_2^3\), with the labels \(0,\ldots,7\) interpreted as
		their three-bit binary expansions, and for
		\(S\in\binom V4\) put
		\(\sigma(S)=\bigoplus_{s\in S}s\).  Let \(c_S\in\mathbb C\), and let
		\(\lambda_{pq}\in\mathbb C^\times\) for every unordered pair
		\(\{p,q\}\in\binom V2\).  Suppose that, whenever \(p\in S\), \(q\notin S\),
		\(h=p+q\), and \(S'=S\mathbin\triangle\{p,q\}\),
		\begin{equation}
			c_S+c_{S'}=
			\begin{cases}
				\lambda_{pq},&\sigma(S)\in\{0,h\},\\
				0,&\sigma(S)\notin\{0,h\}.
			\end{cases}
			\label{eq:cert-q8-canonical-swap-system}
		\end{equation}
		Then the four-sets with nonzero coefficients are exactly the fourteen
		affine planes of \(V\), and there is one \(c\in\mathbb C^\times\) such that
		\begin{equation}
			c_S=
			\begin{cases}
				c,&\sigma(S)=0,\\
				0,&\sigma(S)\ne0.
			\end{cases}
			\label{eq:cert-q8-canonical-coefficients}
		\end{equation}
		Call a four-subset a \emph{nonplane} when it is not an affine plane.
		Let \(G\) be the graph on the fifty-six nonplanes, joining \(S\) to
		\(S\mathbin\triangle\{p,q\}\) whenever \(p\in S\), \(q\notin S\), and
		the latter four-set is also a nonplane.  Let \(H\) be the bipartite
		incidence graph between planes and unordered variable pairs, with an edge
		between \(S\) and \(\{p,q\}\) exactly when
		\(\lvert S\cap\{p,q\}\rvert=1\).  Then \(G\) is connected and contains
		an odd cycle, and \(H\) is connected.
	\end{proposition}
	
	\begin{proof}
		A four-set has xor zero exactly when it is an affine plane: xor zero makes
		the fourth point the affine sum of the other three, and the seven nonzero
		linear functionals on \(\mathbb F_2^3\) each have two affine level-set
		planes.  Thus there are fourteen such sets and fifty-six nonplanes.
		For a nonplane \(S\), the exchanged set \(S'=S\mathbin\triangle\{p,q\}\)
		is again a nonplane exactly when \(\sigma(S)\notin\{0,p+q\}\).
		Lift \(t\in V\) to
		\(\widehat{\mathbf t}=(1,t)\in\mathbb F_2^4\), and let \(M\) be the represented
		matroid on ground set \(V\) with representing vectors \(\widehat{\mathbf t}\). Independence in this matroid means linear independence of those vectors. A four-set is a nonplane exactly
		when its four lifted points form a basis.  Thus \(G\) is the
		basis-exchange graph of the rank-four binary matroid \(M\), whose
		vertices are bases and whose edges exchange one element while preserving
		a basis. By basis exchange, a basis different from a target basis admits an
        exchange that increases their intersection. Repeating such exchanges
        reaches the target basis, so this graph is connected.
		
		In binary labelling it contains the triangle below, where, for example,
		\(0124\) abbreviates the four-set \(\{0,1,2,4\}\):
		\begin{equation}
			\mathtt{0124}\;\text{---}\;\mathtt{0125}\;\text{---}\;
			\mathtt{0245}\;\text{---}\;\mathtt{0124}.
			\label{eq:cert-q8-bad-swap-cycle}
		\end{equation}
		The three xor sums are \(7,6,3\), while the corresponding swap directions
		are \(1,5,4\). Hence all three edges belong to \(G\).  The second
		line of \cref{eq:cert-q8-canonical-swap-system} gives, on every edge of \(G\),
		\begin{equation*}
			c_{S'}=-c_S.
		\end{equation*}
		Alternating around
		\cref{eq:cert-q8-bad-swap-cycle} and using characteristic zero forces the
		three coefficients on the triangle to vanish.  Connectivity then forces
		\(c_S=0\) for every nonplane.
		
		Let \(S\) now be a plane and let \(p\in S,q\notin S\).  Then
		\(S'=S\mathbin\triangle\{p,q\}\) has xor \(p+q\ne0\), so it is a nonplane.
		The first line of \cref{eq:cert-q8-canonical-swap-system} therefore gives
		\begin{equation}
			c_S=\lambda_{pq}\ne0.
			\label{eq:cert-q8-plane-pair-value}
		\end{equation}
		It remains only to show that these values agree.  Regard a plane as
		incident with every unordered pair having one endpoint in the plane and
		one outside it.  Two nonparallel affine planes meet in two points and their
		union has six points. A pair from their intersection to the complement of
		their union is incident with both planes in \(H\).  Two distinct
		parallel planes are complementary and hence have the same incident variable
		pairs.  Thus \(H\) is connected.  Every pair \(\{p,q\}\) is incident
		with a plane:
		choose a linear functional \(\ell\) with \(\ell(p+q)=1\) and take
		\(\{x:\ell(x)=\ell(p)\}\).
		\Cref{eq:cert-q8-plane-pair-value} makes all fourteen plane coefficients
		and all incident \(\lambda_{pq}\)'s equal to one common nonzero \(c\),
		proving \cref{eq:cert-q8-canonical-coefficients}.
	\end{proof}

    \begin{theorem}[Affine-plane coefficients]
		\label{thm:cert-q8-canonical-support-interface}
		Identify \(V=\mathbb F_2^3\).  Let
		\(c_S\in\mathbb C\) for \(S\in\binom V4\), and let
		\(\lambda_{pq}\in\mathbb C^\times\) for every unordered pair
		\(\{p,q\}\in\binom V2\).  Suppose that, whenever \(p\in S\), \(q\notin S\),
		\(h=p+q\), and \(S'=S\mathbin\triangle\{p,q\}\),
		\[
		c_S+c_{S'}=
		\begin{cases}
			\lambda_{pq},&\displaystyle\bigoplus_{s\in S}s\in\{0,h\},\\
			0,&\displaystyle\bigoplus_{s\in S}s\notin\{0,h\}.
		\end{cases}
		\]
		Then there is one \(c\in\mathbb C^\times\) such that
		\[
		c_S=
		\begin{cases}
			c,&\displaystyle\bigoplus_{s\in S}s=0,\\
			0,&\displaystyle\bigoplus_{s\in S}s\ne0.
		\end{cases}
		\]
		Equivalently, the support consists exactly of the fourteen affine planes of
		\(\mathbb F_2^3\), and all supported coefficients are equal.
	\end{theorem}
	
	\begin{proof}
		This follows from
		\cref{prop:cert-q8-canonical-support}.
	\end{proof}

\subsection{Implementation details and supplementary checks}
\label{subsec:certificate-supplementary}

For the arity-six coefficient calculation, the Python loop generator lists
the six entries in reverse order. Both
\eqref{eq:cert-q6-zero-pairs} and \eqref{eq:cert-q6-pair-products} are
invariant under \(z_k\mapsto z_{5-k}\). The two-copy generator uses
\(\mathcal W_4\) directly.
For the C++ enumeration, variables are numbered from \(0\) to \(5\), and
triples are ordered by increasing \(\sum_{k\in S}2^k\). In a nonzero
component \(C\) of the signed equations, write
\(q_i=\varepsilon_i u_C\), where \(\varepsilon_i\in\{1,-1\}\), choose
\(\varepsilon_{\min C}=1\), and name its parameter by \(\min C\).
Choose the lexicographically least encoding under variable permutations
and simultaneous complementation as the orbit representative. Before any
solver query, the Python stage reindexes the C++ output and verifies that
both stages give the same representatives.

For the arity-eight matching classification, the scripts
\texttt{verify\_high\_charge\_q8\_atlas.py} and
\texttt{verify\_high\_charge\_q8\_normalized.py} regenerate the nine
families. The search has \(236\) nodes. The programs share the
matching-generation logic and are not independent implementations.
The normalized script also checks the linear systems for the original
tensor's seventy weight-four coefficients and twenty-eight contraction
scales. For each matching family, requiring every pair contraction to
equal its scalar times the corresponding \(X\)-matching product gives
a \(560\)-by-\(98\) integer matrix, with twenty equations for each of
the twenty-eight deleted pairs. A nonzero rational kernel vector gives
rank at most \(97\), and rank \(97\) modulo
\(1{,}000{,}000{,}007\) gives rank at least \(97\) over \(\mathbb Q\).
This rank calculation is supplementary to the matching classification.

For \cref{prop:cert-q8-canonical-support}, the atlas script checks the
\(14+56\) split, the connectivity of the nonplane exchange graph \(G\),
its displayed odd cycle, and the connectivity of the plane--pair incidence
graph \(H\). These checks supplement the symbolic proof above.
The normalized script separately encodes the multiplicative rank-one
relations among the fourteen plane coefficients by an integer matrix,
whose rows record the exponents in the relations. An explicit kernel and
a nonzero \(13\)-by-\(13\) minor of absolute value \(16\) give rank
\(13\). Rank \(9\) modulo two then forces the nonzero Smith invariant
factors to be nine \(1\)'s and four \(2\)'s. These supplementary checks
use integer rank and divisibility arguments and are not needed for the
symbolic deduction of the affine-plane coefficients.

\end{document}